\pdfoutput=1
\documentclass[a4paper,UKenglish]{lipics-v2019}
\usepackage{microtype}%if unwanted, comment out or use option "draft"
\usepackage{cite}

\title{On the Complexity of Reachability\protect\\in Parametric Markov Decision
Processes}

\titlerunning{On the Complexity of Reachability in Parametric Markov Decision
Processes}%optional, please use if title is longer than one line

\author{Tobias Winkler}{RWTH Aachen University}{tobias.winkler1@rwth-aachen.de}{}{}%mandatory, please use full name; only 1 author per \author macro; first two parameters are mandatory, other parameters can be empty.

\author{Sebastian Junges}{RWTH Aachen University}{sebastian.junges@cs.rwth-aachen.de}{0000-0003-0978-8466}{}%mandatory, please use full name; only 1 author per \author macro; first two parameters are mandatory, other parameters can be empty.

\author{Guillermo A. P\'erez}{University of
  Antwerp}{guillermoalberto.perez@uantwerpen.be}{0000-0002-1200-4952}{}
\author{Joost-Pieter Katoen}{RWTH Aachen University}{katoen@cs.rwth-aachen.de}{0000-0002-6143-1926}{}%mandatory, please use full name; only 1 author per \author macro; first two parameters are mandatory, other parameters can be empty.

\authorrunning{T. Winkler, S. Junges, G. A. P\'erez, J.-P. Katoen}%mandatory. First: Use abbreviated first/middle names. Second (only in severe cases): Use first author plus 'et al.'

\Copyright{Tobias Winkler, Sebastian Junges, Guillermo A. P\'erez, Joost-Pieter
Katoen}%mandatory, please use full first names. LIPIcs license is "CC-BY";  http://creativecommons.org/licenses/by/3.0/

%mandatory

\acknowledgements{We would like to thank Krishnendu Chatterjee for his pointer to CSRGs.}%optional

\EventEditors{John Q. Open and Joan R. Access}
\EventNoEds{2}
\EventLongTitle{42nd Conference on Very Important Topics (CVIT 2016)}
\EventShortTitle{CVIT 2016}
\EventAcronym{CVIT}
\EventYear{2016}
\EventDate{December 24--27, 2016}
\EventLocation{Little Whinging, United Kingdom}
\EventLogo{}
\SeriesVolume{42}
\ArticleNo{23}
\nolinenumbers %uncomment to disable line numbering
\hideLIPIcs  %uncomment to remove references to LIPIcs series (logo, DOI, ...), e.g. when preparing a pre-final version to be uploaded to arXiv or another public repository
\usepackage{nicefrac}
\usepackage{wasysym}
\usepackage{tikz}
\usepackage{upgreek}
\usepackage{mathtools}

\usepackage{wrapfig}

\usepackage{todonotes}

\usetikzlibrary{arrows,positioning,automata,snakes}

\tikzset{>=stealth'}

\usepackage{colonequals}
\newcommand{\QQ}{\mathbb{Q}}
\newcommand{\RR}{\mathbb{R}}

\newcommand{\pol}{f}

\newcommand{\mdp}{\ensuremath{\mathcal{M}}} % G: frak should be avoided!
\newcommand{\sinit}{{\ensuremath{s_\iota}}}
\newcommand{\pmdptuple}{\ensuremath{(S,\Params,\Act,\sinit,P)}}

\newcommand{\supp}{\ensuremath{\mathrm{supp}}}
\newcommand{\Distr}{\ensuremath{\mathrm{Distr}}}
\newcommand{\act}{\ensuremath{a}}
\newcommand{\Act}{\ensuremath{\textsl{Act}}}

\newcommand{\sched}{\ensuremath{\sigma}}
\newcommand{\Sched}{\ensuremath{\Sigma}}
\newcommand{\RSched}{\ensuremath{R\Sigma}}

\newcommand{\Params}{\ensuremath{X}} %{\ensuremath{\vec{p}}}}
\newcommand{\inst}{\ensuremath{\mathit{val}}} %{\ensuremath{\vec{x}}}}
\newcommand{\ParamSpace}{\ensuremath{\mathcal{P}}}

\newcommand{\sol}{\ensuremath{\textsl{sol}}}
\newcommand{\maxsol}{\ensuremath{\textsl{maxsol}}}
\newcommand{\minsol}{\ensuremath{\textsl{minsol}}}

\newcommand{\playerI}{\ensuremath{\mathrm{I}}}
\newcommand{\playerII}{\ensuremath{\mathrm{II}}}

\newcommand{\csrg}{\ensuremath{\mathcal{G}}}
\newcommand{\csrginit}{\ensuremath{\sinit}}

\renewcommand{\phi}{\varphi}

\providecommand{\st}{}
\renewcommand{\st}{\ensuremath{\:\mid\:}}
\newcommand{\quant}{\ensuremath{\mathcal{Q}}}
\newcommand{\reach}{\ensuremath{\mathrm{Reach}}}
\newcommand{\robreach}{\ensuremath{\mathrm{RobReach}}}

\def\Inf{\operatornamewithlimits{inf\vphantom{p}}}

\let\Pr\relax
\def\Pr{\mathop{\mathrm{Pr}}\nolimits}

\begin{document}

\maketitle

\begin{abstract}
  This paper studies parametric Markov decision processes (pMDPs), an extension to Markov decision processes (MDPs) where transitions probabilities are described by polynomials over a finite set of parameters. Fixing values for all parameters yields MDPs.
  In particular, this paper studies the complexity of finding values for these parameters such that the induced MDP satisfies some reachability constraints.
  We discuss different variants depending on the comparison operator in the constraints and the domain of the parameter values.  
  We improve all known lower bounds for this problem, and notably provide ETR-completeness results for distinct variants of this problem.
  Furthermore, we provide insights in the functions describing the induced reachability probabilities, and how pMDPs generalise concurrent stochastic reachability games. 
\end{abstract}

\section{Introduction}

Markov decision processes (MDPs) are \emph{the} model to reason about sequential processes under (stochastic) uncertainty and non-determinism.
Markov chains (MCs) are MDPs without non-determinism.
Often, probability distributions in these models are difficult to assess precisely during design time of a system.
This shortcoming has led to interval MCs~\cite{DBLP:conf/lics/JonssonL91,DBLP:journals/ipl/ChenHK13,DBLP:conf/rp/Sproston18,DBLP:conf/tacas/SenVA06} and interval MDPs (aka: Bounded-parameter MDPs)~\cite{givan2000bounded,DBLP:journals/ai/WuK08,DBLP:conf/cav/PuggelliLSS13}, which allow for interval-labelled transitions.
Analysis under interval Markov models is often too pessimistic:
The actual probabilities on the transitions are considered to be non-deterministically and \emph{locally} chosen.
Intuitively, consider the probability of a coin-flip yielding heads in some stochastic environment. In interval models, the probability may vary with the local memory state of an agent acting in this environment. Such behaviour is unrealistic.
\emph{Parametric} MCs/MDPs~\cite{DBLP:conf/ictac/Daws04,DBLP:journals/fac/LanotteMT07,DBLP:conf/nfm/HahnHZ11,DBLP:journals/ai/DelgadoSB11} (pMCs, pMDPs) overcome this limitation by adding dependencies (or couplings) between various transitions---they add global restrictions to the selection of the probability distributions.
Intuitively, the probability of flipping heads can be arbitrary, but should be independent of an agent's local memory. 
Such couplings are similar to restrictions on schedulers in decentralised/partially observable MDPs, considered in e.g.,\cite{DBLP:journals/mor/BernsteinGIZ02,DBLP:journals/tcs/GiroDF14,DBLP:journals/aamas/SeukenZ08}.

Technically, pMDPs label their transitions with polynomials over a finite set of parameters. Fixing all parameter values yields MDPs. 
The synthesis problem considered in this paper asks to find parameter values
such that the induced MDPs satisfy reachability constraints.  Such reachability
constraints state that the probability---under some/all possible ways to resolve
non-determinism in the MDP---to reach a target state is (strictly) above or
below a threshold.
An example synthesis problem is thus: ``Are there parameter values such
that for all possible ways to resolve the non-determinism, the probability to
reach a target state exceeds $\frac{1}{2}$?''
Variants of the synthesis problem are obtained by varying the reachability constraints, and the domain of the parameter values. %Details are given in Sect.~\ref{sec:landscape}.
%Furthermore, an important subclass restricts parameter values to those which induce an MDP with the same topology (graph-preserving).
Parameter synthesis is supported by
the model checkers PRISM~\cite{DBLP:conf/cav/KwiatkowskaNP11} and Storm~\cite{DBLP:conf/cav/DehnertJK017}, and dedicated
tools PARAM~\cite{param_sttt} and
PROPhESY~\cite{DBLP:conf/cav/DehnertJJCVBKA15}. The complexity of the decision problems corresponding to parameter synthesis is mostly open.

This paper significantly extends complexity results for parameter
synthesis in pMCs and pMDPs.
Table~\ref{tab:reach-zoo} on~page \pageref{tab:reach-zoo} gives an
overview of new results:  
Most prominently, it establishes ETR-completeness of reachability
problems for pMCs with non-strict comparison operators, and establishes
NP-hardness for pMCs with strict comparison operators.  For pMDPs with demonic
non-determinism, it establishes ETR-completeness for any comparison operator.
For angelic non-determinism, mostly the synthesis problems are equivalent to their pMC counterparts.  
When considering
pMDPs with a fixed number of variables, we establish uniform NP upper bounds for parameter synthesis under angelic or demonic non-determinism.
These results are partially based on properties of pMDPs scattered in earlier work, and prominently use a strong connection between polynomial inequalities and parameter synthesis.

Finally, pMDPs are interesting generalisations of
other models: 
\cite{DBLP:conf/uai/Junges0WQWK018} shows that parameter synthesis in
pMCs is equivalent to the synthesis of finite-state controllers
(with a-priori fixed bounds) of partially observable MDPs (POMDPs)~\cite{DBLP:books/daglib/0023820} under
reachability constraints.  
Thus, as a side product we improve
complexity bounds~\cite{DBLP:journals/toct/VlassisLB12,DBLP:conf/aaai/ChatterjeeCD16} for (a-priori fixed) memory bounded strategies in POMDPs.
In this paper, we show how pMDPs generalise
concurrent stochastic reachability games~\cite{Sha53,DBLP:journals/tcs/AlfaroHK07,survey-CSGs}.
We finish the paper by drawing some connections with robust schedulers, i.e.\ the question how to optimally resolve non-determinism taking into account the uncertainty in the stochastic dynamics.

\subparagraph*{Related work.}
Various results in this paper extend work by Chonev~\cite{Chonev17}, who studied
a model of augmented interval Markov chains. These coincide with parametric
Markov chains.  The work also builds upon results by Hutschenreiter \emph{et
al.}\cite{baiercomplexity}, in particular upon the result that pMCs with an
a-priori fixed number of parameters can be checked in P. Furthermore, they study
the complexity of PCTL model checking of pMCs.  The complexity of finite-state
controller synthesis in POMDPs has been studied in~\cite{DBLP:journals/toct/VlassisLB12,DBLP:conf/aaai/ChatterjeeCD16}.
Some of the proofs for ETR-completeness presented here reuse ideas from~\cite{DBLP:journals/mst/SchaeferS17}.

Methods (and implementations) to analyse pMCs by computing their characteristic
\emph{solution function} are considered in \
\cite{DBLP:conf/ictac/Daws04,param_sttt,DBLP:conf/cav/DehnertJJCVBKA15,baiercomplexity,DBLP:journals/tse/FilieriTG16,DBLP:conf/qest/JansenCVWAKB14,DBLP:journals/ai/DelgadoSB11,DBLP:conf/atva/GainerHS18}.
Sampling-based approaches to find feasible instantiations in pMDPs are
considered by~\cite{DBLP:conf/nfm/HahnHZ11,DBLP:conf/tase/ChenHHKQ013},
while~\cite{DBLP:conf/tacas/BartocciGKRS11,DBLP:conf/atva/CubuktepeJJKT18}
utilise optimisation methods. 
Finally,~\cite{DBLP:conf/atva/QuatmannD0JK16} presents a method to prove the absence of solutions in pMDPs by iteratively considering simple stochastic games~\cite{DBLP:books/daglib/0074447}.
Some other works on Markov models
with structurally equivalent yet parameterised dynamics
include~\cite{chatterjee12,solan03,DBLP:conf/concur/ChenFRS14,DBLP:journals/acta/CeskaDPKB17}.
Parameter synthesis with statistical guarantees has been explored in, e.g., \cite{DBLP:conf/tacas/BortolussiS18}. 
Further work on parameter synthesis in Markov models has been surveyed in \cite{prophesy_journal}.

%
%\subparagraph*{Contribution.}
%Furthermore, we deepen\jpk{too vague} our understanding of the form of rational functions that
%describe the reachability probability in pMCs, and lift some results to
%robust strategies.  

\section{Preliminaries}

Let $X$ be a finite set of variables. Let $\QQ[X]$ and $\QQ(X)$ denote the set of
all rational-coefficient polynomial and rational functions on $X$, respectively.
A rational function $f/g$ can be represented as a pair $(f,g)$ of polynomials.
In turn, a polynomial can be represented as a sum of terms, where each term is
given by a \emph{coefficient} and a \emph{monomial}. The (total) degree of a
polynomial is the maximum over the sum of the exponents in the monomials.  A
polynomial is quadratic (respectively, quadric), if its total degree is two
(four) or less.
For a rational function $f(x_1,\dots,x_k) \in \QQ(X)$ and an instantiation
$\inst \colon X \to \RR$ we write $f[\inst]$ for the value
$f(\inst(x_1),\dots,\inst(x_k))$.
We use $\bowtie$ to denote either of $\{ \leq, <, \geq, > \}$ and
$\unrhd$ for either $\{\geq, >\}$ (and $\unlhd$ analogously). With $\overline{\bowtie}$, we denote the complement, e.g. $\overline{\leq}\;=\;>$.

Consider a finite set $S$. Let $\Distr(S)$ denote the set of all distributions
over $S$, and $\supp(\delta) \subseteq S$ the support $\{s \in S \mid \delta(s) >
0\}$ of distribution $\delta \in \Distr(S)$.

\subsection{Parametric Markov models}

%\subparagraph*{Parametric models.}
\begin{definition}[pMDP]
  A \emph{parametric Markov Decision Process} $\mdp$ is a tuple $\pmdptuple$
  with $S$ a (finite) set of \emph{states}, $ \Params$ a finite set of
  \emph{parameters}, $\Act$ a finite set of \emph{actions}, $\sinit \in S$ the initial state, and $P\colon S
  \times \Act \times S \rightarrow \QQ[\Params] \cup \RR$ the probabilistic transition
  function.
\end{definition}
Parameter-free pMDPs coincide with standard MDPs, as in
\cite{puterman05}.  We define $\Act(s) = \{ \act \in \Act \mid \exists s' \in
S.\; P(s,a,s') \neq 0 \}$.  If $|\Act(s)| = 1$ for all $s \in S$, then $\mdp$ is
a \emph{parametric Markov chain} (pMC). We denote its transitions with
$P(s,s')$ and omit the actions.

A pMDP is \emph{simple} if and only if non-constant probabilities labelling transitions
$(s,a,s')$ are of the form $x$ or $1-x$, and the sum of outgoing transitions from a state-action pair always is (equivalent to) $1$.
Formally, simple pMDPs satisfy the following two properties:
\begin{itemize}
  \item $P(s,\act,s') \in \{ x, 1-x \st x \in \Params \} \cup \RR$ for all $s,s'
    \in S$ and $\act \in \Act$; and
  \item $\sum_{s'\in S} P(s,\act,s') = 1$ for all $s \in S$ and $\act \in
    \Act(s)$.
\end{itemize}

%\subparagraph*{Instantiations, parameter space and regions.}
\begin{definition}[Instantiation]
  Let $\mdp = \pmdptuple$ be a pMDP.  
  An instantiation $\inst\colon X \rightarrow \RR$ is \emph{well-defined} 
  if the induced functions $P(s,\act,\cdot)$ are distributions over $S$,
  %from $\Distr(S)$,
   i.e.
  \[ 
    \forall s,s' \in S, \forall \act \in \Act.\; 0 \leq P(s,\act,s')[\inst] \in \RR
    \land \sum_{\hat{s}\in S} P(s,\act,\hat{s})[\inst] = 1.
  \]
\end{definition}
Let $\mdp[\inst]$ denote the parameter-free MDP in which $P(s,\act,s')$ has
been replaced by $P(s,\act,s')[\inst]$.
We denote with $P^\inst_{=0} \coloneqq \{ (s,a,s') \in S \times \Act \times S \mid
P(s,\act,s') \neq 0 \land P(s,\act,s')[\inst] = 0 \}$ the transitions of
$\mdp$ that become $0$ in $\mdp[\inst]$.  A well-defined instantiation $\inst$
is \emph{graph-preserving} if the topology of the pMDP is preserved, i.e. if
$P^\inst_{=0} = \emptyset$.

The (well-defined) parameter space $\ParamSpace^{\mathrm{wd}}_\mdp$ for $\mdp$
is $\{ \inst \colon \Params \to \RR \st \inst \text{ is well defined} \}$ and
the graph-preserving parameter space $\ParamSpace^{\mathrm{gp}}_\mdp \coloneqq \{ \inst
\colon \Params \to \RR \st \inst \text{ is graph-preserving} \}$.  In simple
pMDPs, the well-defined (respectively, graph-preserving) parameter space is the
set of instantiations $\inst \colon \Params \to [0,1]$ (respectively, $\inst
\colon \Params \to (0,1)$). We omit the subscript from $\ParamSpace^{\mathrm{gp}}_\mdp$ and $\ParamSpace^{\mathrm{wd}}_\mdp$  when the pMDP $\mdp$ is understood
from the context.

A \emph{graph-consistent region} $R$ for $\mdp$ is a subset of
$\ParamSpace^{\mathrm{wd}}$ such that all instantiations in $R$ induce the same graph, i.e., $P^\inst_{=0} =
P^{\inst'}_{=0}$ for all $\inst,\inst' \in R$.
\begin{remark} \label{rem:fixedparam}
  For any simple pMDP, $\ParamSpace^\text{wd}$ can be partitioned into
  $3^{|\Params|}$ many graph-consistent regions $R$.  For any graph-consistent
  region, we can (in linear time) construct a simple pMDP $\mdp'$ such that the
  graph-consistent region $R$ corresponds to $\ParamSpace^\mathrm{gp}_{\mdp'}$.
  Essentially, the construction merely removes the transitions $P^{\inst}_{=0}$
  for $\inst \in R$, and adjusts the probabilities of some other transitions to
  $1$ to ensure simplicity.  A complete construction is given in
  \cite{prophesy_journal}.
\end{remark}

\subparagraph*{Reachability, schedulers and induced Markov chains.}
Consider a parameter-free MC $\mdp$ and a state $s_0$.  A \emph{run} of $\mdp$
from $s_0$ is an infinite sequence of states $s_0 s_1 \dots$ such that
$P(s_i,s_{i+1}) > 0$ for all $i \geq 0$. We denote by $\mathrm{Runs}^{s_0}$ the set of
all runs of $\mdp$ that start with the state $s_0$.  The \emph{probability} of
measurable \emph{event} $E \subseteq \mathrm{Runs}^{s_0}$ is 
defined using a standard cylinder construction~\cite{puterman05,BK08}. 
Let $\Pr_{\mdp}(\lozenge T)$ denote the probability to \emph{eventually reach
$T$} from the initial state of $\mdp$; and
$\Pr_{\mdp}(s \rightarrow \lozenge T)$ denote the probability to eventually reach
$T$ starting from state $s$. We omit the subscript $\mdp$ if it is clear from the context.

To define reachability in pMDPs, we need to eliminate the non-determinism. We do
so by means of a \emph{scheduler} (a.k.a. a \emph{policy} or \emph{strategy}).
\begin{definition}[Scheduler]
  A randomised (memoryless) scheduler is a function $\sched\colon S \rightarrow
  \Distr(\Act)$ s.t.\ $\supp(\sched(s)) \subseteq \Act(s)$.  A scheduler is
  deterministic if $|\supp(\sched(s))| = 1$ (i.e. $\sched(s)$ is Dirac) for
  every $s \in S$.  We refer to deterministic schedulers as \emph{schedulers}.
\end{definition}
We denote the set of randomised schedulers with $\RSched$, and (deterministic) schedulers
with $\Sched$.

%\begin{definition}[Induced pMC]
  For pMDP $\mdp = \pmdptuple$ and $\sched \in \RSched_\mdp$, the
  \emph{induced pMC} $\mdp^\sched$ is defined as $(S,X, \sinit, P')$ with $P'(s,s') =
  \sum_{\act \in \Act} \sched(s)(\act) \cdot P(s,\act,s')$. 
%\end{definition}
For simple pMDPs, the induced pMC of a deterministic scheduler
is simple. Under randomised schedulers, the induced pMC can be transformed
into a simple pMC (e.g.\ \cite{DBLP:conf/uai/Junges0WQWK018}). We abbreviate
$\Pr_{\mdp^\sched}$ by $\Pr^\sched_{\mdp}$.

\begin{remark} \label{rem:detsched}
 Deterministic schedulers dominate
  randomised schedulers for reachability properties~\cite{puterman05}, i.e.\ for each MDP there exists a
  deterministic scheduler $\sched$ s.t.\ $\Pr^\sched_\mdp(\lozenge T)
  = \sup_{\sched' \in \RSched} \Pr^{\sched'}_\mdp(\lozenge T)$. 
  Therefore, in the remainder, we focus on deterministic
  schedulers.
\end{remark}

\begin{definition}[Solution function]
  For a pMC $\mdp$ and a state $s$, let the \emph{solution function}
  $\sol^{\mdp,T}_s\colon  \ParamSpace^{\mathrm{wd}}_\mdp \rightarrow [0,1]$ be defined as
  $\sol^{\mdp,T}_s[\inst] \coloneqq \Pr_{\mdp[\inst]}(s \rightarrow \lozenge T)$.  For a pMDP
  $\mdp$, let $\minsol_s^{\mdp,T}[\inst] \coloneqq \min_{\sched \in \Sched} \sol^{\mdp^\sched,T}_s[\inst] = 
  \min_{\sched \in \Sched} \Pr^\sched_{\mdp[\inst]}(s \rightarrow \lozenge T)$.  We
  define $\maxsol_s^{\mdp,T}$ analogously as the maximum.
\end{definition}
Let $\sol^{\mdp,T}$ denote $\sol^{\mdp,T}_{\sinit}$ with the convention that $T$
is omitted
whenever it is clear from the context.  On
$\ParamSpace^{\mathrm{gp}}$, $\sol^{\mdp}$ is described by a rational
function over the parameters~\cite{DBLP:conf/ictac/Daws04,DBLP:journals/fac/LanotteMT07}, and is
computable in $\mathcal{O}\left(\textsl{poly}(|S|\cdot d)^{|\Params|}\right)$, where $d$
is the maximal degree of polynomials in $\mdp$'s
transitions~\cite{baiercomplexity}.
%We omit $T$ whenever it is clear from the context.
The number of resulting monomials is polynomial in $|S|$ and $d$ but exponential
in $|\Params|$.
Furthermore, the degree of $f$ and $g$ in the resulting function $f/g$ is
upper-bounded by $\ell(d)$ --- where $\ell$ is a linear
function.\footnote{Importantly, this means that if the coefficients and
exponents were written in binary for the given pMC then linearly more bits
suffice to do the same for the computed rational function.}
For acyclic pMCs, $\sol^\mdp$ is described by a polynomial.

\subsection{Existential theory of the reals}
Many results in this paper are based on results from the
existential theory of the reals~\cite{BPR06}. We give a brief recap.
We consider the first-order theory of the reals: the set of all valid sentences
in the first-order language $(\mathbb{R},+,\cdot,0,1,<)$. The existential theory
of the reals restricts the language to (purely) existentially quantified
sentences. The complexity of deciding membership, i.e.\ whether a sentence is
(true) in the theory of the reals, is in PSPACE~\cite{DBLP:conf/stoc/Canny88} and NP-hard. A
careful analysis of its complexity is given in
\cite{DBLP:journals/jsc/Renegar92}. In particular, deciding membership for
sentences with an a-priori fixed upper bound on the number of variables is in
polynomial time.  ETR denotes the complexity
class~\cite{DBLP:journals/mst/SchaeferS17} of problems with a polynomial-time
many-one reduction to deciding membership in the existential theory of the
reals.

\section{Problem landscape}
\label{sec:landscape}
In this section, we introduce the family of decision problems of our main interest.  Let a
\emph{simple} pMDP $\mdp$ with all constants rational, and a set $T$ of target states be the given input.
We analyse the decision problems
according to whether the set $\Params$ of parameters from $\mdp$ has bounded
size---with a-priori fixed bound---or arbitrary size.

It remains for us to fix an encoding for rational functions. Henceforth, we
assume the coefficients, exponents and constants are all given as binary-encoded
integer pairs.

\subparagraph*{Decision problems.}
The first problem is the existence of so-called
\emph{robust parameter values} or lack thereof. More precisely, the question is whether
some instantiation of $\mdp$ is such that its maximal or minimal probability
of  eventually reaching $T$ compares with $\frac{1}{2}$ in some desired way. In
symbols, for $\quant_1,\quant_2 \in \{\exists,\forall\}$ and 
$\mathrm{\bowtie} \in
\{\leq, <, >, \geq\}$, let
\[ 
  \quant_1\quant_2\reach^{\bowtie}_{\mathrm{wd}} \stackrel{\mathrm{def}}{\iff}
  \quant_1\;\inst \in \ParamSpace^{\mathrm{wd}},
  \quant_2\;\sched \in \Sched.\;
    \Pr_{\mdp[\inst]}^{\sched}(\lozenge T) \bowtie \frac{1}{2}
\]
be the problem of interest.
We write $\quant_1\quant_2\reach^{\bowtie}_{\mathrm{gp}}$ whenever $\quant_1$ quantifies over graph-preserving instantiations.
We write $\quant_1\quant_2\reach^{\bowtie}_{*}$ to denote both the $\mathrm{wd}$ and $\mathrm{gp}$ variants.
Furthermore, if $\mdp$ is a pMC we omit the second quantifier, e.g. $\exists \reach^{<}_{*}$.
Table~\ref{tab:reach-zoo} surveys the results.

\newcommand{\refsize}[1]{{\scriptsize #1}}
\begin{table}[tb]
\centering
\setlength{\tabcolsep}{5pt}
\begin{tabular}{|ll|c|c|c|}
\hline
           &                      & Fixed \#  &
\multicolumn{2}{c|}{Arbitrary \# parameters} \\
\cline{4-5}
& & parameters &  well-defined & graph-preserving \\
\hline
\hline
\parbox[t]{1mm}{\multirow{3}{*}{\rotatebox[origin=c]{90}{pMC}}} & $\exists\reach^{\geq/\leq}$             & in P~\refsize{\cite{baiercomplexity}}                 &  \multicolumn{2}{c|}{--- ETR-complete \refsize{[Thm.~\ref{thm:etr:pmcsarehard}]} ---}  \\
& $\exists\reach^{>}$             		&     ''             &  NP-hard \refsize{[Thm.~\ref{thm:e_reach_gr_np_hard}]}  &    $\exists\reach^{>}_\mathrm{wd}$-complete \refsize{[Lem.~\ref{lem:gpvswdequiv}]}           \\
& $\exists\reach^{<}$             		&       ''           &  NP-hard \refsize{[Thm.~\ref{thm:e_reach_gr_np_hard}]}                     &    $\exists\reach^{>}_\mathrm{wd}$-complete \refsize{[Lem.~\ref{lem:restrict}]}           \\
\hline
\parbox[t]{1mm}{\multirow{4}{*}{\rotatebox[origin=c]{90}{pMDP}}} & $\exists\exists\reach^{\geq/\leq}$              & in NP \refsize{[Thm.~\ref{lem:fp_ee_reach_in_np}]}       & \multicolumn{2}{c|}{--- ETR-complete --- \refsize{(trivial)}} \\
& $\exists\exists\reach^{>}$      & ''                     & \multicolumn{2}{c|}{--- $\exists\reach^{>}_\mathrm{wd}$-complete \refsize{[Lem.~\ref{lem:existwdgeeqexistwdge}, Cor.~\ref{cor:eeggpeqerggwd}]} ---} \\
& $\exists\exists\reach^{<}$   & ''                     &    $\exists\reach^{<}_\mathrm{wd}$-complete \refsize{[Lem.~\ref{lem:existwdgeeqexistwdge}]}    &   $\exists\reach^{>}_\mathrm{wd}$-hard   \refsize{(trivial)}       \\
& $\exists\forall\reach^{\bowtie}$                    & in NP \refsize{[Thm.~\ref{thm:fp_ea_reach_in_np}]}                    &  \multicolumn{2}{c|}{--- ETR-complete \refsize{[Thm.~\ref{thm:etr:mdpsarehard}]} ---}    \\
\hline 
\end{tabular}
\caption{The complexity landscape for reachability in simple pMDPs. All problems are in ETR.}

\label{tab:reach-zoo}
\end{table}

\begin{proposition}\label{pro:all-etr}
  For every $\quant_1,\quant_2 \in \{\exists,\forall\}$ and 
  $\mathrm{\bowtie} \in \{\leq, <, >, \geq\}$,
  $\quant_1\quant_2\reach^{\bowtie}_{\mathrm{*}}$ are decidable in ETR.
\end{proposition}
\begin{proof}
  Both 
  $\exists\forall\reach^{\bowtie}_\mathrm{*}$ and $\exists\exists\reach^{\bowtie}_\mathrm{*}$  are in ETR (for an encoding, see
  Appendix~\ref{app:fullencoding}). 
  It follows that   $\exists\reach^{\bowtie}_{*}$ are also in ETR.
\end{proof}

\subparagraph*{Problems with fixed threshold.}
In the above-defined problems, we have fixed a threshold of $\frac{1}{2}$. This is no
loss of generality as any \emph{given}  rational threshold can be reduced
to $\frac{1}{2}$:
% The following remark formalises this.
\begin{remark} \label{rem:fixedlambda}
  An arbitrary threshold $0 < \lambda < 1$, $\lambda \in \QQ$, is reducible to $\frac{1}{2}$ by the
  constructions depicted in Fig.~\ref{fig:fixedlambda}: If $\lambda \leq
  \frac{1}{2}$ then we prepend a transition with probability $p = 2\lambda$ to the
  initial state and with probability $1-p$ to a sink state.  Otherwise, if
  $\lambda > \frac{1}{2}$, we prepend a transition with probability $q =
  2(1-\lambda)$ to the initial state and $1-q$ to the target state.
  %We will exploit this argument several times. 
  Conversely, the $\frac{1}{2}$ threshold may analogously be reduced to an arbitrary
  threshold $0 < \lambda < 1$.% by similar constructions.
\end{remark}

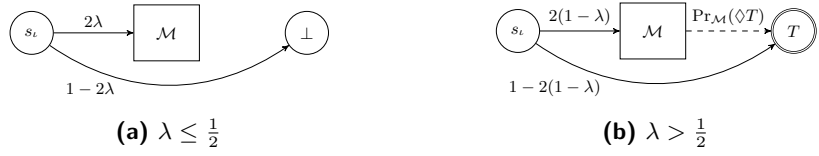
\begin{figure}[tb]
\centering
	\small
	\begin{subfigure}[b]{0.45\textwidth}
	\centering
	\scalebox{0.7}{
		\begin{tikzpicture}[baseline=(start)]
			\node[state] (start) {$\sinit$};
			
			\node[right=1.5cm of start, rectangle, draw, inner sep=12pt, xshift=0mm] (box) {$\mdp$};
			\node[state, right=1.6cm of box] (sink) {$\bot$};
			
			\draw[->] (start) -- node[above] {$2\lambda$} (box);
			\draw[->] (start) edge[bend right=35] node[left,pos=0.3,yshift=-2mm] {$1- 2\lambda$} (sink);
			%\draw[->] (sink) edge[loop above] node[right] {$1$} (sink);
		\end{tikzpicture}
		}
		\subcaption{\centering$\lambda \leq \frac{1}{2}$}
	\end{subfigure}
	\small
	\begin{subfigure}[b]{0.45\textwidth}
	\centering
	\scalebox{0.7}{
		\begin{tikzpicture}[baseline=(start)]
			\node[state] (start) {$\sinit$};
			
			\node[right=1.5cm of start, rectangle, draw, inner sep=12pt, xshift=0mm] (box) {$\mdp$};
			\node[state, accepting, right=1.6cm of box] (T) {$T$};
			
			\draw[->] (start) -- node[above] {$2(1-\lambda)$} (box);
			\draw[->] (start) edge[bend right=35] node[left,pos=0.3,yshift=-2mm] {$1- 2(1-\lambda)$} (T);
			%\draw[->] (T) edge[loop above] node[right] {$1$} (T);
			\draw[->,dashed] (box) -- node[above] {$\Pr_{\mdp}(\lozenge T)$} (T);
		\end{tikzpicture}
		}
	\subcaption{\centering$\lambda > \frac{1}{2}$}
	\end{subfigure}
	\caption{Reductions to reachability threshold $\lambda=\frac{1}{2}$, cf.\ Remark~\ref{rem:fixedlambda}.}
	\label{fig:fixedlambda}
\end{figure}

\subparagraph*{Considerations for the comparison relations.}
\begin{restatable}{lemma}{restaterestrictcomparisons}
 \label{lem:restrict}
  For every $\quant_1,\quant_2 \in \{\exists,\forall\}$, 
  there are polynomial-time Karp reductions
  \begin{itemize}
    \item among the problems $\quant_1 \quant_2 \reach^{>}_{\mathrm{gp}}$ and
      $\quant_1 \quant_2 \reach^{<}_{\mathrm{gp}}$ and
    \item among the problems $\quant_1 \quant_2 \reach^{\geq}_{\mathrm{gp}}$ and
      $\quant_1 \quant_2 \reach^{\leq}_{\mathrm{gp}}$.
  \end{itemize}
\end{restatable}
\noindent The above claim only holds when restricted to graph-preserving parameter spaces.
\subparagraph*{Semi-continuity.}
The following theorem formalises an observation in \cite[Thm. 5]{DBLP:conf/uai/Junges0WQWK018}. 
\begin{restatable}{theorem}{restatesemicontinuous}
 \label{thm:semicontinuous}
  For each simple pMC $\mdp$, the function $\sol^\mdp$ is lower semi-continuous, and continuous on $\ParamSpace^\mathrm{gp}_\mdp$. For acyclic simple pMCs $\mdp$, $\sol^\mdp$ is continuous on $\ParamSpace^\mathrm{wd}_\mdp$.
\end{restatable}
\begin{wrapfigure}{R}{0.15\textwidth}
%\vspace{-2mm}
\scalebox{0.7}{
\begin{tikzpicture}
	\node[state, initial, initial text=,scale=0.4] (s0) {};
	\node[state, right=of s0, accepting,scale=0.4] (s1) {};
	\draw[->] (s0) edge node[above] {$p$} (s1);
	\draw[->] (s0) edge[loop above] node[left, pos=0.3] {$1-p$} (s0);
\end{tikzpicture}	
}
\caption{}
\label{fig:semicontinuous}
\end{wrapfigure}

\smallskip\noindent Continuity on $\ParamSpace^\mathrm{gp}_\mdp$ follows as $\sol^\mdp$ is a rational function bounded by $[0,1]$ on all
well-defined points~\cite{DBLP:conf/atva/QuatmannD0JK16}. 
Graph non-preserving instantiations might yield additional sink states in the induced MC, therefore, the probability may drop when changing a parameter instantiation, e.g.\ $p=0$ in  Fig.~\ref{fig:semicontinuous}.

The semi-continuity is the main reason that we do not have symmetric entries for
upper and lower bounds in Table~\ref{tab:reach-zoo}. 

\noindent The following result follows
immediately from properties of (semi-)continuous functions.
\begin{corollary} \label{cor:pmdpminexists}
  For all pMDPs, the functions $\minsol^\mdp$ and $\maxsol^\mdp$
  are lower semi-continuous and have a minimum. For acyclic pMDPs, these functions are continuous.
\end{corollary}
\begin{corollary}
	For acyclic pMCs, $\sol^\mdp$ is described by a polynomial even on $\ParamSpace^\mathrm{wd}_\mdp$.
\end{corollary}

\section{Fixing the number of parameters}
In this section, we assume that the number of parameters is fixed. We focus
ourselves on graph-preserving instantiations, as the analysis of pMDP $\mdp$ and
$\ParamSpace^\mathrm{wd}_\mdp$ corresponds to analysing constantly many pMDPs $\mdp'$ on
$\ParamSpace^\mathrm{gp}_{\mdp'}$, cf.\ Rem.~\ref{rem:fixedparam}.

\subparagraph*{Upper bounds.}
Below, we establish NP membership for all variants.
\begin{lemma}
In the fixed parameter case, $\exists\exists\reach^{\bowtie}_{*}$ is in NP.
\label{lem:fp_ee_reach_in_np}
\end{lemma}
\begin{proof}
Guess a memoryless scheduler. 
Construct the induced pMC, and verify it in P.
\end{proof}
The main result in this paragraph is:
\begin{restatable}{theorem}{restatefpexfainnp}
In the fixed parameter case, $\exists\forall\reach^{\bowtie}_{*}$ is in NP.
\label{thm:fp_ea_reach_in_np}
\end{restatable}
In the non-parametric case, a scheduler $\sigma$ of an MDP is called \emph{minimal} if it minimises $\Pr^\sched(\lozenge T)$, i.e. if
$\sigma \in \text{argmin}_{\sched' \in \Sched}\; \Pr^{\sched'}(\lozenge T).$
Consider the probabilities $x_s = \Pr^{\sigma}(s \rightarrow \lozenge T)$ for $s \in S$.
It is well-known (see, e.g.,~\cite{puterman05}) that $\sigma$ is minimal if and only if
$
	x_s \leq \sum_{s' \in S}P(s,\act,s')\cdot~x_{s'}
$
 holds for all $s \in S$ and $\act \in \Act(s)$. (There is a similar condition
for \emph{maximal} schedulers.) The minimality criterion can be lifted to
the parametric case: 
Suppose $R \subseteq \ParamSpace^\mathrm{wd}$ is a
graph-consistent region and let $f_s = \sol^{\mdp^\sched}_s$. 
Then $\sigma$ is \emph{somewhere} minimal on $R$ if and only if there exists some $\inst \in R$ such that
\begin{equation} 
	\label{eq:parametricmincrit}
	f_s[\inst] \leq \sum_{s' \in S}P(s,\act,s')\cdot f_{s'}[\inst]
\end{equation}
for all $s \in S$ and $\act \in \Act$. (For \emph{everywhere} minimal
strategies, a universal quantification over $\inst$ yields the correct
criterion).

\begin{restatable}{lemma}{restatefpchecksched}
	\label{lem:fp_check_sched_opt_in_np}
	In the fixed parameter case, checking whether a given strategy is somewhere (resp. everywhere) minimal (resp. maximal) on $\ParamSpace^\mathrm{wd}$ is in P.
\end{restatable}
\begin{proof}[Proof sketch]
  Condition \eqref{eq:parametricmincrit} can be reformulated as the ETR formula
  with $|X|$ many variables%
  \begin{equation}\label{eq:optimalstratparametric}
    \Psi = \exists \inst\colon \Phi_R(\inst) \longrightarrow \Phi_\sigma(\inst)
  \end{equation}
	where $\Phi_R(\inst)$ is a formula which is true if and only if $\inst \in R$ and
  \begin{equation} \label{eq:optimalstratencoding}
    \Phi_\sigma(\inst) = \bigwedge_{s \in S} \bigwedge_{\act \in \Act} \bigg(g_s[\inst] \cdot \prod_{s' \neq s}h_{s'}[\inst]\; \leq \;\sum_{s' \in S}P(s, \act, s') \cdot g_{s'}[\inst] \cdot \prod_{s'' \neq s'}h_{s''}[\inst]\bigg)
  \end{equation}
where $\nicefrac{g_s}{h_s} = f_s$ for $g_s, h_s\in \QQ[\inst]$. (W.l.o.g.\ it holds that $h_s[\inst] > 0$ for all $\inst \in R$.) 
\end{proof}

\begin{proof}[Proof sketch of Thm.~\ref{thm:fp_ea_reach_in_np}]
Consider $\bowtie\;=\;\geq$: Guess a somewhere minimal scheduler. Check its minimality similar to Lem.~\ref{lem:fp_check_sched_opt_in_np}, but extended to simultaneously ensure that the induced pMC satisfies the threshold. The other relations in $\bowtie$ are analogous.
\end{proof}

\subparagraph*{Sets of optimal schedulers.}
For the problems $\forall\forall\reach^{\bowtie}_{*}$ and $\forall\exists\reach^{\bowtie}_{*}$ (with fixed parameters)
we already have coNP-membership (as we considered their complements before).
It is tempting to assume that their NP-membership  can be
established analogous to above, relying on \emph{everywhere} optimal schedulers which, according to Lem.~\ref{lem:fp_check_sched_opt_in_np}, can also be verified in polynomial time.
However, such schedulers do not necessarily exist. What we
need instead is a \emph{set} of somewhere optimal schedulers covering the
entire parameter space---a so called \emph{optimal-scheduler set (OSS)}. 

\begin{definition}[Optimal scheduler set]
	A set $\Omega \subseteq \Sched$ is called an \emph{optimal scheduler set} (OSS) on $R \subseteq \ParamSpace^{\mathrm{wd}}$ if
	\[ 
    \forall \inst \in R, \exists \sched \in \Omega .\; \Pr_{\mdp[\inst]}^\sched(\lozenge T) = \max_{\sched' \in \Sched}\Pr_{\mdp[\inst]}^{\sched'}(\lozenge T),
	\]
	i.e. $\Omega$ contains a maximal scheduler for every point in the region $R$. The notion can be analogously defined for minimal schedulers.
\end{definition}
An OSS of minimal cardinality is called a \emph{minimal optimal scheduler set} (MOSS). 
For many applications it is appropriate to describe a region $R$ via a quantifier-free ETR-formula $\Phi_R$ with $|X|$ free variables such that $Sat(\Phi_R) = R$. 
In that case, we have the following:

\begin{theorem}
\label{thm:fixpar-checkoss}
In the fixed parameter case, checking whether a given $\Omega \subseteq \Sched$ constitutes an OSS on $R = Sat(\Phi_R)$ can be done in time polynomial in the size of $\mdp$, $\Omega$ and $\Phi_R$.
\end{theorem}
\begin{proof}
	For every $\sigma \in \Omega$, we construct the formulas $\Phi_\sigma$ as in \eqref{eq:optimalstratencoding} in polynomial time. Then, we check whether the fixed-parameter ETR-formula
  \[ 
		\exists \inst\colon \Phi_R(\inst) \longrightarrow \bigwedge_{\sigma \in \Omega}\neg\Phi_\sigma(\inst)
  \]
	is unsatisfiable (also in polynomial time). If yes, return true and otherwise false.
\end{proof}

\begin{restatable}{lemma}{restateconpconjecture}
\label{lem:conpconjecture}
	If the size of a MOSS on $\ParamSpace^{\mathrm{wd}}$ is polynomially bounded for fixed-parameter pMDPs, then $\exists\exists\reach^{\bowtie}_{*}$ and $\exists\forall\reach^{\bowtie}_{*}$ are in coNP.
\end{restatable}
\noindent The proof considers the complement of $\exists\forall\reach^{\bowtie}_{*}$, that is $\forall\exists\reach^{\bar{\bowtie}}_{*}$. 
Under the assumption in the lemma, it now suffices to guess a MOSS and verify $\exists\reach^{\bar{\bowtie}}_{*}$ on the induced pMCs in polynomial time, showing that the complement is in NP.

\noindent In the arbitrary parameter case, we obtain an exponential lower bound on the MOSS size: 
\begin{restatable}{lemma}{restateexponentiallmoss}
\label{lem:exponentialmoss}
There exists a family $(\mdp_n)_{n \in \mathbb{N}}$ of simple pMDPs with $n+2$ states s.t.\ $|\Omega| \geq 2^n$ for any OSS $\Omega$ on $\ParamSpace^{\mathrm{wd}}$, i.e., the size of a MOSS can grow exponentially in the  pMDP's size.
\end{restatable}
This lemma, and what follows below, consider the unbounded parameter case, i.e., from now on, parameters are part of the input.

\section{The expressiveness of simple pMCs}

%\subsection{Describing polynomials with pMCs}
We investigate the relation between polynomial inequalities and the $\exists\reach$ problems.
The first lemma in this section is a key ingredient for our complexity analysis later on.
\begin{restatable}[Chonev's trick~{\cite[Remark~7]{Chonev17}}]{lemma}{restatechonevstrick}
  \label{lem:chonevtrick}
  Let $f \in \mathbb{Q}[X]$ be a polynomial, $\mu \in \QQ$ and $0 < \lambda <
  1$. There exists a simple acyclic pMC $\mdp$ with a target state $T$ such that
  for all $\inst \colon X \rightarrow [0,1]$ and all comparison relations
  $\bowtie~\in \{<,\leq,\geq,>,=\}$ it holds that
	\[ 
    f[\inst] \bowtie \mu \Longleftrightarrow \Pr_{\mdp[\inst]}(\lozenge T) \bowtie \lambda.
	\]
  Moreover, if $d$ is the total degree of $f$, $t$ the number of terms in $f$
  and $\kappa$ a bound on the (bit-)size of the coefficients and the thresholds $\mu$,
  $\lambda$, then $\mdp$ can be
  constructed in time $\mathcal{O}(poly(d,t,\kappa))$.
\end{restatable}
\begin{example}
	Consider the inequality ${-}2x^2y + y \geqslant 5$. We reformulate this to:
	$2\cdot\left( (1-x)xy + (1-x)y + (1-y) - 1 \right) + y \geqslant 5$ and then to $2
	\cdot (1-x)xy + 2\cdot (1-x)y + 2\cdot(1-y) + y \geqslant 7$. Observe that both sides now only contain positive coefficients. Furthermore, observe that we wrote the left-hand side as sum of products over $\{ x, 1-x, y, 1-y \}$. After rescaling (with $\frac{1}{8}$), we can construct the pMC $\mdp$ depicted in Fig.~\ref{fig:chonevtrick} and set $\lambda = \frac{7}{8}$.
	\label{ex:chonevtrick}
\end{example}
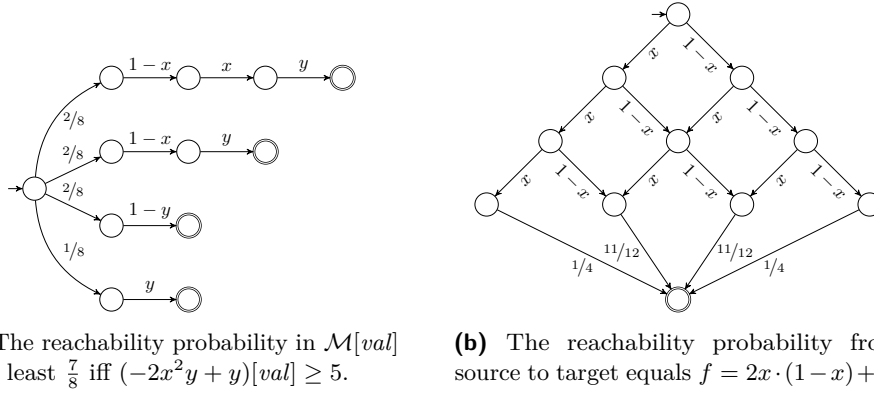
\begin{figure}
\centering
\begin{subfigure}[b]{0.42\textwidth}
\centering
\scalebox{0.7}{
\begin{tikzpicture}[initial text=,scale=0.6]
	\node[state, initial, scale=0.5] (s0) {};
	\node[state, right=of s0,yshift=2.1cm, scale=0.5] (a1) {};
	\node[state, right=of s0,yshift=0.7cm, scale=0.5] (b1) {};
	\node[state, right=of s0,yshift=-0.7cm, scale=0.5] (c1) {};
	\node[state, right=of s0,yshift=-2.1cm, scale=0.5] (d1) {};
	\node[state, right=of a1, scale=0.5] (a2) {};
	\node[state, right=of a2, scale=0.5] (a3) {};
	\node[state, right=of a3, scale=0.5, accepting] (a4) {};
	\node[state, right=of b1, scale=0.5] (b2) {};
	\node[state, right=of b2, scale=0.5,accepting] (b3) {};
	\node[state, right=of c1, scale=0.5,accepting] (c2) {};
	\node[state, right=of d1, scale=0.5,accepting] (d2) {};
	
	\draw[->] (s0) edge[bend left] node[pos=0.75,below,yshift=-1mm] {$\nicefrac{2}{8}$} (a1);
	\draw[->] (s0) edge node[above] {$\nicefrac{2}{8}$} (b1);
	\draw[->] (s0) edge node[above] {$\nicefrac{2}{8}$} (c1);
	\draw[->] (s0) edge[bend right] node[pos=0.75,above,yshift=2mm] {$\nicefrac{1}{8}$} (d1);
	
	\draw[->] (a1) edge node[above] {$1-x$} (a2);
	\draw[->] (a2) edge node[above] {$x$} (a3);
	\draw[->] (a3) edge node[above] {$y$} (a4);
	
	\draw[->] (b1) edge node[above] {$1-x$} (b2);
	\draw[->] (b2) edge node[above] {$y$} (b3);
	
	\draw[->] (c1) edge node[above] {$1-y$} (c2);
	\draw[->] (d1) edge node[above] {$y$} (d2);

\end{tikzpicture}
}

\label{fig:chonevtrick}
\caption{The reachability probability in $\mdp[\inst]$ is at least $\frac{7}{8}$ iff $({-}2x^2y + y)[\inst] \geq 5$.}
\end{subfigure}	
\hspace{5mm}
\begin{subfigure}[b]{0.42\textwidth}
\centering
\scalebox{0.7}{
\begin{tikzpicture}[initial text=, scale=0.6]
	\node[state, initial, scale=0.5] (root) at (0,0) {};
	
	\node[state, scale=0.5] (l) at (-2,-2) {};
	\node[state, scale=0.5] (r) at (2,-2) {};
	
	\node[state, scale=0.5] (ll) at (-4,-4) {};
	\node[state, scale=0.5] (lr) at (0,-4) {};
	\node[state, scale=0.5] (rr) at (4,-4) {};
	
	\node[state, scale=0.5] (lll) at (-6,-6) {};
	\node[state, scale=0.5] (llr) at (-2,-6) {};
	\node[state, scale=0.5] (lrr) at (2,-6) {};
	\node[state, scale=0.5] (rrr) at (6,-6) {};
	
	\node[state,accepting, scale=0.5] (T) at (0,-9) {};
	
	\draw[->] (root) -- node[sloped, anchor=center,below] {$x$} (l);
	\draw[->] (root) -- node[sloped, anchor=center,below] {$1-x$} (r);
	
	\draw[->] (l) -- node[sloped, anchor=center,below] {$x$} (ll);
	\draw[->] (l) -- node[sloped, anchor=center,below] {$1-x$} (lr);
	\draw[->] (r) -- node[sloped, anchor=center,below] {$x$} (lr);
	\draw[->] (r) -- node[sloped, anchor=center,below] {$1-x$} (rr);
	
	\draw[->] (ll) -- node[sloped, anchor=center,below] {$x$} (lll);
	\draw[->] (ll) -- node[sloped, anchor=center,below] {$1-x$} (llr);
	\draw[->] (lr) -- node[sloped, anchor=center,below] {$x$} (llr);
	\draw[->] (lr) -- node[sloped, anchor=center,below] {$1-x$} (lrr);
	\draw[->] (rr) -- node[sloped, anchor=center,below] {$x$} (lrr);
	\draw[->] (rr) -- node[sloped, anchor=center,below] {$1-x$} (rrr);
	
	\draw[->] (lll) -- node[below] {$\nicefrac{1}{4}$} (T);
	\draw[->] (llr) -- node[left] {$\nicefrac{11}{12}$} (T);
	\draw[->] (lrr) -- node[right] {$\nicefrac{11}{12}$} (T);
	\draw[->] (rrr) -- node[below] {$\nicefrac{1}{4}$} (T);
\end{tikzpicture}
}
\subcaption{The reachability probability from source to target equals $f = 2x\cdot(1-x)+\frac{1}{4}$.}
\label{fig:handelmanex}
\end{subfigure}
\caption{Examples for the strong connection between polynomial (inequalities) and pMCs. Transitions to the sink are not depicted for conciseness.}

\end{figure}
\noindent Checking a bound on a given polynomial over $X$ thus is equivalent to
checking a bound on a reachability probability in a simple acyclic pMC over $X$.
For the fixed parameter case, this gives rise to the following equivalence
relating arbitrary pMCs to simple acyclic pMCs.
\begin{restatable}{theorem}{restateacyclicsuffices}
\label{thm:simpleacyclicsuffices}
For any non-simple pMC $\mdp$ with $\ParamSpace^{\mathrm{gp}}_\mdp = (0,1)^{\Params}$ there exists a simple acyclic pMC $\mdp'$ such that
	\[ 
  \{ \inst \in \ParamSpace^{\mathrm{gp}}_\mdp \mid \Pr_{\mdp[\inst]}(\lozenge T) \bowtie \lambda  \} = \{ \inst \in \ParamSpace^{\mathrm{gp}}_{\mdp'} \mid \Pr_{\mdp'[\inst]}(\lozenge T) \bowtie \lambda \}.\]
	In the fixed parameter case, $\mdp'$ can be computed in polynomial time. 
\end{restatable}
\noindent The proof is constructive: one first computes the (rational function) $\sol_\mdp$, reformulates that as a polynomial constraint, and casts that into a simple acyclic pMC using Lemma~\ref{lem:chonevtrick}.

The goal of the rest of this section is to prove a result which is, in a
sense, a stronger version of Lemma~\ref{lem:chonevtrick}. 
In particular, we want to describe polynomials by $\sol^\mdp$ for an acyclic pMC. 
 We call a polynomial
$f \in \QQ[\Params]$ \emph{adequate} if $0 < f[\inst] < 1$ for all
$\inst\colon \Params \rightarrow (0,1)$ and $0 \leq f[\inst] \leq 1$ for all
$\inst\colon \Params \rightarrow [0,1]$. Note that \emph{$\sol^\mdp$ is an adequate
polynomial if $\mdp$ is both simple and acyclic}, and there is no acyclic pMC $\mdp$ (with a single parameter) such that $\sol_\mdp$ is \emph{not} adequate---except where $\sol_\mdp=0$ or $\sol_\mdp=1$.

\begin{restatable}{theorem}{restatehandelmanconstr} \label{thm:handelmanconstr}
  Let $f \in \mathbb{Q}[x]$ be a (univariate) adequate polynomial. There
  exists a simple acyclic pMC $\mdp$ with a target state $T$ such that $f =
  \sol^\mdp$.
\end{restatable}
\noindent Our construction of a pMC for some adequate polynomial is based on the following result:

\begin{restatable}[Handelman's theorem \cite{handelman1988representing}]{lemma}{restatehandelmantheorem}
\label{lem:handelman}
  Let $\beta_1 \geq 0,\dots,\beta_\ell \geq 0$ be linear constraints that define a
  compact convex polyhedron $P \subseteq \RR^n$ with interior. If a polynomial
  $\pol \in \RR[x_1,\dots,x_n]$ is strictly positive on $P$, then $\pol$ may be
  written as
  \begin{equation}\label{eq:handelmanrepr} 
    \pol = \sum_{i=1}^k \lambda_ih_i
  \end{equation}
  where $h_i =
  \beta_1^{e_{i,1}}\cdot\ldots\cdot\beta_\ell^{e_{i,\ell}}$ for some natural exponents $e_{i,j}$ and real coefficients $\lambda_i > 0$.
\end{restatable}
\noindent
Form \eqref{eq:handelmanrepr} is called a \emph{Handelman representation}
of $\pol$ w.r.t.\ $\beta_1 ,\dots, \beta_\ell$. 
The next lemma states the existence of a specific Handelman
representation which we can map to a pMC.

\begin{restatable}{lemma}{restatebinomialrepr}
	\label{lem:binomialrep}
	Let $\pol \in \QQ[x]$ be strictly positive on $[0,1]$. There exists an $n \geq 0$ such that %$\pol$ can be written as
  \begin{equation} \label{eq:binomialrep} 
		\pol = \sum_{k=0}^n p_k\cdot {n \choose k}\cdot x^{n-k}\cdot (1-x)^k \qquad\text{with $p_k \in [0,1]$ for all $0 \leq k \leq n$}
	\end{equation}
\end{restatable}

\begin{example}
	\label{ex:binomialrep}
	Consider $\pol = 2x\cdot(1-x) + \frac{1}{4}$ which is strictly positive on $[0,1]$ and already in a Handelman representation. Following the proof of Lem.~\ref{lem:binomialrep}, we find that
	\[ 
	f = \frac{1}{4}{0 \choose 3}x^3 + \frac{11}{12}{1 \choose 3}x^2(1-x) + \frac{11}{12}{2 \choose 3}x(1-x)^2+ \frac{1}{4}{3 \choose 3}(1-x)^3
	\]
	%and thus 
	%$
	%	1-\pol= \frac{3}{4}{0 \choose 3}x^3 + \frac{1}{12}{1 \choose 3}x^2(1-x) + \frac{1}{12} {2 \choose 3}x(1-%x)^2+ \frac{3}{4}{3 \choose 3}(1-x)^3.
	%$
	\label{ex:polytolinearpmc}
	%Following Example \ref{ex:binomialrep}, 
	The construction (described in the proof of Thm.~\ref{thm:handelmanconstr}) yields the pMC depicted in Fig.~\ref{fig:handelmanex}.
\end{example}

\section{The complexity of reachability in pMCs}
We improve lower bounds for $\exists\reach$ problems. The results depend on the comparison type:

\subparagraph*{Nonstrict inequalities.}

This paragraph is devoted to proving the following theorem:

\begin{restatable}{theorem}{restatepmcetrhard}
\label{thm:etr:pmcsarehard}
	$\exists\reach^{\leq}_{*}$, 
	$\exists\reach^{\geq}_{*}$ are all ETR-complete (even for acyclic pMCs).
\end{restatable}
\begin{restatable}{definition}{restatembFEASdef}
  The decision problem \emph{modified-closed-bounded-4-feasibility} (mb4FEAS-c) asks:
  Given a (non-negative) quadric polynomial $f$, $\exists
  \inst\colon \Params \rightarrow [0,1]$	 s.t.\ $f[\inst] \leq 0$?  The \emph{modified-open-bounded-4-feasibility} (mb4FEAS-o) is analogously defined with $\inst$ ranging over $(0,1)$. 
\end{restatable}
\noindent This problem easily reduces to its $\geq$-variant by multiplying $f$
with $-1$.
\begin{restatable}{lemma}{restatembFEASetr}
\label{lem:etr:mb4feas}
  The problems mb4FEAS-c and mb4FEAS-o are ETR-hard.
\end{restatable}
\noindent Essentially, one reduces from the existence of common roots of quadratic polynomials lying in a unit ball, which is known to be ETR-complete~\cite[Lemma 3.9]{Schaefer2013}. 
The reduction to mb4FEAS follows the reduction\footnote{Essentially the polynomial $f$ in mb4FEAS is constructed by taking the sum-of-squares of the quadratic polynomials, and further operations are adequatly shifting the polynomial.} between unconstrained variants (i.e., variants in which the position of the root is not constrained) of the same decision problems~\cite[Lemma 3.2]{DBLP:journals/mst/SchaeferS17}.

\begin{remark}
  Observe that there may be exactly one satisfying assignment to mb4FEAS-o/c, which may be irrational.
  In contrast, if there exists a satisfying assignment for $f > 0$, then there exist infinitely many
  satisfying (rational) assignments. To the best of our knowledge, the complexity of a variant of mb4FEAS-o/c  with strict bounds is open. Therefore, we have no ETR-hardness for $\exists\reach$ with strict bounds.
  In general, \emph{conjunctions} of strict inequalities are also ETR-complete \cite{DBLP:journals/mst/SchaeferS17}. We exploit this in the proof of Thm.~\ref{thm:etr:mdpsarehard} on page~\pageref{thm:etr:mdpsarehard}.
  \end{remark}

\begin{proof}[Proof of Thm.~\ref{thm:etr:pmcsarehard}]
  The reduction from mb4FEAS-c to $\exists\reach^\leq_\mathrm{wd}$ is a straightforward
  application of Lemma~\ref{lem:chonevtrick} with $\mu = 0$ and $\lambda =
  \frac{1}{2}$.  For $\exists\reach^\leq_\mathrm{gp}$, we reduce from the open variant
  and notice that as the construction in Lemma~\ref{lem:chonevtrick} preserves
  all satisfying instantiations $\inst\colon X \rightarrow [0,1]$ it, in
  particular, also preserves them on the graph-preserving parameter
  space.  For $\geq$, we apply Lemma~\ref{lem:chonevtrick} on ${-}f$.	
\end{proof}

The tight complexity class shows that the assumption of simplicity is not a real restriction.
Furthermore, a similar construction can be used for (sufficiently large\footnote{The bounds should be at least $\delta$ apart, where $\delta$ requires at most single-exponentially many bits.}, linear) subsets of the parameter space. 
In particular, methods~\cite{DBLP:conf/atva/QuatmannD0JK16,DBLP:conf/atva/CubuktepeJJKT18} targeted at a variant of $\exists\reach$ considering a so-called $\epsilon$-preserving parameter space ($[\epsilon$, $1-\epsilon]^k$) target an ETR-complete problem.

\subparagraph*{Strict inequalities.}
In this paragraph, the main result is:
\begin{theorem}
	\label{thm:e_reach_gr_gp_np_hard}
	$\exists\reach_{\text{*}}^>$ and $\exists\reach_{\text{*}}^<$ are NP-hard.
\end{theorem}
The gadget in Fig.~\ref{fig:e_reach_gp_red_e_reach} ensures that for any graph non-preserving instantiation, the probability to reach the target is $0$, while it does not affect reachability probabilities for graph-preserving instantiations. Together with semi-continuity of the solution function, we deduce that assuming graph-preservation is equivalent to not making this assumption:
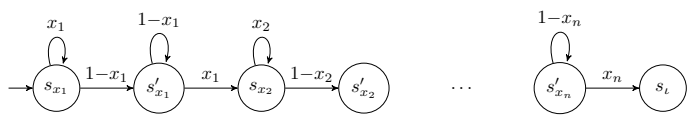
\begin{figure}
\centering
\scalebox{0.7}{
\begin{tikzpicture}[initial text=]
	\node[initial,state] (s1) {$s_{x_1}$};
	\node[state, right=of s1] (s2) {$s'_{x_1}$};
	\node[state, draw, right=of s2] (s3) {$s_{x_2}$};
	\node[state, right=of s3] (s4) {$s'_{x_2}$};
	\node[circle, right=of s4] (s5) {$\hdots$};
	\node[state, right=of s5] (s6) {$s'_{x_n}$};
	\node[state, right=of s6] (s7) {$\sinit$};
	
	\draw[->] (s1) edge node[auto] {$1{-}x_1$} (s2);
	\draw[->] (s2) edge node[auto] {$x_1$} (s3);
	\draw[->] (s3) edge node[auto] {$1{-}x_2$} (s4);
	%\draw[->] (s4) edge node {$p_2$} (s2);
	\draw[->] (s6) edge node[auto] {$x_n$} (s7);

	\draw[->] (s1) edge[loop above] node[auto] {$x_1$} (s1);
	\draw[->] (s2) edge[loop above] node[auto] {$1{-}x_1$} (s2);
	\draw[->] (s3) edge[loop above] node[auto] {$x_2$} (s3);
	\draw[->] (s6) edge[loop above] node[auto] {$1{-}x_n$} (s6);
	
\end{tikzpicture}	
}
\caption{Gadget for the proof of Lemma~\ref{lem:gpvswdequiv}}
\label{fig:e_reach_gp_red_e_reach}
\end{figure}
\begin{restatable}{lemma}{restategpnotnecessary}
	\label{lem:gpvswdequiv}
	There are polynomial-time Karp reductions among
	$\exists\reach_{\text{gp}}^>$ and $\exists\reach^>_\mathrm{wd}$.
\end{restatable}
\noindent We may thus turn our attention to well-defined parameter spaces: 
The decision problem
\[\exists\reach^{\geq1}_\mathrm{wd} \stackrel{\mathrm{def}}{\iff} \exists \inst \in \ParamSpace^{\text{wd}}_{\mdp}.\; \Pr_{\mdp[\inst]}(\lozenge T) \geq 1\]
is NP-complete~\cite[Thm.~3]{Chonev17}.
	A more refined analysis of the 3SAT-reduction yields:
\begin{restatable}{theorem}{restatechonevrefined}
	\label{thm:e_reach_gr_np_hard}
	$\exists\reach^{>}_\mathrm{wd}$ and $\exists\reach^{<}_\mathrm{wd}$ are NP-hard.
\end{restatable}
\begin{proof}[Proof of Thm.~\ref{thm:e_reach_gr_gp_np_hard}]
Lem.~\ref{lem:gpvswdequiv}, Thm.~\ref{thm:e_reach_gr_np_hard}, and Lem.~\ref{lem:restrict} together imply Thm.~\ref{thm:e_reach_gr_gp_np_hard}.
\end{proof}
This concludes our complexity analysis for pMCs.

\section{The complexity of reachability in pMDPs}
\subsection{Exists-exists reachability}
\label{sec:mdp:ee}
By definition, every pMC is a pMDP. Conversely, from any pMDP we can construct a
pMC such that their $\exists\exists\reach^{\bowtie}_\mathrm{wd}$ problems coincide. A similar construction
relates pMCs to the existence of optimal randomised memoryless strategies in
partially observable MDPs~\cite{DBLP:conf/uai/Junges0WQWK018}.

\begin{restatable}{lemma}{restatepmdpsarepmcs}
\label{lem:existwdgeeqexistwdge}
  There are polynomial-time Karp reductions among 
  $\exists\reach^{\bowtie}_\mathrm{wd}$ and $\exists\exists\reach^{\bowtie}_\mathrm{wd}$.
\end{restatable}
\noindent We outline the steps in Fig.~\ref{fig:pmdptopmc} and in the description below.

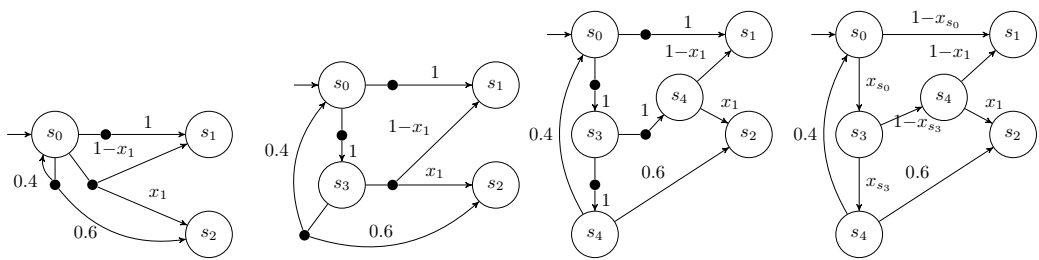
\begin{figure}
\centering
\begin{subfigure}[b]{0.24\textwidth}
\scalebox{0.7}{
\begin{tikzpicture}[baseline=(s0)]
\node[state, initial, initial text=] (s0) {$s_0$};
\node[state, right=2cm of s0] (s1) {$s_1$};
\node[state, below=1cm of s1 ] (s2) {$s_2$};
\node[fill, circle, inner sep=2pt, right=0.4cm of s0] (a0) {};
\node[fill, circle, inner sep=2pt, below=0.4cm of s0,xshift=2em] (a1) {};
\node[fill, circle, inner sep=2pt, below=0.4cm of s0,xshift=0em] (a2) {};

\draw[-] (s0) -- (a0);
\draw[-] (s0) -- (a1);
\draw[-] (s0) -- (a2);

\draw[->] (a0) edge node[auto] {$1$} (s1);
\draw[->] (a1) edge node[auto] {$1{-}x_1$} (s1);
\draw[->] (a1) edge node[auto] {$x_1$} (s2);	

\draw[->] (a2) edge[bend left] node[auto] {$0.4$} (s0);
\draw[->] (a2) edge[bend right] node[left,xshift=-2mm] {$0.6$} (s2);	
	
\end{tikzpicture}
}
\subcaption{simple pMDP}
\label{fig:simplepmdp}
\end{subfigure}
\begin{subfigure}[b]{0.24\textwidth}
\scalebox{0.7}{
\begin{tikzpicture}
	\node[state, initial, initial text=] (s0) {$s_0$};
\node[state, right=2cm of s0] (s1) {$s_1$};
\node[state, below=1cm of s1 ] (s2) {$s_2$};
\node[state, below=1cm of s0] (s3) {$s_3$};
\node[fill, circle, inner sep=2pt, right=0.4cm of s0] (a0) {};
\node[fill, circle, inner sep=2pt, below=0.4cm of s0,xshift=0em] (a2) {};
\node[fill, circle, inner sep=2pt, right=0.4cm of s3] (a3) {};	
\node[fill, circle, inner sep=2pt, below=0.4cm of s3,xshift=-2em] (a4) {};

\draw[-] (s3) -- (a3);
\draw[-] (s3) -- (a4);
\draw[-] (s0) -- (a0);
\draw[-] (s0) -- (a2);
\draw[->] (a3) edge node[auto] {$1{-}x_1$} (s1);
\draw[->] (a3) edge node[auto] {$x_1$} (s2);	
\draw[->] (a4) edge[bend left] node[auto] {$0.4$} (s0);
\draw[->] (a4) edge[bend right] node[auto] {$0.6$} (s2);	

\draw[->] (a0) edge node[auto] {$1$} (s1);
\draw[->] (a2) edge node[auto] {$1$} (s3);
\end{tikzpicture}
}
\subcaption{intermediate step}
\end{subfigure}
\begin{subfigure}[b]{0.24\textwidth}
\scalebox{0.7}{
\begin{tikzpicture}[baseline=(s0)]
	\node[state, initial, initial text=] (s0) {$s_0$};
\node[state, right=2cm of s0] (s1) {$s_1$};
\node[state, below=1cm of s1 ] (s2) {$s_2$};
\node[state, below=1cm of s0] (s3) {$s_3$};
\node[state, below=1cm of s3] (s4) {$s_4$};
\node[state, right=0.7cm of s3,yshift=2em] (s5) {$s_4$};

\node[fill, circle, inner sep=2pt, right=0.4cm of s0] (a0) {};
\node[fill, circle, inner sep=2pt, below=0.4cm of s0,xshift=0em] (a2) {};
	
\node[fill, circle, inner sep=2pt, right=0.4cm of s3] (a3) {};	
	
\node[fill, circle, inner sep=2pt, below=0.4cm of s3] (a4) {};

\draw[->] (s5) edge node[auto] {$1{-}x_1$} (s1);
\draw[->] (s5) edge node[auto] {$x_1$} (s2);	
\draw[->] (s4) edge[bend left] node[auto] {$0.4$} (s0);
\draw[->] (s4) edge node[auto] {$0.6$} (s2);	
\draw[-] (s3) -- (a3);
\draw[-] (s3) -- (a4);
\draw[-] (s0) -- (a0);
\draw[-] (s0) -- (a2);

\draw[->] (a3) edge node[auto] {$1$} (s5);
\draw[->] (a4) edge node[auto] {$1$} (s4);
\draw[->] (a0) edge node[auto] {$1$} (s1);
\draw[->] (a2) edge node[auto] {$1$} (s3);

\end{tikzpicture}
}
\subcaption{simple bin.-dec.}
\end{subfigure}
\begin{subfigure}[b]{0.24\textwidth}
\scalebox{0.7}{
\begin{tikzpicture}[baseline=(s0)]
	\node[state, initial, initial text=] (s0) {$s_0$};
\node[state, right=2cm of s0] (s1) {$s_1$};
\node[state, below=1cm of s1 ] (s2) {$s_2$};
\node[state, below=1cm of s0] (s3) {$s_3$};
\node[state, below=1cm of s3] (s4) {$s_4$};
\node[state, right=0.7cm of s3,yshift=2em] (s5) {$s_4$};

\draw[->] (s5) edge node[auto] {$1{-}x_1$} (s1);
\draw[->] (s5) edge node[auto] {$x_1$} (s2);	
\draw[->] (s4) edge[bend left] node[auto] {$0.4$} (s0);
\draw[->] (s4) edge node[auto] {$0.6$} (s2);	
\draw[->] (s3) edge node[below,pos=0.9] {$1{-}x_{s_3}$} (s5);
\draw[->] (s3) edge node[auto] {$x_{s_3}$} (s4);
\draw[->] (s0) edge node[auto] {$1{-}x_{s_0}$} (s1);
\draw[->] (s0) edge node[auto] {$x_{s_0}$} (s3);

\end{tikzpicture}
}
\subcaption{simple pMC}
\label{fig:simplepmc}
\end{subfigure}
\caption{From simple pMDP to simple pMC}
\label{fig:pmdptopmc}
\end{figure}

\subparagraph*{Binary-decision pMDPs.} The first step of the translation
consists in restricting the non-determinism resolved by a scheduler to (at most)
two options from every state. A binary-decision pMDP is a pMDP such that
$|\Act(s)| \leq 2$ for all
states $s \in S$ and if $|\Act(s)| = 2$ then
\(
  \forall \act \in \Act(s), \forall s' \in S,\; P(s,\act,s') \in
  \{0,1\}.
\)
Any pMDP can be transformed (in polynomial time) into a
binary-decision pMDP by introducing auxiliary states and simulating $k$-ary
non-deterministic choice using a binary-tree-like scheme in which all non-Dirac
transitions are pushed to the leaves (see,
e.g.,~\cite{DBLP:conf/qest/SegalaT05,
DBLP:conf/atva/QuatmannD0JK16,DBLP:conf/uai/Junges0WQWK018}). Such a
construction preserves simplicity.

\subparagraph*{From non-determinism to parameters.}
For a given binary-decision pMDP $\mdp$, we may replace all non-determinism
by parameters, inspired by \cite{DBLP:conf/nfm/HahnHZ11,DBLP:conf/uai/Junges0WQWK018}. We introduce fresh variables $\Params_S = \{ x_s \mid s \in S \}$.	
In $\mdp$, for any state $s$ with $\Act(s) = \{ \act, \act' \}$ we replace
\begin{itemize}
	\item the unique transition $P(s,\act,s') = 1$ by $P(s,\act,s') = x_{s}$
	\item the unique transition $P(s,\act',s') = 1$ by $P(s,\act',s') = 1-x_{s}$	.
\end{itemize}
The outcome is a simple pMC $\mdp'$. To translate instantiations into
schedulers, and vice versa, it is helpful to consider randomised schedulers.
Observe that, by Rem.~\ref{rem:detsched}, instantiations which translate into such
schedulers are always dominated by deterministic ones. 

\noindent Using the previously described construction,
we obtain the following.
\begin{restatable}{lemma}{restateequivsimplepmcs}
\label{lem:equivsimplepmc}
	For all simple pMDPs $\mdp$ one can construct in polynomial time
  a (linearly larger) simple pMC $\mdp'$ s.t.
  \[ 
  \left(
    \exists \inst \in
    \ParamSpace^{\mathrm{wd}}_\mdp,\exists \sched \in \Sched.\; \Pr_{\mdp[\inst]}^\sched(\lozenge T) \bowtie
    \frac{1}{2} \right)
    \iff
    \left(
     \exists \inst \in \ParamSpace^\mathrm{wd}_{\mdp'}.\;
    \Pr_{{\mdp'}[\inst]}(\lozenge T) \bowtie \frac{1}{2}
    \right) .
  \]
\end{restatable}

\begin{corollary}
	\label{cor:eeggpeqerggwd}
	%$\exists\exists\reach^>_\text{gp} \equiv_p \exists\reach^>$.
There are polynomial-time Karp reductions among 
  $\exists\exists\reach^>_\mathrm{gp}$ and $ \exists\reach^>_\mathrm{wd}$.
\end{corollary}

\begin{proof}
Minor adaptions in the proofs of Lemma~\ref{lem:existwdgeeqexistwdge} and Lemma~\ref{lem:gpvswdequiv}.	
\end{proof}

\subsection{Exists-forall reachability}
Contrary to pMCs, we obtain ETR-completeness in pMDPs for any comparison relation:
\begin{restatable}{theorem}{restatemdpsetrhard}
	\label{thm:etr:mdpsarehard}
	$\exists\forall\reach^{\bowtie}_{*}$ are all ETR-complete (even for acyclic pMDPs with a single non-deterministic state).
\end{restatable}

\noindent For the strict relations, we use a different problem to reduce from.
\begin{restatable}{definition}{restatebconINEQdef}
  The decision problem \emph{bounded-conjunction-of-inequalities} (bcon4INEQ-c) asks:
  Given a family of quadric polynomials $f_1,\hdots, f_m$, $\exists \inst:
  \Params \to  [0,1]$ s.t.\ $\bigwedge_{i=1}^m f_i[\inst] < 0$? The open variant
  (bcon4INEQ-o) can be defined analogously.
\end{restatable}
\noindent By a reduction from mb4FEAS (adapted from~\cite[Thm 4.1]{DBLP:journals/mst/SchaeferS17}):
\begin{restatable}{lemma}{restatebconetrhard}
\label{lem:etr:bcon4INEQhard}
	The bcon4INEQ-o/c problems are ETR-hard.
\end{restatable}

\begin{proof}[Proof sketch of Thm.~\ref{thm:etr:mdpsarehard}]
	ETR-hardness for non-strict inequalities follows from Thm.~\ref{thm:etr:pmcsarehard}. For strict inequalities, we reduce from bcon4INEQ-o/c: Generalise the construction from Thm.~\ref{thm:etr:pmcsarehard}: Build a pMDP as in Fig.~\ref{fig:mdpsareetrhard} with pMCs $\mdp(f_i)_{\lambda=\nicefrac{1}{2},\mu=0}$ created by Lemma~\ref{lem:chonevtrick}. 
\end{proof}

\begin{figure}
\centering
\begin{tikzpicture}
	\node[circle, draw, initial, initial text=] (sinit) {};
	\node[below=1.3cm of sinit] (dots) {$\dots$};
	\node[rectangle, left=0.7cm of dots, draw, inner sep = 6pt] (f1) {\footnotesize$\mdp(f_1)_{\lambda=\nicefrac{1}{2},\mu=0}$};
	\node[rectangle, right=0.7cm of dots, draw, inner sep = 6pt] (fm) {\footnotesize$\mdp(f_m)_{\lambda=\nicefrac{1}{2},\mu=0}$};
	\node[circle, fill, inner sep = 2pt, above=0.5cm of f1] (a1) {};
	\node[circle, fill, inner sep = 2pt, above=0.5cm of fm] (am) {};
	
	\draw[->] (a1) edge  node[right] {\scriptsize$1$} (f1);
	\draw[->] (am) edge  node[right] {\scriptsize$1$} (fm);
	\draw[->] (sinit) edge  node[right] {} (a1);
	\draw[->] (sinit) edge  node[right] {} (am);

\end{tikzpicture}
\caption{Construction for the proof of Thm.~\ref{thm:etr:mdpsarehard}}
\label{fig:mdpsareetrhard}
\end{figure}
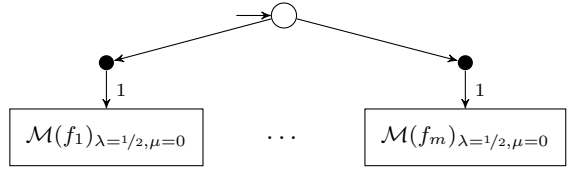

\paragraph*{Relation to stochastic games}
%In this section, 
We will now argue that
%we show that 
pMDPs are---in a sense---a generalisation of
\emph{Concurrent Stochastic Reachability Games} (CSRG), a model which has been
extensively studied~\cite{Sha53,DBLP:journals/tcs/AlfaroHK07,survey-CSGs,DBLP:conf/lics/HansenKM09,chi17}. We use this to establish more fine-grained results about
the bit-complexity of the parameter synthesis problem below.

\subparagraph*{Playing a stochastic game.} A CSRG is a two-player game $\csrg$
played on a finite set $S$ of states. The objective of player $\playerI$
%(\emph{reachability} player)
is to reach a target states $T \subseteq S$ 
while player $\playerII$
%(\emph{safety player})
has to avoid ever reaching a state in $T$.
A play of $\csrg$ begins in an initial state $\csrginit$ and proceeds as
follows: In state $s$, both players $\playerI$ and $\playerII$
\emph{concurrently} select an action $a \in A_s$ (resp.\ $b \in B_s$), the finite
set of actions available to player $\playerI$ (resp.\ $\playerII$) in state $s$.
The game then picks a successor state $s'$ according to a fixed 
probability distribution $P(\cdot | s,a,b)$ over $S$, and the
play continues in $s'$. The transition from $s$ to $s'$ is called a
\emph{round} of $\csrg$. Player $\playerI$ wins $\csrg$ once a state in
$T$ is reached. Otherwise, if a target is never reached, then $\playerII$ wins. 

%\subparagraph*{Strategies.} 
A \emph{strategy} $\sigma$ of a player is, in essence, a scheduler. However,
strategies in a CSRG map state-action sequences $s_0 (a_0,b_0) \dots s_{k-1}
(a_{k-1},b_{k-1}) s_k$
to a probability distribution over the actions $A_{s_k}$
(resp.\ $B_{s_k}$) available in the current state $s_k$. We call $\sigma$ a
\emph{stationary} strategy if it does not depend on the history but only on the
current state, i.e. it is a randomised memoryless scheduler.
Let $\Sigma^i$ denote the set of stationary strategies for player $i \in
\{\playerI, \playerII\}$.

\subparagraph*{Instantiations and MDPs.} The \emph{instantiation} $\csrg^\sigma$
of $\csrg$ with a stationary strategy for player $\playerI$ is the structure
obtained by forcing player $\playerI$ to follow $\sigma$. Notice that
$\csrg^\sigma$ is a finite MDP $\mdp$. Its transition probability function
$P_\mdp$ is obtained by letting
\begin{equation} \label{eq:csrginst} 
	P_\mdp(s,b,s') = \sum_{a \in A_s} \sigma(a | s)P(s' | s,a,b)
\end{equation}
for all $s, s' \in S$ and actions $b \in B_s$ of player $\playerII$.
(Instantiations are defined completely symmetrically for strategies of player
$\playerII$.) Conversely, every MDP may be viewed as a CSRG where $|A_s| = 1$
(or $|B_s| = 1$) for all $s \in S$, i.e.\ one of the players does never have any
choice.

\subparagraph*{Value of a CSRG.} Let $\Pr_\csrg^{\sigma, \tau}(\lozenge T)$
be the probability that $T$ is reached if player $\playerI$ plays according to $\sigma$
and player $\playerII$ according to $\tau$. The \emph{value} of $\csrg$ is defined as follows
%The player-$\playerI$-value (resp.
%player-$\playerII$-value) is defined as
\[
  V(\csrg)%_\playerI(\csrg)
  \coloneqq \sup_{\sigma}
  \Inf_{\tau} \Pr_\csrg^{\sigma,
  \tau}(\lozenge T)
\]
where the $\sup$ and $\inf$ range over all strategies of both players
respectively.
Intuitively, it is the maximal winning probability of player~$\playerI$ that can be
guaranteed against all strategies of  player~$\playerII$.
The existence of stationary optimal strategies for player~$\playerII$~\cite{partha71,partha76}
%as stated in Thm. \ref{thm:csrgprops}
allows us to encode a CSRG in a pMDP as we will show next.
\begin{restatable}{theorem}{restatecsrgencoding}
	\label{thm:csrgencoding}
	For any given CSRG $\csrg$, there exists a simple pMDP $\mdp$ such that
	\[
    V(\csrg) = 
	\min_{\tau \in \Sigma^\playerII}\max_{\sigma \in \Sigma^\playerI}\Pr_\csrg^{\sigma, \tau}(\lozenge T)
	=
	\min_{\inst \in \ParamSpace^{\mathrm{wd}}}\max_{\sched \in \Sched}\Pr_{\mdp[\inst]}^\sched(\lozenge T)
	\]
	and $\mdp$ can be computed in polynomial time (in the size of $\csrg$).
\end{restatable}
\noindent As a direct consequence, we obtain CSRG-hardness.
\begin{restatable}{corollary}{restatecsrghardness} \label{thm:csrghardness}
  Determining whether $V(\csrg) \unlhd \lambda$, for $\lambda \in
  \mathbb{Q}$, reduces to $\exists\forall\reach^{\unlhd}_\mathrm{wd}$.
\end{restatable}

\noindent It follows from \cite[Thms.\ 6 and 12]{chi17} that optimal rational
instantiations may be complex.
\begin{theorem}
  There are pMDPs for which rational optimal and $\varepsilon$-optimal parameter
  values minimising the value $\max_{\sched \in
  \Sched}\Pr_{\mdp[\inst]}^\sched(\lozenge T)$ require exponentially-many bits
  to be written as a binary-encoded integer-pair.
\end{theorem}

\section{Robust reachability}
In this section, we briefly consider pMDPs in which we focus on obtaining (robust) schedulers rather than (robust) parameter values: 
We swap the quantification order from the $\quant_1 \quant_2 \reach$ problem. Intuitively, we ask
whether some scheduler gives guarantees on the maximal or minimal probability of
all instantiations of the pMDP eventually reaching $T$. 
Formally, for each $\quant_1,\quant_2 \in \{\exists,\forall\}$ and
$\mathrm{\bowtie} \in
\{\leq, <, >, \geq\}$, let
\[
  \quant_1\quant_2\robreach^{\bowtie}_\mathrm{wd} \stackrel{\mathrm{def}}{\iff}
  \quant_1 \sched \in \Sigma,\;\quant_2 \inst \in \ParamSpace^{\mathrm{wd}}.\;
    \Pr_{\mdp[\inst]}^\sched(\lozenge T) \bowtie \frac{1}{2}.
\]
We adopt the same conventions as for the $\quant_1\quant_2\reach$ problem when
considering graph-preserving instantiations.
Variants which use the same quantifier twice, or consider pMCs yield the same results as for $\reach^{\bowtie}$, and are therefore omitted.

Robust strategies have been widely studied in the field of operations
research (see, e.g.,~\cite{ng05,wkr13}) and are the main focus of \emph{reinforcement
learning}~\cite{sb98}. It is known
that the robust-reachability problem as defined above is not the most
general question one can ask. Indeed, we restrict our attention to
\emph{memoryless schedulers} while,
in general, optimal robust schedulers require
memory and randomisation~\cite{ArmingBS17}. 

Our interest in the robust-reachability problem is twofold. First,
it naturally corresponds to the quantifier-swapped version of the reachability
problem. Second, memoryless schedulers are desirable in practice for
their comprehensibility and ease of implementation.% to engineers.

\begin{restatable}{theorem}{restatefprobreachnpcomp}
\label{thm:fp_e_a_robreach_np_compl}
  In the fixed parameter case, $\exists\forall\robreach^{\nicefrac{<}{>}}_{{*}}$ are NP-complete. NP-hardness holds even for acyclic pMDPs with a single parameter.
\end{restatable}
\begin{proof}[Proof sketch]
Membership in NP is analogous to Lem.~\ref{lem:fp_ee_reach_in_np}. 
NP-hardness is based on a reduction from 3-SAT, with a construction similar to Fig.~\ref{fig:mdpsareetrhard}. \end{proof}

\begin{proposition}
	The decision problems
  $\exists\forall\robreach^{\bowtie}_{*}$   are  NP-hard and coNP-hard, and in PSPACE. For non-strict inequalities, the problems are coETR-hard.
\end{proposition}
\begin{proof}
	NP-hardness follows from Thm.~\ref{thm:fp_e_a_robreach_np_compl},
	coNP/coETR-hardness follows from Thm.~\ref{thm:e_reach_gr_gp_np_hard}, Thm.~\ref{thm:e_reach_gr_np_hard}, and Thm.~\ref{thm:etr:pmcsarehard}, respectively.
	Iterating over all (finitely many) schedulers, check each scheduler in ETR or in coETR (and thus in PSPACE).
\end{proof}
Consequently, it is unlikely that either of the problems are in ETR or  coETR, as then ETR and coETR would coincide (which is not impossible, but unlikely~\cite{DBLP:journals/mst/SchaeferS17}).

\section{Conclusions}
We have studied the complexity of various reachability problems for
\emph{simple} pMCs and pMDPs. All the problems we have considered are easily
seen to be solvable in PSPACE via reductions to the existential theory of the
reals. We have complemented this observation with lower bounds, i.e.  ETR
hardness for several versions of the problem both for pMCs and pMDPs.
These lower bounds naturally extend to general pMCs and pMDPs.

We
have given an NP decision procedure for pMDPs with a fixed number of parameters. 
The exact complexity of pMDP reachability problems with this
restriction remains open, and our  upper bounds do not straightforwardly generalise beyond simple pMDPs (see Rem.~\ref{rem:fixedparam}).

Finally, we have established a tight connection between polynomials and pMCs (even beyond~\cite{Chonev17}).
However, our results do not allow us to conclude whether there always are
``small'' pMCs for every polynomial. Such a result would provide more evidence
of ETR being the right framework to solve problems for our parametric models.

\clearpage
\bibliography{literature}

\clearpage
\appendix
\section{Full proofs}
The subsections reflect sections in the main paper.

\subsection{Introduction}
No further proofs.

\subsection{Preliminaries}
In the sequel, for decision problems $A$ and $B$,
we write $A \leq_p B$ to denote the fact that there exists a polynomial-time
computable Karp reduction from $A$ to $B$. If $A \leq_p B$ and $B \leq A_p$
then we write $A \equiv_p B$.

\subsection{Problem landscape}
\subsubsection{Proof of Lemma~\ref{lem:restrict}}

\restaterestrictcomparisons*

\begin{proof}
  We prove only the first item for arbitrary $\quant_1$ and
  $\quant_2 = \exists$. All other cases
  can be proven analogously.
  First, we deduce from \cite[Thm.~10.122 and Thm.~10.127]{BK08} that in
  polynomial time, and without regarding the actual transition probabilities, we
  can compute from $\mdp$ and a target set $T$, a target set\footnote{Which is some
  adequate union of particular maximal end components in $\mdp$}  $T'$ such that 
  \begin{equation}\label{eq:maxeqmininv}
    \max_{\sched \in \Sched} \Pr_{\mdp}^\sched(\lozenge T) = 1 - \min_{\sched
    \in \Sched} \Pr_{\mdp}^\sched(\lozenge T').
  \end{equation}
  Please observe that the step above in general does not work without the restriction to graph-preserving instantiations.
  And combine this to obtain
  \begin{align*}
    \quant_1 \inst \in \ParamSpace^{\mathrm{gp}},\;
    \exists \sched \in \Sched.\;
    \Pr_{\mdp[\inst]}^\sched(\lozenge T) > \frac{1}{2}
    & \iff 
    \quant_1 \inst \in \ParamSpace^{\mathrm{gp}}.\;
    \max_{\sched \in \Sched}
    \Pr_{\mdp[\inst]}^\sched(\lozenge T) > \frac{1}{2}\\
    & \overset{\eqref{eq:maxeqmininv}}{\iff}
    \quant_1 \inst \in \ParamSpace^{\mathrm{gp}}.\;
    \left(1 - \min_{\sched \in \Sched}
          \Pr_{\mdp[\inst]}^\sched(\lozenge T')\right) > \frac{1}{2}\\
    & \iff
    \quant_1 \inst \in \ParamSpace^{\mathrm{gp}}.\;
    \min_{\sched \in \Sched} 
    \Pr_{\mdp[\inst]}^\sched(\lozenge T') < \frac{1}{2}\\
    & \iff
    \quant_1 \inst \in \ParamSpace^{\mathrm{gp}},
    \exists \sched \in \Sched.\;
    \Pr_{\mdp[\inst]}^\sched(\lozenge T') < \frac{1}{2}.
  \end{align*}
\end{proof}

\subsubsection{Proof of Theorem~\ref{thm:semicontinuous}}
\restatesemicontinuous*

The main intuition behind the argument presented herein is the following: for
instantiations which create more transitions labelled by probability $0$, the
set of states which can no longer reach the target should increase. In all other
cases, the solution function will be continuous since it is a rational function
of the parameter values.

\begin{proof}
Continuity on $\ParamSpace^\mathrm{gp}_\mdp$ follows from the fact that $\sol^\mdp$ is a rational function which is bounded by $[0,1]$ on all well-defined points~\cite{DBLP:conf/atva/QuatmannD0JK16}.

To prove semi-continuity, we use the following formal definition:
A function $f$ is \emph{lower semi-continuous} at an instantiation $\inst$ if $\forall \epsilon > 0$ there exists a neighbourhood $U$ of $\inst$ s.t. $f[\inst] \leq f[\inst'] + \epsilon$ for all $\inst' \in U$.

%Let $P_{=0}^\inst = \{ (s,s') \in S\times S \mid P(s,s') \neq 0 \land P(s,s')[\inst] = 0 \}$ define the transitions of pMC $\mdp$ that become $0$ under the instantiation $\inst$.
%For any two $\inst, \inst'$ in a graph-consistent region, $P_{=0}^\inst = P_{=0}^{\inst'}$.
%For any graph-preserving $\inst$, $P_{=0}^\inst = \emptyset$.
Let $S^\inst_{=0} = \{ s \in S \mid \sol^\mdp_s[\inst] = 0 \}$ denote the set of states that reach the target with probability zero in instantiation $\inst$.	
We make the following two  observations:
\begin{itemize}
\item For any $\inst, \inst'$, $P_{=0}^\inst \subseteq P_{=0}^{\inst'}$ implies $S_{=0}^\inst \subseteq S_{=0}^{\inst'}$, and $P_{=0}^\inst = P_{=0}^{\inst'}$ implies $S^\inst_{=0} = S^{\inst'}_{=0}$.
Essentially, removing transitions may cut states from having a path to the target states, but never adds new paths. 
\item Any instantiation $\inst$ has a neighbourhood $U$ s.t.\ $\forall \inst'
  \in U$, $P_{= 0}(\inst') \subseteq P_{= 0}(\inst)$.
  Essentially, each $\inst$ has a neighbourhood in which fewer parameters are assigned to $0$ or $1$ respectively (which are the only values which lead to transitions disappearing).
\end{itemize}
Thus:
%\begin{align*}
% \label{eq:neighbourhoodsszero}
$\text{for all }\inst \text{ there exists a neighbourhood }U 
\text{ of }\inst\text{ s.t. }
\forall \inst' \in U.\; S_{= 0}^\inst \subseteq S_{= 0}^{\inst'}.$
% \end{align*} 
 
 For each $\inst$, $S$ can be partitioned into:
 $S = T  \uplus S_{=0}^\inst  \uplus S'$.
For conciseness, first assume that the pMC is acyclic.
We use structural induction over the graph of the pMC.
\begin{itemize}
\item For $s \in T, \sol_s^\mdp$ is constant and thus continuous.
\item For $s \in S^\inst_{=0}$, $\sol_s^\mdp[\inst] = 0$ and $\sol_s^\mdp$ is non-negative, so it is certainly lower semi-continuous on neighbourhood $U$.
\item For $s \in S'$, $\sol_s^\mdp = \sum_{s'\in S} P(s,s') \cdot \sol_{s'}^\mdp$. By induction, $\sol_{s'}^\mdp$ is lower semi-continuous for each $s' \in S$. Then, the (weighted) sum of lower semi-continuous functions is also lower semi-continuous on $U$.
\end{itemize}

For cyclic pMCs, observe that in each strongly connected component (SCC), either all states have probability zero to reach a target, or none.
Then, apply the structural induction on the level of SCCs (which, for graph-preserving instantiations are preserved). More formally,
one can consider a
new pMC in which the SCCs have been contracted. The new pMC is acyclic.
Additionally, it is straightforward to prove that SCC-contraction
preserves reachability
probabilities.

To prove continuity on acyclic MCs, consider again the  structural induction. 
As we are no longer interested in extending this to cyclic MCs, we do not consider $S_{=0}$ states.
 For each $\inst$, $S$ then can be partitioned into:
 $S = T \uplus S'$.
The structural induction simplifies to:
\begin{itemize}
\item For $s \in T, \sol_s^\mdp$ is constant and thus continuous.
\item For $s \in S'$, $\sol_s^\mdp = \sum_{s'\in S} P(s,s') \cdot \sol_{s'}^\mdp$. By induction, $\sol_{s'}^\mdp$ is continuous for each $s' \in S$. Then, the (weighted) sum of continuous functions is also continuous on $U$.
\end{itemize}
Observe that this cannot be applied to general SCCs, as there the equation system used above has only a unique fixed point if $S_{=0}$ states are explicitly set to zero.

\end{proof}

\subsection{Fixing the number of parameters}

\subsubsection{Proof of Theorem~\ref{thm:fp_ea_reach_in_np} and Lemma~\ref{lem:fp_check_sched_opt_in_np}}
\restatefpchecksched*
\begin{proof}
Following Rem. \ref{rem:fixedparam}, $\ParamSpace^\text{wd}$ can be partitioned
into constantly many graph-preserving regions $R$ (in fact, $3^{|\Params|}$
many). Those regions are checked one-by-one as follows:
  The rational functions $g_s/h_s$ are computed with respect to $R$, which takes
  at most polynomial time \cite{baiercomplexity}. For somewhere minimal strategies, formula $\Psi$ from
  \eqref{eq:optimalstratparametric} is thus of polynomial size and has only a
  fixed number of variables. Hence it can be checked for satisfiability in
  polynomial time. The overall procedure returns true if at least one of the
  checks was sat. For the other three cases (somewhere maximal and everywhere
  minimal/maximal), a similar procedure applies.
\end{proof}
%%%%%%%%%%%%%%%%%%%%%%%%%%%%%%%%%%%%%%%%%%%%%%%%%%%%%%%%%%%%%

\restatefpexfainnp*
\begin{proof}
	We only give the proof for the $\geq$-relation, the other cases are analogous. Observe that
	\begin{align*}
		\exists \inst \in \ParamSpace^\text{wd}, \forall \sched \in \Sched.\; \Pr_{\mdp[\inst]}^\sched(\lozenge T) \geq \frac{1}{2} \Longleftrightarrow \exists \inst \in \ParamSpace^\text{wd}.\; \min_{\sched \in \Sched} \Pr_{\mdp[\inst]}^\sched(\lozenge T) \geq \frac{1}{2}
	\end{align*}
  which means that it is sufficient and necessary for the
    answer to the problem to be positive that there be a somewhere optimal
    strategy which (for the valuation for which it is minimal, it) 
  simultaneously induces a reachability probability at least $\frac{1}{2}$.
  Hence we may guess a somewhere minimal scheduler and check its minimality using the encoding from Lem.
  \ref{lem:fp_check_sched_opt_in_np} with a conjunction that the initial state satisfies the bound.
\end{proof}

\subsubsection{Proof of Lemma~\ref{lem:conpconjecture}}
\restateconpconjecture*
\begin{proof}
Consider the complement of $\exists\forall\reach^{\bowtie}_{*}$, that is $\forall\exists\reach^{\overline{\bowtie}}_{*}$. It quantifies over all schedulers. If a minimal optimal scheduler set is only polynomially large, then we guess such a set in polynomial time, verify that it is an OSS (using Thm.~\ref{thm:fixpar-checkoss}) and then check the induced pMC under each policy (which is in P).
The problem $\exists\exists\reach^{\bowtie}_{*}$ can be verified similarly, also taking into account the proof of Thm.~\ref{thm:fp_ea_reach_in_np}. 
\end{proof}

\subsubsection{Proof of Lemma \ref{lem:exponentialmoss}}

\restateexponentiallmoss*
\begin{proof}
	Define $\mdp_n = (S_n, \Params_n, \Act, \sinit, P_n)$ as
	$S_n = \{s_1,...,s_{n+1}\}~\uplus~\{\bot\}$,
	$\Params_n = \{x_1,...,x_n\}$, $\sinit = s_1$,
	$\Act = \{a,b\}$, and
	\[P_n(s_i,a,s_{i+1}) = x_i,~~P_n(s_i,b,s_{i+1}) = 1-x_i\]
	for all $1 \leq i \leq n$, and $T = \{ s_{n+1} \}$.
	The remaining unspecified probability mass is directed to the sink $\bot$. Let us fix a $\sched \in \Sched$ and let $\inst \in \ParamSpace^\mathrm{wd}$ be such that
	\[\inst(x_i) = \begin{cases}
		1, \text{ if } \sigma(s_{i}) = a, \\
		0 \text{ if } \sigma(s_{i}) = b.
	\end{cases}
	\]
	Then clearly $\Pr_{{\mdp_n}[\inst]}^{\sigma}(\lozenge T) = 1$ and $\Pr_{{\mdp_n}[\inst]}^{\sigma'}(\lozenge T) = 0 < 1$ for all $\sched' \neq \sched$. Thus we have shown that for all schedulers $\sched$, there is an instantiation for which  $\sched$ is the unique maximal scheduler, so no such $\sched$ may be excluded from a MOSS. The lemma follows as there are $|\Sched| = 2^n$ many schedulers $\sched$.
	
\end{proof}

\subsection{Expressiveness of simple pMCs}

\subsubsection{Proof of Lemma \ref{lem:chonevtrick}}
\restatechonevstrick*
	\begin{proof}
	Following \cite{Chonev17},  observe that a monomial $-x_1 \cdot \ldots \cdot x_d$ of degree $d \geq 0$ can be written as
	\begin{align}
		\label{eq:chonevmonomialtrick}
		%-\prod_{i=1}^n x_i = \sum_{i=1}^n\bigg( (1-p_i) \prod_{j = i+1}^n p_j \bigg) - 1 \\
		-x_1 \cdot \ldots \cdot x_d = {-}1 + \sum_{i=1}^d(1-x_i) \cdot x_{i+1} \cdot \ldots \cdot x_d ,
	\end{align}
	which is readily proved by induction on $d$: For $d = 0$, both sides are ${-}1$ (with the convention that an empty product equals ${-}1$). For $d \geq 0$, we multiply both sides of \eqref{eq:chonevmonomialtrick} by $x_{d+1}$ to obtain
	\begin{align*}
		-x_1 \cdot \ldots \cdot x_d \cdot x_{d+1} &= {-} x_{d+1} + \sum_{i=1}^{d}(1-x_i) \cdot x_{i+1} \cdot \ldots \cdot x_d \cdot x_{d+1}  \\
		&= (1 - x_{d+1}) - 1  + \sum_{i=1}^{d}(1-x_i) \cdot x_{i+1} \cdot \ldots \cdot x_d \cdot x_{d+1} \\
		&= {-1} + \sum_{i=1}^{d+1}(1-x_i) \cdot x_{i+1} \cdot \ldots \cdot x_d \cdot x_{d+1}.
	\end{align*}	Hence applying \eqref{eq:chonevmonomialtrick} to every term of $f$, we can write for some $m \geq 0$ 
	\begin{align}
		\label{eq:chonevtrickintermediate}
		f = \sum_{i=1}^m\alpha_i \cdot g_i + \beta
	\end{align}
	 where the $\alpha_i$ are \emph{positive} rational coefficients, the $g_i$  are nonempty products of terms from $\{x, (1-x) \mid x \in X\}$ and $\beta \in \mathbb{Q}$ is a constant term. We may assume that $\beta \leq 0$, otherwise $\beta = \beta \cdot x + \beta \cdot (1-x)$ for any $x \in \Params$ and we may pull $\beta$ inside the sum \eqref{eq:chonevtrickintermediate}. Let $N$ be an integer such that $N > max\{\sum_{i=1}^m \alpha_i, \mu - \beta\}$ and let $\tilde{f} \coloneqq \frac{f - \beta}{N}$. We have that
	\begin{align}
		f[\inst] \bowtie \mu 
		\iff \tilde{f}[\inst] \bowtie \frac{\mu - \beta}{N} \eqqcolon \lambda'.
	\end{align}
	for all valuations $\inst\colon \Params \rightarrow \RR$ and all comparison relations $\bowtie$. The new threshold $\lambda'$ is strictly between $0$ and $1$. The modified polynomial $\tilde{f}$ naturally corresponds to a simple acyclic pMC $\tilde{\mdp}$ with $\Pr_{\tilde{\mdp}}(\lozenge T) = \tilde{f}$ as shown in Fig.~\ref{fig:chonevtrick}.
	\begin{figure}
	\centering
	\scalebox{0.7}{
	\tikzset{decoration={snake,amplitude=.4mm,segment length=2mm,
                       post length=0mm,pre length=0mm}}
	\begin{tikzpicture}
		\node[state,initial,initial text=s,scale=0.6] (s0)  {};
		\node[right= of s0] (dots) {$\vdots$};
		\node[state,above=of dots,scale=0.6] (s1) {};
		\node[state,below=of dots,scale=0.6] (sm) {};
		\node[state,right=1.5cm of s1,scale=0.6,accepting]    (v1) {};
		\node[state,right=1.5cm of sm,scale=0.6,accepting]    (vm) {};
		
		\draw[->] (s0) edge node[above,xshift=-2mm] {$\frac{\alpha_1}{N}$} (s1);
		\draw[->] (s0) edge node[below,xshift=-2mm] {$\frac{\alpha_m}{N}$} (sm);
		
		\draw[->,decorate] (s1) -- node[above] {$g_1$} (v1);
		\draw[->,decorate] (sm) -- node[above] {$g_m$} (vm);
		
	\end{tikzpicture}	
	}
	\caption{The essential construction of the pMC in Lemma~\ref{lem:chonevtrick}: Any probability mass not drawn goes to a sink. }
	\label{fig:chonevtrick}
	\end{figure}
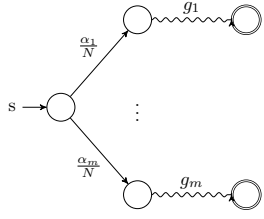

	We now modify $\tilde{\mdp}$ as outlined in Rem.~\ref{rem:fixedlambda} to obtain a pMC $\mdp$ such that for all $\inst \in \ParamSpace^{\mathrm{wd}}_\mdp$ it holds that for the given threshold $0 < \lambda < 1$

	\begin{align*}
		\Pr_{\tilde{\mdp}[\inst]}(\lozenge T) \bowtie \lambda' \iff \Pr_{\mdp[\inst]}(\lozenge T) \bowtie \lambda.
	\end{align*}
		
	For the complexity of the construction, notice that $m$ in the sum \eqref{eq:chonevtrickintermediate} is in $\mathcal{O}(td)$ where $t$ and $d$ are bounds on the total degree and the number of terms of $f$, respectively. The $g_i$ are products of at most $d$ terms. The $\alpha_i$ are the absolute values of the original coefficients of $f$ and $\beta$ is the sum of at most $t$ of those coefficients. Hence the $\alpha_i$, $\beta$, $N$ and the polynomial $\tilde{f}$ can be computed in time $O(poly(t,d,\kappa))$. The same then also holds for the pMC $\tilde{\mdp}$ and the final pMC $\mdp$ adapted to the desired threshold.
\end{proof}	

%%%%%%%%%%%%%%%%%%%%%%%%%%%%%%%%%%%%%%%%%%%%%%

\subsubsection{Proof of Theorem~\ref{thm:simpleacyclicsuffices}}
\restateacyclicsuffices*
\begin{proof}
Let $\mdp$ be a nonsimple pMC.
     Compute the solution function $\sol^\mdp$, which is a rational function, say $\frac{h}{g}$.  
     We make the following two observations:
     \begin{itemize}
     \item 
      by definition, $\sol^\mdp$ describes a probability, thus $\forall \inst \in \ParamSpace^{\mathrm{gp}}_\mdp.\; g[\inst] \neq 0$, and
     \item  $g$ is continuous by Thm.~\ref{thm:semicontinuous}, and $\ParamSpace^{\mathrm{gp}}_\mdp$ is convex. Thus either $\forall \inst \in \ParamSpace^{\mathrm{gp}}_\mdp.\; g[\inst] > 0$ or $\forall \inst \in \ParamSpace^{\mathrm{gp}}_\mdp.\; g[\inst] < 0$. Which case applies can be easily evaluated.
     \end{itemize}
Now fix some $\inst \in \ParamSpace^{\mathrm{gp}}_\mdp$.     
     W.l.o.g., assume $g[\inst] > 0$.
     Let $f = h - (g-1)\cdot \lambda$. By construction $\Pr_{\mdp[\inst]}(\lozenge T) = \sol^\mdp[\inst] \bowtie \lambda$ iff $f[\inst] \bowtie \lambda$.
     By applying Lemma~\ref{lem:chonevtrick} on $f \bowtie \lambda$, there exists a simple acyclic   
     pMC $\mdp'$ and $\lambda' \in \QQ$ such that $\Pr_{\mdp'[\inst]}(\lozenge T) \bowtie \lambda'$\footnote{Closer inspection yields that $\lambda' = \lambda$.}.
     The case $g[\inst] < 0$ is analogous.

  For a fixed number of parameters, this transformation  has polynomial time
  complexity since the solution function becomes computable in polynomial time
  and the pMC from Lemma~\ref{lem:chonevtrick} is constructible in polynomial
  time too.
\end{proof}

\subsubsection{Proof of Theorem \ref{thm:handelmanconstr} and Lemma \ref{lem:binomialrep}}

\restatehandelmantheorem*
\restatebinomialrepr*

\begin{proof}
Both $\pol$ and $1-\pol$ are strictly positive on the (1 dimensional) polyhedron defined by the constraints $x \geq 0, 1-x \geq 0$. Applying Lemma \ref{lem:handelman} yields:
	\begin{align}
		\label{eq:eqfromhandelman}
		\pol = \sum_{i=1}^m\lambda_ix^{e_i}(1-x)^{d_i} \text{ and } 1 - \pol = \sum_{i=1}^{m'}\lambda_i'x^{e_i'}(1-x)^{d_i'}
	\end{align}
	with $\lambda_i > 0$ ($\lambda_i' > 0$, resp.), and pairwise distinct $(e_i,d_i)$ ($(e_i',d_i')$, resp.) for all $1\leq i \leq m$ ($1\leq i \leq m'$, resp.).
	Let $n = \max\{e_1 + d_1,...,e_m+d_m,e_1'+d_1',...,e'_{m'}+d'_{m'}\}$, i.e. $n$ is the maximum degree of a term appearing in one of the two forms \eqref{eq:eqfromhandelman}. Using that
	\[
	1 = 1^l = (x + (1-x))^l = \sum_{k=0}^l{l \choose k}x^{l-k}(1-x)^k
	\]
	for all $l\geq 0$ and with the convention ${l \choose k} = 0$ for $k \notin \{0,...,l\}$, we transform \eqref{eq:eqfromhandelman} as follows:
	\begin{align*}
		\pol &= \sum_{i=1}^m\lambda_ix^{e_i}(1-x)^{d_i} \\
		&= \sum_{i=1}^m\bigg[\underbrace{\sum_{k=0}^{n-e_i-d_i}{n-e_i-d_i \choose k}x^{n-e_i-d_i-k}(1-x)^{k}}_{=1}\bigg]\lambda_ix^{e_i}(1-x)^{d_i} & \\
		&= \sum_{i=1}^m\lambda_i\sum_{k=0}^{n-e_i-d_i}{n-e_i-d_i \choose k}x^{n-d_i-k}(1-x)^{k+d_i} & \\
		&= \sum_{i=1}^m\lambda_i\sum_{k=d_i}^{n-e_i}{n-e_i-d_i \choose k-d_i}x^{n-k}(1-x)^{k} &\text{(index shift on $k$)}\\
		&= \sum_{i=1}^m\lambda_i\sum_{k=0}^{n}{n-e_i-d_i \choose k-d_i}x^{n-k}(1-x)^{k} &\text{(adding zero terms)}\\
		&= \sum_{k=0}^n\bigg[\underbrace{{n \choose k}^{-1}\sum_{i=1}^{m}\lambda_i{n-e_i-d_i \choose k-d_i}}_{\coloneqq p_k \geq 0}\bigg]{n \choose k}x^{n-k}(1-x)^{k}&
	\end{align*}
	Notice that the binomial coefficients ${n-e_i-d_i \choose k}$ are well defined (i.e. $n-e_i-d_i \geq 0$) for all $i$ and all $k$ due to the way we have chosen $n$.
	In the exact same way we obtain an expression for $1-\pol$ with coefficients $p_k' \geq 0$. It remains to show that $p_k \leq 1$ for all $0 \leq k \leq n$. From $f + (1-f) = 1$ it follows that
	\[
	\sum_{k=0}^n(p_k+p_k')\cdot {n \choose k}\cdot x^{n-k} \cdot (1-x)^k = 1.
	\]
	The Binomial Theorem states that $p_k + p_k' = 1$ for all $0 \leq k \leq n$ is a solution of this equation. However, the terms $x^{n-k}\cdot(1-x)^k$ are linearly independent over $\RR$, so it must be the only solution.
	Thus $p_k, p_k' \leq 1$ for all $0\leq k \leq n$.
%	$$
%	\begin{pmatrix}
%		{n \choose 0} 	& {n \choose 1} & \dots \\
%		0 				& {n \choose 1} & \dots
%	\end{pmatrix}
%	$$
\end{proof}
\restatehandelmanconstr*
\begin{proof}
	First suppose that $\pol$ is strictly positive on $[0,1]$ (this is slightly stronger than being adequate). According to Lem. \ref{lem:binomialrep}, we can write $\pol$ as 
	\[
	\pol = \sum_{k=0}^np_k\cdot {n \choose k}\cdot x^{n-k}\cdot(1-x)^k
	\]
	for some $n \geq 0$ and $p_k \in [0,1]$ for all $0 \leq k \leq n$.

	Consider the simple pMC $\mdp_n$ which has the topology of a `pyramid' of height $n$ and where the transition to the left child always has probability $x$ and the one to the right child is taken with probability $1-x$. For example, the upper four rows in Fig. \ref{fig:handelmanex:app} constitute $\mdp_3$.
	Let $\sinit$ be the top of the pyramid and $s_0,\dots,s_n$ be the $n+1$ states in the `basement' of $\mdp_n$. It is not difficult to prove by induction on $n$ and using the basic identity ${n \choose k} = {n-1 \choose k-1} + {n-1 \choose k}$ that for all $0 \leq k \leq n$, there are ${n \choose k}$ many paths, each with probability $x^{n-k}(1-x)^k$, from $\sinit$ to $s_k$.
	Now we add an additional target state $T$ to $\mdp_n$ and new transitions with the constant probability $p_k$ from $s_k$ to $T$. In this pMC, it then holds that
	\[
	\Pr(\lozenge T) = \sum_{k=0}^np_k{n \choose k}x^{n-k}(1-x)^k = f.
	\]
	If $\pol$ was not strictly positive on $[0,1]$, then $\pol = x^e(1-x)^d \pol'$ for some integers $e, d \geq 0$ and a polynomial $\pol'$ which is strictly positive on $[0,1]$. We can then apply the construction for $\pol'$ and simply prepend the corresponding $x$ and $1-x$ transitions.
\end{proof}

\begin{figure}
\centering

\scalebox{0.7}{
\begin{tikzpicture}[initial text=, scale=0.6]
	\node[state, initial, scale=0.5] (root) at (0,0) {};
	
	\node[state, scale=0.5] (l) at (-2,-2) {};
	\node[state, scale=0.5] (r) at (2,-2) {};
	
	\node[state, scale=0.5] (ll) at (-4,-4) {};
	\node[state, scale=0.5] (lr) at (0,-4) {};
	\node[state, scale=0.5] (rr) at (4,-4) {};
	
	\node[state, scale=0.5] (lll) at (-6,-6) {};
	\node[state, scale=0.5] (llr) at (-2,-6) {};
	\node[state, scale=0.5] (lrr) at (2,-6) {};
	\node[state, scale=0.5] (rrr) at (6,-6) {};
	
	\node[state,accepting, scale=0.5] (T) at (0,-9) {};
	
	\draw[->] (root) -- node[sloped, anchor=center,below] {$x$} (l);
	\draw[->] (root) -- node[sloped, anchor=center,below] {$1-x$} (r);
	
	\draw[->] (l) -- node[sloped, anchor=center,below] {$x$} (ll);
	\draw[->] (l) -- node[sloped, anchor=center,below] {$1-x$} (lr);
	\draw[->] (r) -- node[sloped, anchor=center,below] {$x$} (lr);
	\draw[->] (r) -- node[sloped, anchor=center,below] {$1-x$} (rr);
	
	\draw[->] (ll) -- node[sloped, anchor=center,below] {$x$} (lll);
	\draw[->] (ll) -- node[sloped, anchor=center,below] {$1-x$} (llr);
	\draw[->] (lr) -- node[sloped, anchor=center,below] {$x$} (llr);
	\draw[->] (lr) -- node[sloped, anchor=center,below] {$1-x$} (lrr);
	\draw[->] (rr) -- node[sloped, anchor=center,below] {$x$} (lrr);
	\draw[->] (rr) -- node[sloped, anchor=center,below] {$1-x$} (rrr);
	
	\draw[->] (lll) -- node[below] {$\nicefrac{1}{4}$} (T);
	\draw[->] (llr) -- node[left] {$\nicefrac{11}{12}$} (T);
	\draw[->] (lrr) -- node[right] {$\nicefrac{11}{12}$} (T);
	\draw[->] (rrr) -- node[below] {$\nicefrac{1}{4}$} (T);
\end{tikzpicture}
}
\label{fig:handelmanex:app}
\caption{The reachability probability from source to target equals $f = 2x\cdot(1-x)+\frac{1}{4}$.}
\end{figure}
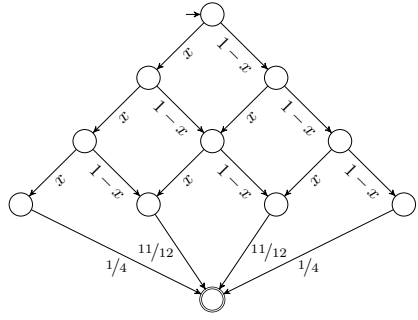

\subsection{Complexity of pMCs}

\subsubsection{Proof of Lemma~\ref{lem:etr:mb4feas}}
\restatembFEASdef*
\restatembFEASetr*

\noindent We start our reduction from the following ETR-complete problem:
\begin{definition}
  The decision problem \emph{bounded-closed quadratic conjunction} (bQUAD-c)
  asks: Given quadratic polynomials $f_1,\hdots,f_m$, $\exists \inst\colon
  \Params \rightarrow [-1,1]$	s.t.\ $\bigwedge_{i=1}^m f_i[\inst] = 0$?  The open
  variant of the decision problem (bQUAD-o) is defined analogously, where the range of $\inst$ is $(-1,1)$.
\end{definition}
\begin{theorem}(From \cite[Lemma 3.9]{Schaefer2013}\footnote{Strictly, the lemma
    is an immediate corollary to (the proof of) \cite[Lemma 3.9]{Schaefer2013}.
    Although Schaefer states the result for the Euclidean norm and a closed
    ball, the proof ensures that if there is a  solution to the original
    constraint system in the unbounded space, there is also a solution to the
    constructed constraint system in the open ball. Therefore, there is also
    a solution in the closed ball or the open/closed box. This adaption also
    occurs implicitly in \cite{DBLP:journals/mst/SchaeferS17}, where a closed
    box is assumed.})
\label{thm:etr:bquad}
	The bQUAD-o/c problems are ETR-complete.
\end{theorem}
\noindent 
We first reduce bQUAD to b4FEAS (like mb4FEAS, but with different range, and equality).
\begin{definition}
	The decision problem closed-bounded-4-feasibility (b4FEAS-c) asks: 
Given a (non-negative) polynomial $f$ of degree at most $4$, $\exists \inst\colon \Params \rightarrow [-1,1]$	s.t. $f[\inst] = 0$?
The open-bounded-4-feasibility (b4FEAS-o) can be defined analogously where $\inst\colon \Params \rightarrow (0,1)$. 
\end{definition}

\begin{lemma}
\label{lem:etr:bquadtob4feas}
	bQUAD-c $\;\leq_p\;$ b4FEAS-c and bQUAD-o $\;\leq_p\;$ b4FEAS-o.
\end{lemma}
\begin{proof}
 The proof is a straightforward adaption from the proof of \cite[Lemma 3.2]{DBLP:journals/mst/SchaeferS17}.
	For quadratic polynomials $f_1,\hdots,f_m$, consider the quadric polynomial $f = \sum_{i=1}^m \left( f_i \right)^2$. 
	Clearly, each $(f_i)^2$ is non-negative and has the same roots as $f_i$. As the sum of non-negative polynomials is non-negative, $f$ is non-negative, and has roots exactly where the $f_i$ have common roots.
    As it preserves all roots, it in  particular preserves roots in any bounded domain such as $[-1,1]^\Params$.	
    The bitlength of the encoding at most doubles~\cite[Lemma 3.2]{DBLP:journals/mst/SchaeferS17}.
\end{proof}

\begin{lemma}
\label{lem:etr:bfeastomb4feas}
b4FEAS-c $\;\leq_p\;$ mb4FEAS-c and b4FEAS-o $\;\leq_p\;$ mb4FEAS-o.
\end{lemma}

\begin{proof}
We show this in three steps:
\begin{enumerate}
\item We reduce b4FEAS-c to the problem
 \begin{description}
  \item {\textbf{P1}-c:} given a quadric\footnote{degree at most 4} (non-negative) polynomial $f'$,
   \[ \exists \inst\colon \Params \rightarrow [0,2]\text{ s.t. }f'[\inst] = 0.\]
	\end{description}
	Let $f$ be the input to b4FEAS-c. We construct $f'$ as $f[x_1 \mapsto x_1-1, \hdots x_n \mapsto x_n-1]$. 
	\begin{itemize}
	\item \textbf{Example:} For $f = -x_1^2x_2^2 + x_1^2$, we obtain $f' = -(x_1-1)^2(x_2-1)^2 + (x_1-1)^2$. 
	 Consider $\inst = [x_1 \mapsto 1, x_2 \mapsto -1]$: We have that $f[\inst] = 0$. \\
	 Now consider $\inst' = [x_1 \mapsto 2, x_2 \mapsto 0]$: We have $f'[\inst'] = -1\cdot 1 + 1 = 0$. 
	\item \textbf{Correctness:} 
	For each $\inst \in [-1,1]^X$, it holds that $f[\inst] = f'[\inst']$ for $\inst'$  with $\inst'(x_i) = 
	\inst(x_i) + 1$. Observe that $\inst' \in [0,2]^X$.
	\emph{There is a bijection between roots of $f$ and $f'$}:
	Assume $\inst$ to be a root of $f$. Then $\inst'$ is a root of $f'$ with $\inst'$ as above. 
	Analogously, any root $\inst'$ of $f'$ corresponds to a root of $f$.
	Each term remains (at most) quadric, thus $f'$ is also quadric.
		As $f$ is non-negative, so is $f'$. 
	
	\item \textbf{Complexity:}
	The substitution can be applied on a per-term basis: Consider the term $t_j = c_j \cdot x_{j_1}x_{j_2}x_{j_3}x_{j_4}$: It becomes the polynomial $f'_j = c_j \cdot (x_{j_1} - 1) \cdot (x_{j_2} - 1) \cdot (x_{j_3} - 1) \cdot (x_{j_4} - 1)$ which in its expanded form contains 16 terms $t'_{j,1}, \hdots t'_{j,16}$ with (for each $j$) pairwise different monomials, thus all coefficients in $f'_j$ have the same magnitude as the coefficient of $t_j$. We obtain $f'$ $\sum_j \sum_i t_{j,i}$.
	\end{itemize}
\item We reduce \textbf{P1}-c to the problem
  \begin{description}
  \item {\textbf{P2}-c:} given a quadric (non-neative) polynomial  $\hat{f}$,
	\[ \exists \inst\colon \Params \rightarrow [0,1]\text{ s.t. }\hat{f}[\inst] = 0.\] 
	\end{description}
	Let $f'$ be the input to \textbf{P1}-c. 
	We construct $\hat{f}$ by $f'[x_1 \mapsto 2\cdot x_1,\hdots, x_k \mapsto 2\cdot x_k]$.
	\begin{itemize}
	\item \textbf{Example:} For $f' = x_1^2-4x_1+x_2 + 4$, we construct $\hat{f} = 4x_1^2-8x_1+2x_2 + 4$.
	Consider $\inst' = [x_1 \mapsto 2, x_2 \mapsto 0]$: We have that $f'[\inst'] = 4-8+0+4=0$. \\
	 Now consider $\hat{\inst} = [x_1 \mapsto 1, x_2 \mapsto 0]$: We have $\hat{f}[\hat{\inst}] = 4-8+0+4 = 0$. 
	\item \textbf{Correctness:} 
	For each $\inst \in [0,2]^X$, it holds that $f[\inst] = f'[\inst']$ for $\inst'$  with $\inst'(x_i) = 
	\frac{\inst(x_i)}{2}$. Observe that $\inst' \in [0,1]^X$.
	\emph{There is a bijection between roots of $f$ and $f'$}:
	Assume $\inst$ to be a root of $f$. Then $\inst'$ is a root of $f'$ with $\inst'$ as above. 
	Analogously, any root $\inst'$ of $f'$ corresponds to a root of $f$.
	\emph{The polynomial remains non-negative}. The monomials (and thus the degree of the polynomial) remains unchanged.
	\item \textbf{Complexity:} Each coefficient is multiplied at most four times by two; and requires thus only constantly many bits more.
	\end{itemize}
\item We reduce the problem \textbf{P2}-c to mb4FEAS-c:
This is straightforward, just observe that for a non-negative polynomial $f$, it holds that for any $\inst$, $f[\inst] \leq 0 \iff f[\inst] = 0$.
\end{enumerate}
The reduction chain reduces 	b4FEAS-c to mb4FEAS-c, and can be analogously applied to reduce b4FEAS-o to mb4FEAS-o.
\end{proof}

\begin{proof}[Proof of Lemma~\ref{lem:etr:mb4feas}]
 Immediate corollary to Thm.~\ref{thm:etr:bquad}, Lemma~\ref{lem:etr:bquadtob4feas} and Lemma~\ref{lem:etr:bfeastomb4feas}.
\end{proof}

%\restatepmcetrhard*

%%%%%%%%%%%%%%%%%%%%%%%%%%%%%%%%%%%%%%%%%%%%%%%%%%%%%%%%%%%%%%

\subsubsection{Proof of Lemma~\ref{lem:gpvswdequiv}}
\restategpnotnecessary*
\noindent By prepending a pMC with the gadget outlined in Fig.~\ref{fig:e_reach_gp_red_e_reach_repeated}, we can assure that graph non-preserving parameter instantiations never satisfy lower-bounded reachability.
\begin{lemma}
\label{lem:e_reach_gp_red_e_reach}
	$\exists\reach_{\mathrm{gp}}^> \;\leq_p\; \exists\reach^>_\mathrm{wd}$ and $\exists\reach_{\mathrm{gp}}^\geq \;\leq_p\; \exists\reach^\geq_\mathrm{wd}$.
\end{lemma}

\noindent Even stronger, for strict lower-bounded reachability, we can restrict our attention to graph-preserving parameter instantiations.
The lemma below reformulates~\cite[Thm.~5]{DBLP:conf/uai/Junges0WQWK018}, and is an immediate consequence of Thm.~\ref{thm:semicontinuous}.
\begin{lemma}
\label{lem:e_reach_less_red_e_reach_less_gp}
$\exists\reach^>_\mathrm{wd} \;\leq_p\; \exists\reach_{\text{gp}}^>$.
\end{lemma}

\begin{proof}[Proof of Lemma~\ref{lem:gpvswdequiv}] 
Immediate from Lemmas~\ref{lem:e_reach_gp_red_e_reach} and~\ref{lem:e_reach_less_red_e_reach_less_gp}.	
\end{proof}
\begin{figure}
\begin{tikzpicture}[initial text=]
	\node[initial,state] (s1) {$s_{x_1}$};
	\node[state, right=of s1] (s2) {$s'_{x_1}$};
	\node[state, draw, right=of s2] (s3) {$s_{x_2}$};
	\node[state, right=of s3] (s4) {$s'_{x_2}$};
	\node[circle, right=of s4] (s5) {$\hdots$};
	\node[state, right=of s5] (s6) {$s'_{x_n}$};
	\node[state, right=of s6] (s7) {$\sinit$};
	
	\draw[->] (s1) edge node[auto] {$1{-}x_1$} (s2);
	\draw[->] (s2) edge node[auto] {$x_1$} (s3);
	\draw[->] (s3) edge node[auto] {$1{-}x_2$} (s4);
	%\draw[->] (s4) edge node {$p_2$} (s2);
	\draw[->] (s6) edge node[auto] {$x_n$} (s7);

	\draw[->] (s1) edge[loop above] node[auto] {$x_1$} (s1);
	\draw[->] (s2) edge[loop above] node[auto] {$1{-}x_1$} (s2);
	\draw[->] (s3) edge[loop above] node[auto] {$x_2$} (s3);
	\draw[->] (s6) edge[loop above] node[auto] {$1{-}x_n$} (s6);
	
\end{tikzpicture}	
\caption{Gadget for the proof of Lemma~\ref{lem:e_reach_gp_red_e_reach}}
\label{fig:e_reach_gp_red_e_reach_repeated}
\end{figure}
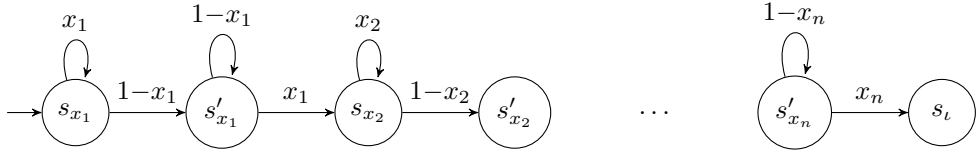
\begin{proof}[Proof of Lemma~\ref{lem:e_reach_gp_red_e_reach}]
Given a simple pMC $\mdp$. We extend the pMC with a gadget outlined in Fig.~\ref{fig:e_reach_gp_red_e_reach_repeated}. Formally, we construct an pMC $\mdp'$ with states $S'=S \cup \{ s_x, s'_x \mid x \in \Params \}$, initial state $s_{x_1}$ and \[ P'(s,s') = \begin{cases}
 P(s,s) & \text{if }s,s' \in S \\
 x    & \text{if }s = s' = s_x \\
 1{-}x & \text{if }s = s' = s'_x \\
 1{-}x & \text{if }s = s_x \text{ and } s' = s'_x \\
 x & \text{if }s = s'_x \text{ and } s' = \text{next}(s'_{x}) \\
 0 & \text{otherwise.}	
 \end{cases}
 \]
where $\text{next}(s'_x)$ is $s_{x+1}$ if $x = x_i$ for some $i < |\Params|$, and $\sinit$ if $i = |\Params|$.
The pMC $\mdp'$ is only linearly larger than $\mdp$. Observe that the construction of the gadget can be adapted for non-simple pMCs (with different well-defined parameter spaces).
By construction \[ \Pr_{\mdp'[\inst]}(\lozenge T) =  \Pr_{\mdp'[\inst]}(\lozenge \{ \sinit \}) \cdot \Pr_{\mdp[\inst]}(\lozenge T) \]
We observe the following: \begin{align}& \forall \inst \in \ParamSpace^\text{gp}.\; \Pr_{\mdp'[\inst]}(\lozenge \{ \sinit \}) = 1 \text{ and thus } \Pr_{\mdp'[\inst]}(\lozenge T) = \Pr_{\mdp[\inst]}(\lozenge T)\quad\text{, and}\label{eq:e_reach_gp_red_e_reach:obs1} \\
& \forall \inst \in \ParamSpace^\text{gp} \setminus \ParamSpace^\text{wd}.\; \Pr_{\mdp'[\inst]}(\lozenge \{ \sinit \}) = 0\label{eq:e_reach_gp_red_e_reach:obs2}.
\end{align}

\noindent We have:
\begin{align*}
	& \exists \inst \in \ParamSpace^{\text{gp}}_{\mdp} .\;\Pr_{\mdp[\inst]}(\lozenge T) {\unrhd} \lambda
	\implies 
	\exists \inst \in \ParamSpace^{\text{gp}}_{\mdp'} .\;\Pr_{\mdp'[\inst]}(\lozenge T) {\unrhd} \lambda 
	\overset{\eqref{eq:e_reach_gp_red_e_reach:obs1}}{\implies}
	\\
	&\exists \inst \in \ParamSpace^{\text{wd}}_{\mdp'} .\;\Pr_{\mdp'[\inst]}(\lozenge T) {\unrhd} \lambda,
\end{align*}
and
\begin{align*}
	& \not\exists \inst \in \ParamSpace^{\text{gp}}_{\mdp} .\;\Pr_{\mdp[\inst]}(\lozenge T) {\unrhd} \lambda 
	\implies 
	\forall \inst \in \ParamSpace^{\text{gp}}_{\mdp} .\;\Pr_{\mdp[\inst]}(\lozenge T) {\unlhd} \lambda 
	\overset{\eqref{eq:e_reach_gp_red_e_reach:obs2}}{\implies}\\
	& \forall\inst \in \ParamSpace^{\text{wd}}_{\mdp}.\; \Pr_{\mdp'[\inst]}(\lozenge T) {\unlhd} \lambda 
	\implies
	\not\exists \inst \in \ParamSpace^{\text{wd}}_{\mdp}.\;\Pr_{\mdp'[\inst]}(\lozenge T) {\unrhd} \lambda .
\end{align*}
Together, 
\begin{align*}
\exists \inst \in \ParamSpace^{\text{gp}}_{\mdp}.\; \Pr_{\mdp[\inst]}(\lozenge T) {\unrhd} \lambda
	\quad \iff \quad
	\exists \inst \in \ParamSpace^{\text{wd}}_{\mdp'}.\; \Pr_{\mdp'[\inst]}(\lozenge T) {\unrhd} \lambda.
\end{align*}

\end{proof}

\subsubsection{Proof of Theorem \ref{thm:e_reach_gr_np_hard}}
\restatechonevrefined*
\noindent
We extend the proof of \cite[Thm.\ 3]{Chonev17} and show that the same 3SAT-reduction also applies to the $\exists\reach^>_\mathrm{wd}$ and $\exists\reach^<_\mathrm{wd}$ problem. 
\begin{proof}
We first reformulate the construction in our notation: 
Let $\phi = \phi_1 \wedge \ldots \wedge \phi_k$ be a given 3SAT-formula, i.e. the $\phi_j$ are of the form
$$\phi_j = l_{j,1} \vee l_{j,2} \vee l_{j,3},$$
where the $l_{i,j}$ are \emph{literals} (variables or negated variables). Let $x_1, \hdots, x_m$ be the variables of $\varphi$.
 The nonsimple pMC for $\varphi$ is outlined in Fig.~\ref{fig:chonevconstruction}. 
 Formally, the pMC $\mdp_\varphi = (S, \Params, \sinit, P)$ is defined as follows:
\begin{itemize}
	\item $S =
		\{v_i \mid 0 \leq i \leq m\}~\uplus~
		\{x_i, \overline{x_i} \mid 1 \leq i \leq m\}~\uplus~
	    \{\phi_i \mid 1 \leq i \leq k\}~\uplus~
	    \{T, \bot\}$ are the $3m + k + 3$ states, $v_0 = \sinit$ is the initial state, $T$ and $\bot$ indicate target and sink respectively,
	\item $\Params = \{y_1, ..., y_m\}~\uplus~\{z_{1,1}, z_{1,2}, z_{1,3}, ..., z_{k,1}, z_{k,2}, z_{k,3}\}$ are the $m + 3k$ parameters,
	\item for all $1 \leq i \leq m$ and $1 \leq j \leq k$ we define the transition probabilities as
	\begin{align*}
		&P(v_{i-1}, x_i) = y_i, && P(v_{i-1}, \overline{x_i}) = 1-y_i, \\
		&P(x_i, v_i) = y_i, && P(\overline{x_i}, v_i) = 1-y_i, \\
		&P(x_i, \bot) = 1-y_i, && P(\overline{x_i}, \bot) = y_i, \\
		&P(v_m, \phi_j) = \frac{1}{k+1}, && P(v_m, T) = \frac{1}{k+1}, \\
		&P(\phi_j, x_i) = z_{j,r} \text{ if } l_{j,r} = x_i, && P(\phi_j, \overline{x}_i) = z_{j,r} \text{ if } l_{j,r} = \overline{x}_i.
	\end{align*}
\end{itemize}
Moreover, we let $P(s,t) = 0$ for each pair $(s,t)$ of states not specified above.
At the end of the proof, we describe how to reformulate the pMC into a simple pMC.

Observe that under any well-defined instantiation, there are exactly two bottom strongly connected components, namely $\bot$ and $T$. As a consequence, it holds that:
\begin{align}
\label{eq:gadget:twoendcomponents}
	\forall \inst \in \ParamSpace_{\mdp_\phi}^\text{wd}.\; \quad \Pr_{\mdp_\phi[\inst]}(\lozenge T) + \Pr_{\mdp_\phi[\inst]}(\lozenge \bot) = 1.
\end{align}
\begin{figure}
	\centering
	
	\begin{tikzpicture}[scale=0.8, every node/.style={scale=0.8},initial text=]
		\node[state,initial] (v_0) {$v_0$};
		\node[state, right=7mm of v_0,yshift=12mm] (x_1) {$x_1$};
		\node[state, right=7mm of v_0,yshift=-12mm] (notx_1) {$\overline{x}_1$};
		
		\node[state, right=of v_0,xshift=14mm] (v_1) {$v_1$};
		\node[state, right=7mm of v_1,yshift=12mm] (x_2) {$x_2$};
		\node[state, right=7mm of v_1,yshift=-12mm] (notx_2) {$\overline{x}_2$};
		
		\node[state, right=of v_1,xshift=14mm] (v_2) {$v_2$};
		\node[state, right=7mm of v_2,yshift=12mm] (x_3) {$x_3$};
		\node[state, right=7mm of v_2,yshift=-12mm] (notx_3) {$\overline{x}_3$};
		
		\node[right=of v_2,xshift=7mm] (dots) {$\ldots$};
		\node[state, right=1mm of dots,yshift=12mm] (x_m) {$x_m$};
		\node[state, right=1mm of dots,yshift=-12mm] (notx_m) {$\overline{x}_m$};
		
		\node[state, right=of dots,xshift=5mm] (v_m) {$v_m$};
		
		\node[state, right=of v_m, xshift=10mm,yshift=16mm] (phi_1) {$\phi_1$};
		\node[state, right=of v_m, xshift=10mm,yshift=-16mm] (phi_k) {$\phi_k$};
		\node[right=of v_m, xshift=14mm,rotate=90,anchor=center] (dots2) {$\dots$};
		
		\node[left=of phi_1, xshift=-6mm,yshift=12mm,rotate=90,anchor=center] (literals_1) {to literals};
		\node[left=of phi_k, xshift=-6mm,yshift=-12mm,rotate=90,anchor=center] (literals_k) {to literals};
		
		\node[state,above=of v_2,yshift=16mm] (bot1) {$\bot$};
		\node[state,below=of v_2,yshift=-16mm] (bot2) {$\bot$};
		
		\node[state,accepting,above=of v_m,yshift=16mm] (T) {$T$};
		
		%%%%%%%%%%%%%%%%
		
		\draw[->] (v_0) -- node[sloped, anchor=center,above] {$y_1$} (x_1);
		\draw[->] (v_0) -- node[sloped, anchor=center,below] {$1-y_1$} (notx_1);
		\draw[->] (x_1) -- node[sloped, anchor=center,above] {$y_1$} (v_1);
		\draw[->] (notx_1) -- node[sloped, anchor=center,below] {$1-y_1$} (v_1);
		
		\draw[->] (v_1) -- node[sloped, anchor=center,above] {$y_2$} (x_2);
		\draw[->] (v_1) -- node[sloped, anchor=center,below] {$1-y_2$} (notx_2);
		\draw[->] (x_2) -- node[sloped, anchor=center,above] {$y_2$} (v_2);
		\draw[->] (notx_2) -- node[sloped, anchor=center,below] {$1-y_2$} (v_2);
		
		\draw[->] (v_2) -- node[sloped, anchor=center,above] {$y_3$} (x_3);
		\draw[->] (v_2) -- node[sloped, anchor=center,below] {$1-y_3$} (notx_3);
		
		\draw[->] (x_m) -- node[sloped, anchor=center,above] {$y_m$} (v_m);
		\draw[->] (notx_m) -- node[sloped, anchor=center,below] {$1-y_m$} (v_m);
		
		\draw[->] (v_m) -- node[sloped, anchor=center,above] {$\nicefrac{1}{k+1}$} (phi_1);
		\draw[->] (v_m) -- node[sloped, anchor=center,below] {$\nicefrac{1}{k+1}$} (phi_k);
		
		\draw[->] (phi_1) edge[bend left=30] node[sloped, anchor=center,above] {$z_{1,1}$} (literals_1);
		\draw[->] (phi_1) edge[bend right=30] node[sloped, anchor=center,above] {$z_{1,3}$} (literals_1);
		\draw[->] (phi_1) edge[bend left=0] node[sloped, anchor=center,above] {$z_{1,2}$} (literals_1);
		
		\draw[->] (phi_k) edge[bend left=30] node[sloped, anchor=center,below] {$z_{k,1}$} (literals_k);
		\draw[->] (phi_k) edge[bend right=30] node[sloped, anchor=center,below] {$z_{k,3}$} (literals_k);
		\draw[->] (phi_k) edge[bend left=0] node[sloped, anchor=center,below] {$z_{k,2}$} (literals_k);
		
		\draw[->] (x_1) -- node[sloped, anchor=center,above] {$1-y_1$} (bot1);
		\draw[->] (x_2) -- node[sloped, anchor=center,above] {$1-y_2$} (bot1);
		\draw[->] (x_3) -- node[sloped, anchor=center,below] {$1-y_3$} (bot1);
		\draw[->] (x_m) -- node[sloped, anchor=center,above] {$1-y_m$} (bot1);
		
		\draw[->] (notx_1) -- node[sloped, anchor=center,above] {$y_1$} (bot2);
		\draw[->] (notx_2) -- node[sloped, anchor=center,above] {$y_2$} (bot2);
		\draw[->] (notx_3) -- node[sloped, anchor=center,above] {$y_3$} (bot2);
		\draw[->] (notx_m) -- node[sloped, anchor=center,above] {$y_m$} (bot2);
		
		\draw[->] (v_m) -- node[left] {$\nicefrac{1}{k+1}$} (T);
		
	\end{tikzpicture}	
	
	\caption{Chonev's construction used in the proof of Thm. \ref{thm:e_reach_gr_np_hard} ($\bot$ is duplicated for readability).}
	\label{fig:chonevconstruction}
\end{figure}
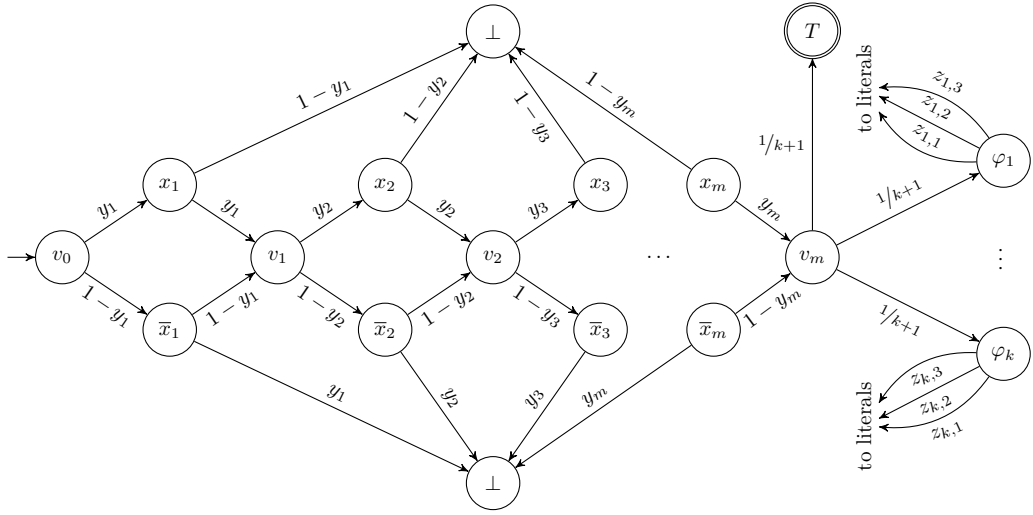

We first show the following claim to simplify our proof afterwards:\\
\noindent\textbf{Auxiliary claim:} If $\phi$ is unsatisfiable, then for all $\inst \in \ParamSpace_{\mdp_\phi}^\text{wd}$ there exists some clause $\phi_{j^*}$, such that $P(l_{j^*,r}, \bot)[\inst] \geq \frac{1}{2}$ for all $r\in\{1,2,3\}$, or more formally
\begin{align}
	\label{eq:auxclaim}
	\phi \text{ is unsat} \Longrightarrow \forall \inst \in \ParamSpace_{\mdp_\phi}^\text{wd},\; \exists j^* \in \{1,\ldots,k\}, \forall r \in\{1,2,3\}.\; P(l_{j^*,r}, \bot)[\inst] \geq \frac{1}{2}
\end{align}
\emph{Proof of the auxiliary claim:} By contraposition. Suppose for some $\inst$ and for every clause $\phi_j$ there is a `witness' literal $l_{j,r}$ with $P(l_{j,r}, \bot)[\inst] < \frac{1}{2}$.
By definition of $P$,  either
\begin{enumerate}
	\item $l_{j,r}$ is a variable $x_i$ and $\inst(y_i) > \frac{1}{2}$ or
	\item $l_{j,r}$ is a negated variable $\overline{x_i}$ and $\inst(y_i) < \frac{1}{2}$.
\end{enumerate}
Consider some variable assignment $\mathcal{I}$ for $\phi$, with
$$\mathcal{I}(x_i) =
\begin{cases}
	1, & \text{if } \inst(y_i) > \frac{1}{2}, \\
	0, & \text{if } \inst(y_i) < \frac{1}{2}, \\
	\text{arbitrary}, & \text{if } \inst(y_i) = \frac{1}{2}.
\end{cases}$$
In both case 1 and 2 above, $\mathcal{I}$ satisfies clause $\phi_j$. 
Thus $\phi$ is satisfiable. By contraposition, statement \eqref{eq:auxclaim} holds.
\medskip

\newcommand{\prfromto}[2]{\Pr(#1 \rightarrow \lozenge #2)}
\noindent \emph{Proof for correctness of reduction:} 
We show:
\begin{align*}
	\phi \text{ is sat} \iff \exists \inst \in \ParamSpace^{\text{wd}}.\; \Pr_{\mdp_\phi[\inst]}(\lozenge T) > \frac{2}{3}.
\end{align*}
\noindent The first step is to show that
\begin{align}
	\label{eq:chonevpart2}
	\phi \text{ is unsat} \Longrightarrow \forall \inst \in \ParamSpace^{\text{wd}}: \Pr_{\mdp_\phi[\inst]}(\lozenge T) \leq \frac{2}{3}.
\end{align}
%Once again let $\phi$ be unsat and fix a parameter valuation $\inst$.
We denote by $\prfromto{s}{t}$ the probability to move from state $s$ to $t$ in more than zero steps in $\mdp_\phi[\inst]$ for a fixed $\inst$.
Let $j^*$ be like in the auxiliary claim \eqref{eq:auxclaim}. By construction of $\mdp_\phi$ we have that
\[
\prfromto{l_{j^*,r}}{v_m} \leq 1 - P(l_{j*,r}, \bot)[\inst] \leq \frac{1}{2}
\]
for all $r \in \{1,2,3\}$ and hence
\begin{align*}
	\prfromto{\phi_{j^*}}{v_m} = \sum_{r=1}^3 \inst(z_{j^*, r})\cdot \prfromto{l_{j^*,r}}{v_m} \leq \frac{1}{2}\sum_{r=1}^3 \inst(z_{j^*, r}) = \frac{1}{2}.
\end{align*}
Consequently, for $\prfromto{v_m}{v_m}$ it holds that
\begin{align*}
	\prfromto{v_m}{v_m} &= P(v_m, \phi_{j^*}) \cdot \prfromto{\phi_{j^*}}{v_m} + \sum_{j \neq j^*}P(v_m, \phi_j)\cdot \prfromto{\phi_j}{v_m} \\
	& \leq \frac{1}{k+1}\cdot\frac{1}{2} + \frac{k-1}{k+1} = \frac{2k-1}{2(k+1)}.
\end{align*}
Plugging this into the equation
\begin{align*}
	\prfromto{v_m}{T} = \frac{1}{k+1} + \prfromto{v_m}{v_m}\cdot \prfromto{v_m}{T},
\end{align*}
yields 
$\prfromto{v_m}{T} \leq \frac{2}{3}$
by a straightforward calculation. All paths from $v_0$ to $T$ go through $v_m$, thus: 
\begin{align*}
\Pr_{\mdp_\phi[\inst]}(\lozenge T) = \prfromto{\sinit}{T} = \prfromto{\sinit}{v_m} \cdot \prfromto{v_m}{T}
\end{align*}
Together,
\begin{align*}
\Pr_{\mdp_\phi[\inst]}(\lozenge T) = \prfromto{\sinit}{T} = \prfromto{\sinit}{v_m} \cdot \prfromto{v_m}{T} \leq \prfromto{v_m}{T} \leq \frac{2}{3}	
\end{align*}
proves \eqref{eq:chonevpart2}. 

\noindent For completeness, we reproduce the following from \cite{Chonev17}:
\begin{align*}
	\phi \text{ is sat} \Longrightarrow \exists \inst \in \ParamSpace^{\text{wd}}.\; \Pr_{\mdp_\phi[\inst]}(\lozenge T) = 1 > \frac{2}{3}.
\end{align*}
Choose some satisfying assignment $\mathcal{I}$ for $\varphi$. We construct $\inst \in  \ParamSpace^{\mathrm{wd}}$ in two steps. First, let  $\inst(y_i) = \mathcal{I}(x_i)$ for all $1 \leq i \leq m$.
Second, for each clause $\varphi_i$, select one literal $l_{i,j}$ which makes $\varphi_i$ true under $\mathcal{I}$, and set $\inst(z_{i,j}) = 1$. Set all other $z_{i,j}$ to $0$. Under this assignment, 
$\bot$ is unreachable, we conclude using \eqref{eq:gadget:twoendcomponents} that $\Pr_{\mdp_\phi[\inst]}(\lozenge T) = 1$.

\noindent Together, we obtained:
\begin{align*}
	\phi \text{ is sat} \iff \exists \inst \in \ParamSpace^{\text{wd}}.\; \Pr_{\mdp_\phi[\inst]}(\lozenge T) > \frac{2}{3}.
\end{align*}
We can adapt $\mdp_\phi$ to a simple pMC $\mdp_\phi'$ by updating the encoding around states $\varphi_i$, outlined in Fig.~\ref{fig:chonevconstruction:simple}.
Formally, the pMC $\mdp'_\phi = (S, \Params, \sinit, P)$ is defined as follows:
\begin{itemize}
	\item $S =
		\{v_i \mid 0 \leq i \leq m\}~\uplus~
		\{x_i, \overline{x_i} \mid 1 \leq i \leq m\}~\uplus~
	    \{\phi_i, \varphi'_i \mid 1 \leq i \leq k\}~\uplus~
	    \{T, \bot\}$ are the $3m + 2k + 3$ states, $\sinit = v_0$ is the initial state, $T$ and $\bot$ indicate target and sink respectively,
	\item $\Params = \{y_1, ..., y_m\}~\uplus~\{z_{1,1}, z_{1,*}, ..., z_{k,1}, z_{k,*}\}$ are the $m + 2k$ parameters,
	\item for all $1 \leq i \leq m$ and $1 \leq j \leq k$ we define the transition probabilities as
	\begin{align*}
		&P(v_{i-1}, x_i) = y_i, && P(v_{i-1}, \overline{x_i}) = 1-y_i, \\
		&P(x_i, v_i) = y_i, && P(\overline{x_i}, v_i) = 1-y_i, \\
		&P(x_i, \bot) = 1-y_i, && P(\overline{x_i}, \bot) = y_i, \\
		&P(v_m, \phi_j) = \frac{1}{k+1}, && P(v_m, T) = \frac{1}{k+1}, \\
		&P(\phi_j, x_i) = z_{j,1} \text{ if } l_{j,1} = x_i, && P(\phi_j, \overline{x}_i) = z_{j,1} \text{ if } l_{j,1} = \overline{x}_i\\
		&P(\phi_j, \phi'_j) = 1-z_{j,1},  && \\
		&P(\phi'_j, x_i) = z_{j,*} \text{ if } l_{j,2} = x_i, && P(\phi_j, \overline{x}_i) = z_{j,*} \text{ if } l_{j,2} = \overline{x}_i\\
		&P(\phi'_j, x_i) = 1-z_{j,*} \text{ if } l_{j,3} = x_i, && P(\phi_j, \overline{x}_i) = 1-z_{j,*} \text{ if } l_{j,3} = \overline{x}_i.
	\end{align*}
\end{itemize}
Moreover, we let $P(s,t) = 0$ for each pair $(s,t)$ of states not specified above.
Observe that for $\inst' \in \ParamSpace^\text{wd}_{\mdp'_\varphi}$, there exists $\inst \in \ParamSpace^\text{wd}_{\mdp_\varphi} $ s.t.:
 \[
 \Pr_{\mdp_\phi[\inst]}(\lozenge T) = \Pr_{\mdp'_\phi[\inst']}(\lozenge T),
 \] 
 by setting 
\begin{itemize}
\item 
$\inst(z_{i,2}) = \left(1 - \inst(z_{i,1}) \right) \cdot \inst(z_{i,*})$, and
\item 
$\inst(z_{i,3}) = \left(1 - \inst(z_{i,1}) \right) \cdot \left( 1- \inst(z_{i,*})\right)$	
\end{itemize}
for all $1 \leq i \leq k$.
Analogously, we have that for some $\inst \in \ParamSpace^\text{wd}_{\mdp_\varphi} $, there exists $\inst' \in \ParamSpace^\text{wd}_{\mdp'_\varphi}$.

\begin{figure}
\centering
\begin{tikzpicture}[scale=0.8, every node/.style={scale=0.8},initial text=]
\node[state,initial,gray] (v_0) {$v_0$};
		\node[state, right=7mm of v_0,yshift=12mm,gray] (x_1) {$x_1$};
		\node[state, right=7mm of v_0,yshift=-12mm,gray] (notx_1) {$\overline{x}_1$};
		
		\node[state, right=of v_0,xshift=14mm,gray] (v_1) {$v_1$};
		\node[state, right=7mm of v_1,yshift=12mm,gray] (x_2) {$x_2$};
		\node[state, right=7mm of v_1,yshift=-12mm,gray] (notx_2) {$\overline{x}_2$};

		\node[right=of v_1,xshift=7mm,gray] (dots) {$\ldots$};
		\node[state, right=1mm of dots,yshift=12mm,gray] (x_m) {$x_m$};
		\node[state, right=1mm of dots,yshift=-12mm,gray] (notx_m) {$\overline{x}_m$};
		
		\node[state, right=of dots,xshift=5mm,gray] (v_m) {$v_m$};
		
\node[state, right=of v_m, xshift=10mm,yshift=16mm] (phi_1) {$\phi_1$};
		\node[state, right=of v_m, xshift=10mm,yshift=-16mm] (phi_k) {$\phi_k$};
		\node[state, right=of phi_1,yshift=7mm] (phip_1) {$\phi'_1$};
		\node[state, right=of phi_k,yshift=-7mm] (phip_k) {$\phi'_k$};
		\node[right=of v_m, xshift=14mm,rotate=90,anchor=center] (dots2) {$\dots$};
		\node[right=0.68cm of dots2, xshift=14mm, yshift=-3.5mm, rotate=90,anchor=center] (dots3) {$\dots$};

		\node[left=of phi_1, xshift=-6mm,yshift=12mm,rotate=90,anchor=center] (literals_1) {to literals};
		\node[left=of phi_k, xshift=-6mm,yshift=-12mm,rotate=90,anchor=center] (literals_k) {to literals};

	\draw[->] (phi_1) edge[bend left=0] node[pos=0.4, sloped, anchor=center,above] {$z_{1,1}$} (literals_1);
		\draw[->] (phi_1) edge node[anchor=center,below,sloped] {$1-z_{1,1}$} (phip_1);
		\draw[->] (phip_1) edge[bend left=0] node[sloped, anchor=center,above] {$z_{1,*}$} (literals_1);
		\draw[->] (phip_1) edge[bend right=30] node[sloped, anchor=center,above] {$1-z_{1,*}$} (literals_1);
		
		\draw[->] (phi_k) edge[bend left=0] node[pos=0.4, sloped, anchor=center,below] {$z_{k,1}$} (literals_k);
		\draw[->] (phi_k) edge node[anchor=center,above,sloped] {$1-z_{k,1}$} (phip_k);
		\draw[->] (phip_k) edge[bend left=0] node[sloped, anchor=center,below] {$z_{k,*}$} (literals_k);
		\draw[->] (phip_k) edge[bend left=30] node[sloped, anchor=center,below] {$1-z_{k,*}$} (literals_k);

		\node[state,above=of v_1,yshift=16mm,gray] (bot1) {$\bot$};
		\node[state,below=of v_1,yshift=-16mm,gray] (bot2) {$\bot$};
		
		\node[state,accepting,above=of v_m,yshift=16mm,gray] (T) {$T$};
		
		%%%%%%%%%%%%%%%%
		
		\draw[->,gray] (v_0) -- node[sloped, anchor=center,above,gray] {$y_1$} (x_1);
		\draw[->,gray] (v_0) -- node[sloped, anchor=center,below,gray] {$1-y_1$} (notx_1);
		\draw[->,gray] (x_1) -- node[sloped, anchor=center,above,gray] {$y_1$} (v_1);
		\draw[->,gray] (notx_1) -- node[sloped, anchor=center,below,gray] {$1-y_1$} (v_1);
		
		\draw[->,gray] (v_1) -- node[sloped, anchor=center,above,gray] {$y_2$} (x_2);
		\draw[->,gray] (v_1) -- node[sloped, anchor=center,below,gray] {$1-y_2$} (notx_2);

		\draw[->,gray] (x_m) -- node[sloped, anchor=center,above,gray] {$y_m$} (v_m);
		\draw[->,gray] (notx_m) -- node[sloped, anchor=center,below,gray] {$1-y_m$} (v_m);
		
		\draw[->,gray] (v_m) -- node[sloped, anchor=center,above,gray] {$\nicefrac{1}{k+1}$} (phi_1);
		\draw[->,gray] (v_m) -- node[sloped, anchor=center,below,gray] {$\nicefrac{1}{k+1}$} (phi_k);

		\draw[->,gray] (x_1) -- node[sloped, anchor=center,above,gray] {$1-y_1$} (bot1);
		\draw[->,gray] (x_2) -- node[sloped, anchor=center,above,gray] {$1-y_2$} (bot1);
		\draw[->,gray] (x_m) -- node[sloped, anchor=center,above,gray] {$1-y_m$} (bot1);
		
		\draw[->,gray] (notx_1) -- node[sloped, anchor=center,above,gray] {$y_1$} (bot2);
		\draw[->,gray] (notx_2) -- node[sloped, anchor=center,above,gray] {$y_2$} (bot2);
		\draw[->,gray] (notx_m) -- node[sloped, anchor=center,above,gray] {$y_m$} (bot2);
		
		\draw[->,gray] (v_m) -- node[left,gray] {$\nicefrac{1}{k+1}$} (T);

\end{tikzpicture}	
\caption{Adaption of the gadget from Fig.~\ref{fig:chonevconstruction} to simplicity (grey = unchanged).}
\label{fig:chonevconstruction:simple}
\end{figure}
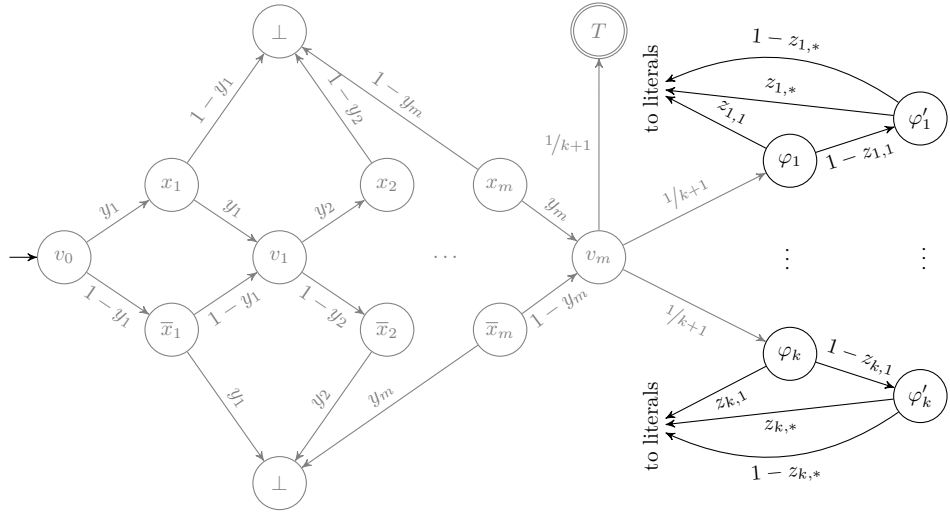

$\mdp_\phi$ can now be further modified to match the usual threshold $\lambda = \frac{1}{2}$ according to Rem. \ref{rem:fixedlambda}. It follows that $\exists\reach^>_\mathrm{wd}$ is NP-hard.

\noindent For $\exists\reach^<_\mathrm{wd}$, we remark that by \eqref{eq:gadget:twoendcomponents}, 
\[ 
\Pr_{\mdp_\phi[\inst]}(\lozenge T) > \frac{2}{3} \quad \iff \quad \Pr_{\mdp_\phi[\inst]}(\lozenge \bot) < \frac{1}{3}.
 \]
Thus, 
\begin{align*}
	\phi \text{ is sat} \iff \exists \inst \in \ParamSpace^{\text{wd}}.\; \Pr_{\mdp_\phi[\inst]}(\lozenge T) < \frac{1}{3}.
\end{align*}
Again, $\mdp_\phi$ can now be further modified to match the usual threshold $\lambda = \frac{1}{2}$. It follows that $\exists\reach^<_\mathrm{wd}$ is NP-hard.
\end{proof}

%%%%%%%%%%%%%%%%%%%%%%%%%%%%%%%%%%%%%%%%%%%%%%%

\subsection{Complexity of pMDPs}

\subsubsection{Proof of Lemma~\ref{lem:existwdgeeqexistwdge} and  Lemma~\ref{lem:equivsimplepmc}}

\restatepmdpsarepmcs*

\noindent The lemma follows immediately from:

\restateequivsimplepmcs*

The construction consists of two parts, as outlined in Sect.~\ref{sec:mdp:ee}.
Analogously, the lemma is a corollary to the following two lemmas:

\begin{definition}
	A simple pMDP $\mdp$ is called \emph{binary} if 
	\begin{itemize}
		\item for all $s \in S$: $|\Act(s)| \leq 2$, and
		\item for all $s,s' \in S$, $\act \in \Act$: $|\Act(s)| = 2 \implies P(s,\act, s') = 1$.
	\end{itemize}
\end{definition}

\begin{lemma}
\label{lem:makebinary}
	For all pMDPs $\mdp$ one can construct in polynomial time
  a (linearly larger) binary pMDP $\mdp'$ s.t.
  \[ 
    \left(\exists \inst \in
    \ParamSpace^{\mathrm{wd}}_\mdp, \exists \sched \in \Sched.\; \Pr_{\mdp[\inst]}^\sched(\lozenge T) \bowtie
    \frac{1}{2}\right) \iff \left(\exists \inst \in \ParamSpace^\mathrm{wd}_{\mdp'}, \exists \sched \in \Sched.\;
    \Pr_{{\mdp'}[\inst]}^\sched(\lozenge T) \bowtie \frac{1}{2}\right). 
  \]
\end{lemma}
\begin{lemma}
\label{lem:makepmc}
	For all binary pMDPs $\mdp$ one can construct in polynomial time
  a (linearly larger) simple pMC $\mdp'$ s.t.
  \[ 
    \left(\exists \inst \in
    \ParamSpace^{\mathrm{wd}}_\mdp,\exists \sched \in \Sched.\; \Pr_{\mdp[\inst]}^\sched(\lozenge T) \bowtie
    \frac{1}{2}\right)\; \iff\; \left(\exists \inst \in \ParamSpace^\mathrm{wd}_{\mdp'}.\;
    \Pr_{{\mdp'}[\inst]}^\sched(\lozenge T) \bowtie \frac{1}{2}\right). 
  \]
\end{lemma}

\begin{proof}[Proof of Lemma~\ref{lem:makebinary}]
	Let $\mdp = (S,\Params,\Act,\sinit,P)$ be a simple pMDP with $\Act = \{ \act_1, \hdots \act_m \}$.
	We construct a binary pMDP $\mdp'=(S',\Params,\Act',\sinit,P')$.\footnote{The construction here is not building a binary tree as recommended in the main paper, as a formalisation of that is less concise. Instead, we follow a construction sketched in \cite[Fig. 5b]{DBLP:conf/uai/Junges0WQWK018}.}
	Therefore, let $S' = S \times \{1,\hdots,m \}$, and $\Act' = \{ a,b \}$.
	Then $P'$ is defined by 
 \begin{itemize}
\item $	P'(\langle s, i \rangle,a,\langle s,i+1 \rangle) = 1$ for all $s \in S$ and $1 \leq i \leq m$.
\item  $P'(\langle s, i \rangle,b,\langle s',1 \rangle) = P(s,\act_i,s')$ for all $s,s' \in S$, $\act_i \in \Act$, and $1 \leq i \leq m$.
 	\end{itemize}
 	And $P'(\cdot) = 0$ otherwise.
 	
	\noindent Observe that if there is a transition $P(s,\act_i,s')$ then there is a path 
	\[ \langle s, 1 \rangle \xrightarrow{a:1}  \langle s, 2 \rangle \xrightarrow{a:1}  \langle s, 3 \rangle \hdots 
	\langle s,i \rangle \xrightarrow{b:P(s,\act_i,s')} \langle s',1 \rangle.\]
	Thus, a scheduler selecting $\act_i$ in $s$ can be mimicked by selecting $a$ $i-1$ times, and then selecting $b$. The probability of the path is $P(s,\act_i,s')$.
	Similarly, a scheduler in $\mdp'$ can be always mimicked. 
	Thus, the induced pMCs are weakly bisimilar, and they satisfy the same reachability properties.
\end{proof}

\begin{proof}[Proof of Lemma~\ref{lem:makepmc}]

Let $\mdp = (S,\Params,\Act,\sinit,P)$ be a binary pMDP with $\Act = \{ a, b \}$.
We construct the simple pMC $\mdp' = (S,\Params',\sinit,P')$.
We introduce fresh variables $\Params_S = \{ x_s \mid s \in S \}$: Thus $\Params' = \Params \cup \Params_S$.
Then, $P'$ is given by 
\[
P'(s,s') = \begin{cases}
  x_{s} & \text{if }P(s,a,s') = 1	,\\
  1-x_{s} & \text{if }P(s,b,s') = 1,\\
  P(s,\act',s') & \text{if }\Act(s) = \{ \act' \}, \act' \in \Act,\\
  0 & \text{otherwise.}
 \end{cases}
\]
It holds that \begin{align*} \{ \mdp^\sched[\inst] \mid \sched \in \Sched_\mdp,\;\inst\in \ParamSpace \} 
& =
 \{ \mdp'[\inst'] \mid \;\inst'\in \ParamSpace_{\mdp'} \text{ with } \inst(\Params_S) \subseteq \{0,1\}  \} \\
& \subseteq 
 \{ \mdp'[\inst'] \mid \;\inst'\in \ParamSpace_\mdp \} .\end{align*}
Therefore, clearly 
\[ \exists \sched \in \Sched,\; \exists \inst \in \ParamSpace_\mdp.\; \Pr_{\mdp[\inst]}^\sched(\lozenge T) \bowtie \frac{1}{2} \implies \exists \inst \in \ParamSpace_{\mdp'} \Pr_{{\mdp'}[\inst]}^\sched(\lozenge T) \bowtie \frac{1}{2}. 
\]%
For the other direction, consider that we can reverse the translation, i.e., replace parameters which only occur at a single state (e.g. $\Params_S$ above) by  randomised choices over actions. 
\begin{align*}
\{ \mdp'[\inst'] \mid \;\inst'\in \ParamSpace_{\mdp'} \}
=	
\{ \mdp^{\tau}[\inst] \mid \;\inst\in \ParamSpace_\mdp\;\tau \in \RSched \}
\end{align*}
Following Rem.~\ref{rem:detsched}, randomised schedulers are dominated by deterministic ones for every instantiated (parameter-free) MDP, thus:
\begin{align*}
 \exists \inst \in \ParamSpace_{\mdp'}.\; \Pr_{{\mdp'}[\inst]}^\sched(\lozenge T) \bowtie \frac{1}{2}
 & \implies 
 \exists \tau \in \RSched,\; \exists \inst \in \ParamSpace_\mdp.\; \Pr_{\mdp[\inst]}^\sched(\lozenge T) \bowtie \frac{1}{2} \\
 & \implies \exists \sched \in \Sched,\; \exists \inst \in \ParamSpace_\mdp.\; \Pr_{\mdp[\inst]}^\sched(\lozenge T) \bowtie \frac{1}{2}.
\end{align*}

\end{proof}

\subsubsection{Proof of Theorem~\ref{thm:etr:mdpsarehard} and Lemma~\ref{lem:etr:bcon4INEQhard}}
\restatebconINEQdef*
\restatebconetrhard*

We reduce from mb4FEAS-c (ETR-hard by Lemma~\ref{lem:etr:mb4feas}) to bcon4INEQ-c, using an adaption of the construction of the proof of \cite[Thm.\ 4.1]{DBLP:journals/mst/SchaeferS17}.	
\restatembFEASdef*

\noindent In the following, we give ETR encodings. 
We use $\doteq$, $\gtrdot$, $\lessdot$, $\leqslant$, $\geqslant$, to clarify the usage of (in)equalities in constraints.
\begin{proof}[Proof of Lemma~\ref{lem:etr:bcon4INEQhard}]
(adapted from~\cite[Thm 4.1]{DBLP:journals/mst/SchaeferS17})
We ask for an $\inst \in [0,1]^\Params$ such that $f[\inst] \leq 0$. As we can assume $f$ to be non-negative, this is equivalent to $f[\inst] = 0$ (cf.~\textbf{P2}-c in the proof of Lemma~\ref{lem:etr:bquadtob4feas}).
We introduce a fresh variable $z$.
Consider the two semi-algebraic sets \[ A = \{ \inst' \in \RR^{|\Params|+1} \mid \inst \models \big( f \doteq z \land \bigwedge_{x \in \Params} 0 \leqslant x \leqslant 1 \big) \}, \] 
and 
\[ B = \{ \inst' \in \RR^{|\Params|+1} \mid \inst \models \big( z \doteq 0 \land \bigwedge_{x \in \Params} 0 \leqslant x \leqslant 1 \big) \}.  \]
Observe that $A$ and $B$ are disjoint, iff no $\inst \in \RR^\Params$ exists such that $f[\inst] = 0$.
We now construct an bcon4INEQ-c instance that is satisfiable if $A$ and $B$ are not disjoint.

The sets $A$ and $B$ are compact\footnote{Notice that we use the same sets $A,B$ also for reducing the open variant, to ensure compactness.}.
If $A$ and $B$ are disjoint, then there is a positive (minimal) distance $\delta$ between $A$ and $B$. By \cite[Corollary 3.8]{DBLP:journals/mst/SchaeferS17}, $\delta$ is at least $2^{{-}2^{L+5}}$, where $L$ is the minimal number of bits to encode $f$ in a sum-of-terms fashion.
Thus, if there exists an point where the distance is between $A$ and $B$ would be smaller than $\delta$, then $A$ and $B$ cannot be disjoint, and hence they must overlap.
By construction, the distance between $A$ and $B$ is given by difference between of $f$ and $0$.  
This observation yields the following constraints
\[
	|f-0| \lessdot \delta
\quad\land\quad
	\delta \lessdot 2^{{-}2^{L+5}} 
\quad\land\quad
	\inst \in [0,1]^\Params
\]
Clearly, if there is a satisfying assignment, then $0 < \delta < 1$, and as $f$ is non-negative, we have $|f| = f$.
To construct (in polynomial time) an equivalent ETR formula, have to reformulate the constant in the second constraint: Therefore, we reformulate $\delta < 2^{{-}2^{L+5}}$ by iterative squaring (or power iteration) to
\[ \delta \lessdot \delta_1 \cdot \delta_1 
\quad\land\quad
   \delta_1 \lessdot \delta_2 \cdot \delta_2
\quad\land\quad
\hdots
\quad\land\quad
   \delta_{L+3} \lessdot \delta_{L+4} \cdot \delta_{L+4} 
\quad\land\quad  
   \delta_{L+4} \lessdot \frac{1}{2}
\] 
which is trivially to reformulate as
\[ \Delta = \left(\delta - \delta_1 \cdot \delta_1 \lessdot 0
\quad\land\quad
\hdots
\quad\land\quad
   \delta_{L+3} - \delta_{L+4} \cdot \delta_{L+4} \lessdot 0
\quad\land\quad  
   \delta_{L+4} - \frac{1}{2} \lessdot 0\right) \]

Together we obtain the following input to bcon4INEQ-c:
\[ f - \delta \lessdot 0 \land \Delta \]
which is has a satisfying $\inst'\colon \Params \cup \{ \delta, \delta_1 \hdots \delta_{L+4} \} \rightarrow [0,1]$ iff there exists $\inst \in [0,1]^\Params$ s.t. $f[\inst] \leq 0$.
All constraints are at most degree 4.

\end{proof}

\restatemdpsetrhard*

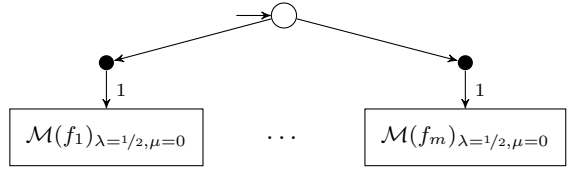
\begin{figure}
\centering
\begin{tikzpicture}
	\node[circle, draw, initial, initial text=] (sinit) {};
	\node[below=1.3cm of sinit] (dots) {$\dots$};
	\node[rectangle, left=0.7cm of dots, draw, inner sep = 6pt] (f1) {\footnotesize$\mdp(f_1)_{\lambda=\nicefrac{1}{2},\mu=0}$};
	\node[rectangle, right=0.7cm of dots, draw, inner sep = 6pt] (fm) {\footnotesize$\mdp(f_m)_{\lambda=\nicefrac{1}{2},\mu=0}$};
	\node[circle, fill, inner sep = 2pt, above=0.5cm of f1] (a1) {};
	\node[circle, fill, inner sep = 2pt, above=0.5cm of fm] (am) {};
	
	\draw[->] (a1) edge  node[right] {\scriptsize$1$} (f1);
	\draw[->] (am) edge  node[right] {\scriptsize$1$} (fm);
	\draw[->] (sinit) edge  node[right] {} (a1);
	\draw[->] (sinit) edge  node[right] {} (am);

\end{tikzpicture}
\caption{Construction for the proof of Thm.~\ref{thm:etr:mdpsarehard}}
\label{fig:mdpsareetrhardrepeated}
\end{figure}

\begin{proof}
ETR-hardness for non-strict inequalities follows from Thm.~\ref{thm:etr:pmcsarehard}. For strict inequalities, we reduce from bcon4INEQ-o/c.
We show the proof for $\exists\forall\reach^>_\mathrm{wd}$.
	We reduce from bcon4INEQ-c. 
For given $f_1,\hdots,f_m$, we construct pMCs $\mdp(f_1)_{\lambda=\nicefrac{1}{2},\mu=0}, \hdots, \mdp(f_m)_{\lambda=\nicefrac{1}{2},\mu=0}$ with target states $T_i$ by applying Lemma~\ref{lem:chonevtrick} to $f_i$ (with $\lambda=\frac{1}{2}$ and $\mu = 0$.
Then, we construct a pMDP by taking the disjoint union of the pMCs and adding a fresh initial state, with non-deterministic actions into each pMC, as outlined in Fig.~\ref{fig:mdpsareetrhardrepeated}.
Formally, let $\mdp(f_i)_{\lambda=\nicefrac{1}{2},\mu=0} = (S_i,\Params,\sinit^i, P_i)$.
We construct a pMDP $\mdp = (S, \Params, \Act, \sinit, P)$ with
\begin{itemize}
\item $S = \bigcup S_i \cup \{ s_0 \}$
\item $\Act = \{ \act_i \mid 1 \leq i \leq m \}$
\item $\sinit = s_0$
\item $P$ given by:
\begin{align*}
P(s,\act,s') = \begin{cases}
 P_i(s,s') &\text{if } s,s' \in S_i, \act = \act_i\text{ for some }i,\\
 1 & \text{if }s=s_0, s'=\sinit^i, \act=\act_i\text{ for some }i\\
 0 & \text{otherwise.} 	
 \end{cases}
\end{align*}
\end{itemize}
We consider target states $T = \bigcup T_i$.
The construction clearly is in polynomial time. 
The $\mdp$ has $m$ schedulers $\sched_1,\hdots,\sched_m$ with $\sched_i = [ s_0 \mapsto \act_i]$ (and all other actions trivially selected).

\noindent By construction, \[
\exists \inst \in \ParamSpace^\mathrm{wd}_\mdp.\;
 \Big(
\Pr_{\mdp[\inst]}^{\sched_i}(\lozenge T) < \frac{1}{2}
\text{ iff }
f_i[\inst] < 0
\Big).
\]
Then, \[ \exists \inst \in \ParamSpace^\mathrm{wd}_\mdp.\; \bigwedge_i \Pr_{\mdp[\inst]}^{\sched_i}(\lozenge T) < \frac{1}{2}\quad \iff \quad \exists \inst \in [0,1]^\Params \; \bigwedge_i f_i[\inst] < 0,  \]
or equivalently, 
\[ \exists \inst \in \ParamSpace^\mathrm{wd}_\mdp,\; \forall \sched \in \Sched.\; \Pr_{\mdp[\inst]}^{\sched}(\lozenge T) < \frac{1}{2}\quad \iff \quad \exists \inst \in [0,1]^\Params \; \bigwedge_i f_i[\inst] < 0,  \]
\end{proof}

%%%%%%%%%%%%%%%%%%%%%%%%%%%%%%%%%%%%%%%%%%%%%%%%%%%%%%%%%%%%%

\paragraph*{NP-hardness for a special case} 
We provide a further NP-hardness result for acyclic pMDPs with just two schedulers:

\begin{lemma}
\label{lem:efnphard}
	$ \exists\forall\reach_\mathrm{gp}^>$ and $ \exists\forall\reach_\mathrm{gp}^\geq$ is NP-hard (even for acyclic pMDPs with just two  schedulers).
	\end{lemma}

\noindent To prove the theorem, we consider pMCs with multiple objectives:
\begin{definition}
Given a pMC with sets of states $T$ and $T'$ and constants $\lambda_1$, $\lambda_2$.\[
 2\exists\reach^{\bowtie \lambda_1 \lambda_2}_\mathrm{gp} 	\stackrel{\mathrm{def}}{\iff} \exists \inst \in \ParamSpace^{\mathrm{gp}}.\; \left( \Pr_{\mdp[\inst]}(\lozenge T_1) \bowtie \lambda_1 \land  \Pr_{\mdp[\inst]}(\lozenge T_2) \bowtie \lambda_2 \right) \]
\end{definition}
The problem $2\exists\reach^{\unrhd}_\text{gp}$ is NP-hard~\cite[Theorem 8]{infocomp}\footnote{The original statement uses one $>$ and one $\geq$ relation, but the given proof does not depend on the (non)strictness.}.

\begin{lemma}
$2\exists\reach^{\bowtie \lambda_1 \lambda_2}_\mathrm{gp} \;\leq_p\; \exists\forall\reach_\mathrm{gp}^{\bowtie}$
	\label{lem:2etoef}
\end{lemma}

\begin{proof}
	Given a pMC $\mdp$ with $T_1, T_2$ and $\lambda_1, \lambda_2 \in \QQ$. We construct a pMDP $\mdp'$.
	We show the construction for $\lambda_2 < \frac{1}{2}$ and $\lambda_1 > \frac{1}{2}$. The construction for the other cases is analogous (using adaptions as outlined in Remark~\ref{rem:fixedlambda}).
	The construction is outlined in Fig.~\ref{fig:2etoef}.

	Formally, we construct $(S',\Params,\{\act_1,\act_2, \act_d\},\sinit',P')$ with 
	$S' = S \times \{1,2\} \cup \{ \sinit, t, \bot \}$, $\sinit'$ and let $T = \{ t \}$.
\[ P'(s,\act,s') = \begin{cases} P(\bar{s},\bar{s}') & \text{if } \act=\act_d, s=\langle \bar{s},i\rangle, s'=\langle\bar{s}',i\rangle \text{ for some }i \in \{1,2\}\\
    1 & \text{if } \act=\act_d, s=\langle \bar{s}, i \rangle, \bar{s} \in T_i \text{ for some }i \in \{1,2\} s'=t \\
    1 & \text{if } \act=\act_d, s=s', s\in \{\bot, t\} \\
    \frac{2}{\lambda_1} & \act=\act_1, s=\sinit, s'=\langle \sinit, 1 \rangle \\
    1-\frac{2}{\lambda_1} & \act=\act_1, s=\sinit, s'=\bot \\
    \frac{\frac{1}{2}-\lambda_2}{1-\lambda_2} & \act=\act_2, s=\sinit, s'=\langle \sinit, 2 \rangle \\
    1-\frac{\frac{1}{2}-\lambda_2}{1-\lambda_2} & \act=\act_2, s=\sinit, s'=t \\
0 & \text{otherwise} 
 \end{cases}
\]
	The pMDP has two schedulers scheduler $\sched_i$ with $i\in \{1,2\}$ s.t. $\sched_i(s) = \act_i$.
	Then it holds that:	
	\begin{align*}
	 & \exists \inst  \in \ParamSpace^{\mathrm{gp}},\;\forall \sched \in \Sched.\;  \Pr_{\mdp'[\inst]}^\sched (\lozenge T) \bowtie \frac{1}{2} \\
	\iff 
		 & \exists \inst \in \ParamSpace^{\mathrm{gp}}.\; \left( \Pr_{\mdp'[\inst]}^{\sched_1}(\lozenge T) \bowtie \frac{1}{2} \land  \Pr_{\mdp'[\inst]}^{\sched_2}(\lozenge T) \bowtie \frac{1}{2} \right)
		 \\
	\iff 
		 & \exists \inst \in \ParamSpace^{\mathrm{gp}}.\; \left( \Pr_{\mdp'[\inst]}^{\sched_1}(\langle\sinit,1\rangle \rightarrow \lozenge T_1) \bowtie \lambda_1 \land  \Pr_{\mdp'[\inst]}^{\sched_2}(\langle\sinit,2\rangle \rightarrow \lozenge T_2) \bowtie \lambda_2 \right)\\
		 \iff 
		 & \exists \inst \in \ParamSpace^{\mathrm{gp}}.\;\left( \Pr_{\mdp[\inst]}(\lozenge T_1) \bowtie \lambda_1 \land  \Pr_{\mdp[\inst]}(\lozenge T_2) \bowtie \lambda_2 \right)
	\end{align*}
\end{proof}

%The idea of the reduction is given by the construction in Fig.~\ref{fig:2etoef}.
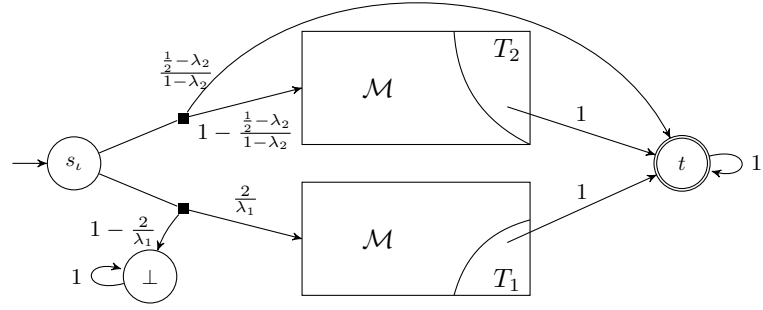
\begin{figure}
\centering
	\begin{tikzpicture}
		
		\draw (1,1) rectangle (4,2.5);
		\draw (1,3) rectangle (4,4.5);
		
		\draw (3,1) edge[bend left=30] (4, 2);
		\draw (3,4.5) edge[bend right] (4, 3);
		
		\node at (2,1.75) {$\mdp$};
		\node at (2,3.75) {$\mdp$};

		\node at (3.7,1.2) {$T_1$};
		\node at (3.7,4.25) {$T_2$};
		
		\node[state,accepting,scale=0.8] (t) at (6,2.75) {$t$};
		\draw[->] (t) edge [loop right] node {\footnotesize$1$} (t);
		
		\node[state,initial, initial text=,scale=0.8] (s) at (-2,2.75) {$\sinit$};
		\node[inner sep=2pt, fill,right=of s, yshift=0.6cm] (a1) {};
		\node[inner sep=2pt, fill,right=of s, yshift=-0.6cm] (a2) {};
		
		\node[state,scale=0.8] (b) at (-1,1.25) {$\bot$};
		\draw[->] (b) edge [loop left] node {\footnotesize$1$} (b);
		
		\draw[-] (s) -- (a1);
		\draw[-] (s) -- (a2);
		\draw[->] (a2) edge node[above] {\footnotesize$\frac{2}{\lambda_1}$} (1,1.75);
		\draw[->] (a1) edge node[below] {\footnotesize$1-\frac{\frac{1}{2}-\lambda_2}{1-\lambda_2}$} (1,3.75);
		\draw[->] (a2) edge[bend right=10] node[left] {\footnotesize$1-\frac{2}{\lambda_1}$} (b);
		\draw[->] (a1) edge[bend left=60] node[left, pos=0.1] {\footnotesize$\frac{\frac{1}{2}-\lambda_2}{1-\lambda_2}$} (t);
		\draw[->] (3.7, 3.5) -- node[above] {\footnotesize$1$} (t);
		\draw[->] (3.7, 1.7) -- node[above] {\footnotesize$1$} (t);

	\end{tikzpicture}
	\caption{Reducing $2\exists\reach^{\bowtie \lambda_1 \lambda_2}_\text{gp}$ to $\exists\forall\reach^{\bowtie}_\text{gp}$ for the case $\lambda_2 < \frac{1}{2} < \lambda_1$}
	\label{fig:2etoef}
\end{figure}

\begin{proof}[Proof of Theorem~\ref{lem:efnphard}]
As corollary to the NP-hardness of $2\exists\reach^{\unrhd \lambda_1 \lambda_2}_\text{gp}$  and Lemma~\ref{lem:2etoef}.
\end{proof}

%%%%%%%%%%%%%%%%%%%%%%%%%%%%%%%%%%%%%%%%%%%%%%%%%%%
\subsubsection{Proof of Theorem \ref{thm:csrgencoding}}
\restatecsrgencoding*

\noindent
For convenience, we define values for both players in a CSRG as follows:
\begin{align*}
  V_\playerI(\csrg) 
    &\coloneqq
      \sup_{\sigma}\Inf_{\tau}\Pr_\csrg^{\sigma, \tau}(\lozenge T)\\
  V_\playerII(\csrg)
    &\coloneqq
    1 - \sup_{\sigma}\Inf_{\tau}\Pr_\csrg^{\sigma, \tau}(\lozenge T).
\end{align*}
We now recall some known results regarding CSRGs.
\begin{theorem}[From~\cite{partha71,partha76}]
	CSRGs enjoy the following properties:
	\label{thm:csrgprops}
	\begin{itemize}
		\item It holds that
		\[
      \sup_{\sigma}\Inf_{\tau}\Pr_\csrg^{\sigma, \tau}(\lozenge T)
		  =
		  \Inf_{\tau}\sup_{\sigma}\Pr_\csrg^{\sigma, \tau}(\lozenge T),
		\]
		and thus
		$V_\playerI(\csrg) + V_\playerII(\csrg) = 1$ (i.e. CSRGs are determined).
		\item There is a stationary strategy $\tau^*$ for $\playerII$ such that
		\[
      \Inf_{\tau}\sup_{\sigma}\Pr_\csrg^{\sigma, \tau}(\lozenge T) = \min_{\tau
      \in \Sigma^\playerII}\max_{\sigma \in \Sigma^\playerI}\Pr_\csrg^{\sigma,
      \tau}(\lozenge T)
		  =\max_{\sigma \in \Sigma^\playerI}\Pr_\csrg^{\sigma, \tau^*}(\lozenge T).\]
	\end{itemize}
\end{theorem}

We can now proceed with the proof of the claim.
\begin{proof}[Proof of Theorem \ref{thm:csrgencoding}]
	Let $\csrg$ be a CSRG with state space $S_\csrg$ and transition probability distributions $P_\csrg(\cdot | s,a,b)$. We define $\mdp = \pmdptuple$ as follows:
	\begin{itemize}
		\item $S = S_\csrg~\uplus~\{s_a \mid s \in S_\csrg, a \in A_s\}~\uplus~\{s_{ab} \mid s \in S_\csrg, a \in A_s, b \in B_s\}$,
		\item $\Params = \biguplus_{s \in S_\csrg} B_s$ and $\Act = \biguplus_{s \in S_\csrg} A_s$ (we may assume w.l.o.g. that the $A_s$ and $B_s$ are pairwise disjoint for all $s \in S_\csrg$),
		\item for all $s,s' \in S_\csrg, a \in A_s, b \in B_s$ we define
		\begin{itemize}
			\item $P(s,a,s_{ab}) = b$,
			%\item $P(s_a, s_{ab}) = b$,
			\item $P(s_{ab}, s') = P_\csrg(s'|s,a,b)$,
		\end{itemize}
		where we omit the actions in states where only one action is available.
	\end{itemize}
	Clearly, the construction can be carried out in polynomial time.
	Note that $\mdp$ is not a simple pMDP if $|B_s| > 2$ for some $s \in S_\csrg$. Later we will argue how it can be made simple. 
	
  The intuition is that there is a one-to-one correspondence between stationary
  strategies $\tau$ for $\playerII$ in $\csrg$ and parameter valuations $\inst
  \in \ParamSpace_{\mdp}^{\mathrm{wd}}$: A parameter $b \in \Params$ corresponds
  to a player-$\playerII$ action $b \in B_s$ for some $s\in S_\csrg$. Using
  $\tau$, we define the valuation $\inst = \phi(\tau) \in
  \ParamSpace^{\mathrm{wd}}$ as $\inst(b) = \tau(b|s)$, where $\tau(b|s)$ is the
  probability assigned to $b$ in state $s$ under $\tau$. Note that this
  construction also works in the opposite direction, which we will denote by
  $\phi^{-1}(\inst)$.
	
  Formally, we show that for stationary strategies $\tau$ of player $\playerII$,
  the instantiations $\csrg^\tau$ and $\mdp[\inst]$ with $\inst = \phi(\tau)$
  coincide as MDPs: For $\csrg^\tau$ we have, similar to \eqref{eq:csrginst},
  that
	\begin{align*}
		P_{\csrg^\tau}(s,a,s')
		= \sum_{b \in B_s} \tau(b|s)P_\csrg(s'|s,a,b)
	\end{align*}
	and for $\mdp[\inst]$ we conclude from the definition above that
	\begin{align*}
	\label{eq:mdpcoincide}
		P_{\mdp[\inst]}(s,a,s')
		= \sum_{b \in B_s} P(s,a,s_{ab}) P(s_{ab}, s')
		= \sum_{b \in B_s} \tau(b|s) P_\csrg(s'|s,a,b)
		= P_{\csrg^\tau}(s,a,s')
	\end{align*}
	where we merged states $s_{ab}$ and $s'$ into a single one. Let $\tau^*$ be like in Thm. \ref{thm:csrgprops} and $\inst^* = \phi(\tau^*)$. Then we have that
	\[
	\min_{\tau \in \Sigma^\playerII}\max_{\sigma \in \Sigma^\playerI}\Pr_\csrg^{\sigma, \tau}(\lozenge T)
	=\max_{\sigma \in \Sigma^\playerI} \Pr_\csrg^{\sigma, \tau^*}(\lozenge T)
	\overset{(*)}{=}\max_{\sched \in \Sched}\Pr_{\mdp[\inst^*]}^\sched(\lozenge T)
	\geq \min_{\inst \in \ParamSpace^{\mathrm{wd}}} \max_{\sched \in \Sched}\Pr_{\mdp[\inst]}^\sched(\lozenge T),
	\]
	where $(*)$ is true because the induced MDPs coincide as shown above.
	If we let $\inst^* \in \ParamSpace_{\mdp}^{\mathrm{wd}}$ be a minimizing parameter valuation and $\tau^* = \phi^{-1}(\inst^*)$, then it holds conversely that
	\[
	\min_{\inst \in \ParamSpace^{\mathrm{wd}}} \max_{\sched \in \Sched}\Pr_{\mdp[\inst]}^\sched(\lozenge T)
	= \max_{\sched \in \Sched}\Pr_{\mdp[\inst^*]}^\sched(\lozenge T)
	\overset{(*)}{=} \max_{\sigma \in \Sigma^\playerI}\Pr_\csrg^{\sigma, \tau^*}(\lozenge T)
	\geq \min_{\tau \in \Sigma^\playerII}\max_{\sigma \in \Sigma^\playerI}\Pr_\csrg^{\sigma, \tau}(\lozenge T)
	\]
	and the claim follows.
	
	We now argue how $\mdp$ can be made simple\footnote{The construction below requires all parameters to occur in the same combination, and then mimics the construction to binary pMDPs in Lem~\ref{lem:makebinary} (where it is applied to action rather than to parametric transitions).}. 
	We give a slightly more general argument: Consider a pMC `gadget' $\mathcal{H}$ with a single state $s$ which has multiple outgoing transitions labeled with parameters $x_1,\ldots,x_k$, $k > 2$. We may iterate the construction suggested in Fig.~\ref{fig:nonsimpletosimple} to obtain a simple gadget $\mathcal{H}'$ by introducing $k-2$ new states. 
	In order to retain an equivalent behaviour with respect to reachability of the $k$ output states of $\mathcal{H}$, a given valuation $\inst$ for the $x_1,\ldots,x_k$ has to be adapted to a new valuation $\inst'$ for $\mathcal{H}'$ as follows: $\inst'(x_1) = \inst(x_1)$ and for all $1 < i \leq k$,
	\[
	\inst(x_i) = \begin{cases}
		\frac{\inst(x_i)}{1-\inst(x_1)}, &\text{ if } \inst(x_1) < 1, \\
		0, &\text{ else.}
	\end{cases}
	\]
Similarly, we can obtain a valuation $\inst$ from a given $\inst'$. We apply this transformation to all states $S_\csrg$ with $|B_s| > 2$. This is also possible in polynomial time.
\end{proof}

	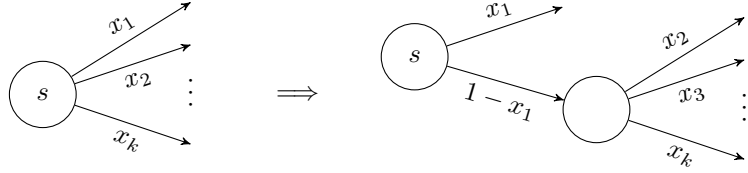
\begin{figure}[tb]
\centering
\begin{tikzpicture}
	\node[state] (s) {$s$};
	\node[right=15mm of s, yshift=13mm] (s1) {};
	\node[right=15mm of s, yshift=7mm] (s2) {};
	\node[right=15mm of s, yshift=-3mm,rotate=90] (dots) {$\dots$};
	\node[right=15mm of s, yshift=-7mm] (sk) {};
	
	\draw[->] (s) -- node[sloped,anchor=center,above] {$x_1$} (s1);
	\draw[->] (s) -- node[sloped,anchor=center,below] {$x_2$} (s2);
	\draw[->] (s) -- node[sloped,anchor=center,below] {$x_k$} (sk);
	
	\node[right=25mm of s] {$\Longrightarrow$};

	\node[state, right=40mm of s, yshift=5mm] (t) {$s$};
	\node[right=15mm of t, yshift=7mm] (t1) {};
	\node[state,right=15mm of t, yshift=-7mm] (t2) {};
	\node[right=15mm of t2, yshift=13mm] (r1) {};
	\node[right=15mm of t2, yshift=7mm] (r2) {};
	\node[right=15mm of t2, yshift=-3mm,rotate=90] (dots) {$\dots$};
	\node[right=15mm of t2, yshift=-7mm] (rk) {};
	
	\draw[->] (t) -- node[sloped,anchor=center,above] {$x_1$} (t1);
	\draw[->] (t) -- node[sloped,anchor=center,below] {$1-x_1$} (t2);
	
	\draw[->] (t2) -- node[sloped,anchor=center,above] {$x_2$} (r1);
	\draw[->] (t2) -- node[sloped,anchor=center,below] {$x_3$} (r2);
	\draw[->] (t2) -- node[sloped,anchor=center,below] {$x_k$} (rk);
\end{tikzpicture}
\caption{Construction used to obtain a simple pMDP in the proof of Thm. \ref{thm:csrgencoding}.}

\label{fig:nonsimpletosimple}
\end{figure}

%%%%%%%%%%%%%%%%%%%%%%%%%%%%%%%%%%%%%%%%%%%%%%%%%%%%%%%%%%%%%

\subsubsection{Proof of Corollary~\ref{thm:csrghardness}}
\restatecsrghardness*
\begin{proof}
	As a simple corollary to Thm. \ref{thm:csrgencoding}:
	Recall from Thm. \ref{thm:csrgprops} that
	\[
	V(\csrg) = \min_{\tau \in \Sigma^\playerII}\max_{\sigma \in \Sigma^\playerI}\Pr_\csrg^{\sigma, \tau}(\lozenge T).
	\]
	Now use Thm. \ref{thm:csrgencoding} to obtain an $\mdp$ in polynomial time with
	\[
	\min_{\tau \in \Sigma^\playerII}\max_{\sigma \in \Sigma^\playerI}\Pr_\csrg^{\sigma, \tau}(\lozenge T)
	=
	\min_{\inst \in \ParamSpace^{\mathrm{wd}}}\max_{\sched \in \Sched}\Pr_{\mdp[\inst]}^\sched(\lozenge T).
	\]
	The claim follows noticing that
	\[
	\min_{\inst \in \ParamSpace^{\mathrm{wd}}}\max_{\sched \in \Sched}\Pr_{\mdp[\inst]}^\sched(\lozenge T) \unlhd \lambda
	\iff
	\exists \inst \in \ParamSpace^{\mathrm{wd}}, \forall \sched \in \Sched .\;\Pr_{\mdp[\inst]}^\sched(\lozenge T) \unlhd \lambda.
	\]
\end{proof}

%%%%%%%%%%%%%%%%%%%%%%%%%%%%%%%%%%%%%%%%%%%%%%%%%%%%%%%%%%%%%
%%%%%%%%%%%%%%%%%%%%%%%%%%%%%%%%%%%%%%%%%%%%%%%%%%%%%%%%%%%%%

\subsection{Robust reachability}

%%%%%%%%%%%%%%%%%%%%%%%%%%%%%%%%%%%%%%%%%%%%%%
\subsubsection{Proof of Theorem.~\ref{thm:fp_e_a_robreach_np_compl}}\label{App:fp-esfp-np}

\restatefprobreachnpcomp*

\begin{proof}
	We provide a reduction from 3-SAT. Let $\phi$ be a 3-SAT formula with clauses $\phi_1 \wedge \ldots \wedge \phi_m = \phi$ and variables $X_1,...,X_n$. We may restrict ourselves to the case where there exists no literal (a variable or negated variable) that is contained in all clauses, as otherwise $\phi$ is trivially satisfiable.
	
	For a literal $l$ we let $Sat(l) \coloneqq \{\phi_i \mid 1 \leq i \leq
m \text{ and } l \in \phi_i\}$ be the set of all clauses satisfied by $l$ and $Sat^\mathcal{C}(l)$ its complement.
	
	We first construct a \emph{nonsimple} pMDP $\mdp = \pmdptuple$ as depicted in Fig. \ref{fig:robstratNPexample}.
	Formally, let
  $S \coloneqq \{X_1,...,X_n\} \uplus \{\sinit,T,F\}$, where $T$ is the target and $F$
  is a sink, $X_i \coloneqq \{x_i\}$ and $\Act \coloneqq \{\alpha,\beta\}$.
	In order to define $P$ we consider the polynomials
\begin{align*}
	f_{i,\alpha} \coloneqq \frac{1}{2}\bigg(1 + \hspace{-3mm}\prod_{\phi_k \in Sat^\mathcal{C}(X_i)}\left(x-\frac{k}{m+1}\right)^2\bigg),\ \ 
	f_{i,\beta} \coloneqq \frac{1}{2}\bigg(1 + \hspace{-3mm}\prod_{\phi_k \in Sat^\mathcal{C}(\overline{X}_i)}\left(x-\frac{k}{m+1}\right)^2\bigg).	
\end{align*}
The only states where more than one action is enabled are the $X_1,\hdots,X_n$. In all other cases we omit the action. For all $1 \leq i \leq n$ we let
\begin{itemize}
	\item $P(\sinit,X_i) \coloneqq \frac{1}{n}$,
	\item $P(X_i,\alpha,T) \coloneqq f_{i,\alpha}$,
	\item $P(X_i,\beta,T) \coloneqq f_{i,\beta}$.
\end{itemize}
All the probability mass left unspecified in any state-action pair leads to the sink $F$. 
For the well-defined (i.e. not necessarily graph-preserving) case, we need to show that 
\begin{align}
	\label{eq:proofobligation}
	\phi \text{ is sat} \iff \exists \sched \in \Sched , \forall \inst \in \ParamSpace_{\mdp}^{\text{wd}} .\; \Pr_{\mdp[\inst]}^\sched(\lozenge T) > \frac{1}{2}.
\end{align}
The intuition is that there is a one-to-one-correspondence between the variable assignments of $\phi$ and the schedulers of $\mdp$: From an assignment $\mathcal{I}\colon \{X_1,...,X_n\} \rightarrow \{0,1\}$, a scheduler $\sched$ is obtained via $\sched(X_i) \coloneqq \alpha$ if $\mathcal{I}(X_i) = 1$ and $\sched(X_i) \coloneqq \beta$ if $\mathcal{I}(X_i) = 0$. We will refer to $\sigma$ as the scheduler corresponding to $\mathcal{I}$. 
Conversely, for every scheduler there is an assignment corresponding to it.
Moreover, note that the (global) minima of the polynomials $f_{i,\gamma}$, $\gamma \in \{\alpha,\beta\}$, are given by the set of points
$\{\big(\frac{k}{m+1}, \frac{1}{2}\big) \mid k \in Sat^\mathcal{C}(l_i)\}$
where $l_i = X_i$ if $\gamma = \alpha$ and $l_i = \overline{X}_i$ if $\gamma = \beta$.
By construction of $\mdp$ we have that for a fixed scheduler $\sched$
\begin{align}
	\label{eq:eqforf}
	\Pr^\sched(\lozenge T)
	= \frac{1}{n}\bigg(\sum_{\sched(X_i) = \alpha}f_{i,\alpha} + \sum_{\sched(X_i) = \beta}f_{i,\beta}\bigg)
	= \frac{1}{n} \sum_{i=1}^n \Delta_i
	\eqqcolon f,
\end{align}
where the polynomials $\Delta_i$ are defined in terms of the assignment $\mathcal{I}$ corresponding to $\sched$:
\begin{align*}
	\Delta_i \coloneqq \mathcal{I}(X_i) f_{i,\alpha} + (1-\mathcal{I}(X_i)) f_{i,\beta}.
\end{align*}
We observe that $f[\inst] \geq \frac{1}{2}$ for all $\inst \in
\ParamSpace_{\mdp}^{\text{wd}}$ and $f[\inst] = \frac{1}{2}$ can only be achieved if
$\inst(x) = k/(m+1)$ for some $1 \leq k\leq m$.
Notice that since $f(0) > \frac{1}{2}$ and $f(1) > \frac{1}{2}$ holds in particular, we do not need to show \eqref{eq:proofobligation} separately for the graph-preserving case.

\medskip
\noindent\emph{$\Longrightarrow$ of~\eqref{eq:proofobligation}}
Let $\mathcal{I}\colon \{X_1,...,X_n\} \rightarrow \{0,1\}$ be a satisfying assignment of $\phi$ and $\sched$ its corresponding scheduler. Since $\mathcal{I}$ is satisfying it holds that for every clause $\phi_k$, $1 \leq k \leq m$, there exists a variable $X_j$ such that either
\begin{enumerate}
 \item $\mathcal{I}(X_j) = 1$ and $\phi_k \in Sat(X_j)$ $ (\ \Longleftrightarrow \phi_k \notin Sat^\mathcal{C}(X_j)\ )$ or
 \item $\mathcal{I}(X_j) = 0$ and $\phi_k \in Sat(\overline{X}_j)$ $(\ \Longleftrightarrow \phi_k \notin Sat^\mathcal{C}(\overline{X}_j)\ )$.
\end{enumerate}
Hence in the first case, $f_{j,\alpha}[\frac{k}{m+1}] > \frac{1}{2}$ because $\frac{k}{m+1}$ is not a minimum of $f_{j,\alpha}$ by definition and in the second case $f_{j,\beta}[\frac{k}{m+1}] > \frac{1}{2}$ analogously. Thus in both cases
$\Delta_j[\frac{k}{m+1}] > \frac{1}{2}$.
We conclude
\begin{align*}
f\big[\frac{k}{m+1}\big]
\overset{\eqref{eq:eqforf}}{=} \frac{1}{n}\sum_{i=1}^n\Delta_i\big[\frac{k}{m+1}\big]
\geq \frac{1}{n} \bigg(\frac{n-1}{2} + \Delta_j\big[\frac{k}{m+1}\big]\bigg)
> \frac{1}{n} \bigg(\frac{n-1}{2} + \frac{1}{2} \bigg)
= \frac{1}{2}.
\end{align*}
Thus as $k/(m+1)$ is a global minimum of $f$ it follows that $f[\inst] > \frac{1}{2}$ for all $\inst \in \ParamSpace_{\mdp}^{\text{gp}}$.

\medskip
\noindent\emph{$\Longleftarrow$ of~\eqref{eq:proofobligation}}
Suppose that $\phi$ is unsatisfiable and let $\sched$ by any strategy. Then for its corresponding assignment $\mathcal{I}$ there is a clause $\phi_k$ which remains unsatisfied, i.e. for all variables $X_i$ we have, similar as before, that either
\begin{enumerate}
	\item $\mathcal{I}(X_i) = 1$ and $\phi_k \notin Sat(X_i)$ $ (\ \Longleftrightarrow \phi_k \in Sat^\mathcal{C}(X_i)\ )$ or
 	\item $\mathcal{I}(X_i) = 0$ and $\phi_k \notin Sat(\overline{X}_i)$ $(\ \Longleftrightarrow \phi_k \in Sat^\mathcal{C}(\overline{X}_i)\ )$,
\end{enumerate}
which implies that $\Delta_j[\frac{k}{m+1}] = \frac{1}{2}.$
So we see that at $\inst(x) = (k/(m+1))$, the reachability probability cannot be greater than $\frac{1}{2}$:
\begin{align*}
f\big[\frac{k}{m+1}\big]
\overset{\eqref{eq:eqforf}}{=} \frac{1}{n}\sum_{i=1}^n \Delta_j\big[\frac{k}{m+1}\big]
= \frac{1}{n} \cdot \frac{n}{2}
= \frac{1}{2}.
\end{align*}

\medskip
\noindent\emph{Simple pMDPs.} We can now further modify our construction in order to obtain a \emph{simple} pMDP. For this sake we replace the transitions with probability $f_{i,\gamma}$, $\gamma \in \Act$, with the gadgets constructed in the proof of Lemma~\ref{lem:chonevtrick} (choosing $\lambda = \frac{1}{2}$ as threshold). The idea is that those gadgets do not change the global minima (even though they do change the exact function $f_{i,\gamma}$).
\end{proof}

\begin{figure}
\centering
\begin{tikzpicture}[initial text=]
	\node[state,initial] (sinit) {$\sinit$};
	
	\node[state, right=of sinit, yshift=10mm] (x_1) {$X_1$};
	\node[right=of sinit,xshift=4mm,rotate=90, anchor=center] (dots) {$\dots$};
	\node[state, right=of sinit, yshift=-10mm] (x_n) {$X_n$};
	
	\node[circle,fill=black,inner sep=1pt,right of=x_1,xshift=4mm,yshift=3mm] (x_1_alpha) {};
	\node[circle,fill=black,inner sep=1pt,right of=x_1,xshift=4mm,yshift=-3mm] (x_1_beta) {};
	\node[circle,fill=black,inner sep=1pt,right of=x_n,xshift=4mm,yshift=3mm] (x_n_alpha) {};
	\node[circle,fill=black,inner sep=1pt,right of=x_n,xshift=4mm,yshift=-3mm] (x_n_beta) {};
	
	\node[state,accepting,right=of sinit,xshift=35mm] (T) {$T$};
	
	%%%%%%%%%
	
	\draw[->] (sinit) -- node[sloped,anchor=center,above] {$\nicefrac{1}{n}$} (x_1);
	\draw[->] (sinit) -- node[sloped,anchor=center,below] {$\nicefrac{1}{n}$} (x_n);
	
	\draw[->] (x_1) -- node[above] {$\alpha$} (x_1_alpha);
	\draw[->] (x_1) -- node[below] {$\beta$} (x_1_beta);
	\draw[->] (x_n) -- node[above] {$\alpha$} (x_n_alpha);
	\draw[->] (x_n) -- node[below] {$\beta$} (x_n_beta);
	
	\draw[->] (x_1_alpha) edge[bend left=20] node[sloped,anchor=center,above,above] {$f_{1,\alpha}$} (T);
	\draw[->] (x_1_beta) -- node[sloped,anchor=center,above,above] {$f_{1,\beta}$} (T);
	\draw[->] (x_n_alpha) -- node[sloped,anchor=center,above,below] {$f_{n,\alpha}$} (T);
	\draw[->] (x_n_beta) edge[bend right=20] node[sloped,anchor=center,above,below] {$f_{n,\beta}$} (T);
\end{tikzpicture}
\caption{Construction used in the proof of Thm. \ref{thm:fp_e_a_robreach_np_compl}.}

\label{fig:robstratNPexample}
\end{figure}
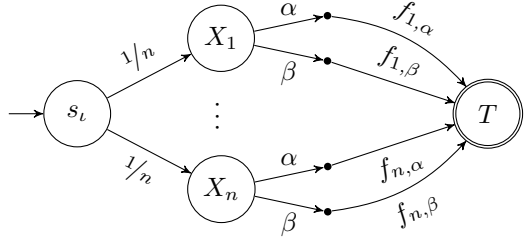

\pagebreak
\section{Full ETR encoding}
\label{app:fullencoding}
In the following, we give ETR encodings. 
We use $\doteq$, $\gtrdot$, $\lessdot$, $\leqslant$, $\geqslant$, to clarify the usage of (in)equalities in constraints.
The encodings are presented in a linear fashion, the first encoding is presented in more detail.

\subparagraph*{Auxiliary notation.}

Let the set $S^{\inst}_{=0,\forall \sched}$ describe all states that reach the target with probability zero using any scheduler, formally \[S^{\inst}_{=0,\forall \sched} \coloneqq \{ s' \mid \forall \sched \in \Sched.\; \Pr^{\sched}_\mdp[\inst](s' \rightarrow T) = 0 \}.\]
Let the set $S^{\inst}_{=0,\exists \sched}$ describe all states that reach the target with probability zero using some scheduler, formally \[S^{\inst}_{=0,\exists \sched} \coloneqq \{ s' \mid \exists \sched \in \Sched.\; \Pr^{\sched}_\mdp[\inst](s' \rightarrow T) = 0 \}.\]

In particular, for graph consistent instantiations $\inst, \inst'$ it holds that  $S^{\inst}_{=0,\forall \sched} = S^{\inst'}_{=0,\forall \sched} $ and $S^{\inst}_{=0,\exists \sched} = S^{\inst'}_{=0,\exists \sched} $. 
We define $S^{\mathrm{gp}}_{=0,\forall \sched} \coloneqq \bigcup_{\inst \in \ParamSpace^\mathrm{gp}} S^{\inst}_{=0,\forall \sched}$ and $S^{\mathrm{gp}}_{=0,\exists \sched} \coloneqq \bigcup_{\inst \in \ParamSpace^\mathrm{gp}} S^{\inst}_{=0,\exists \sched}$.
These sets can be efficiently precomputed via standard graph-based algorithms~\cite{BK08}.

\subsection{$\exists\forall$ with upper bounds}
We consider encodings for: 
\begin{align*}
  \exists\;\inst \in \ParamSpace^{*},
  \forall\;\sched \in \Sched.\;
    \Pr_{\mdp[\inst]}^{\sched}(\lozenge T) \unlhd \frac{1}{2} 
\end{align*}
We thus search for an encoding that checks whether there exist parameter values such that even the maximising scheduler is below the threshold.

\paragraph*{Graph preserving}
We consider:
\begin{align}
\label{eq:defproblemgpub}
  \exists\;\inst \in \ParamSpace^{\mathrm{gp}},
  \forall\;\sched \in \Sched.\;
    \Pr_{\mdp[\inst]}^{\sched}(\lozenge T) < \frac{1}{2} \end{align}
    The non-strict version is analogously encoded.

\subparagraph*{Encoding.}
We take the following straightforward generalisation of the Bellman inequalities~\cite{puterman05,BK08} from~\cite{prophesy_journal}.
We use variables $\{ v_s \mid s \in S \} \cup \Params$,
where variables $v_s$ shall encode the probability to reach the target (for any fixed values for the variables in $\Params$).

We let $S_{?} = S \setminus (T \cup S^{\mathrm{gp}}_{=0,\forall \sched})$.
We encode the pMDP $\mdp$ by the conjunction $\Upphi_{\forall,\unlhd}^\mathrm{gp}(\mdp)$ of:
\begin{align*}
		& \bigwedge_{s \in T}  v_s \doteq 1 
		\quad\land\quad \bigwedge_{s \in  S^{\mathrm{gp}}_{=0,\forall \sched}}  v_s \doteq 0 
		\quad\land\quad \bigwedge_{s \in S_{?}} \bigwedge_{\act \in \Act(s)} v_s \geqslant \sum_{s' \in S} P(s,\act,s') \cdot v_{s'} 			\end{align*}
and encode with $\Upphi(\ParamSpace^\mathrm{gp})$ the parameter space:
\begin{align*}
	\bigwedge_{x\in\Params} 0 \lessdot x \lessdot 1.
\end{align*}

Together, the formula $\Upphi_{\forall,\unlhd}^\mathrm{gp}(\mdp) \land \Upphi(\ParamSpace^\mathrm{gp}) \land v_{\sinit} \lessdot \frac{1}{2}$ encodes\footnote{For a non-strict bound, substitute $\lessdot$ with $\leqslant$} precisely \eqref{eq:defproblemgpub}.
That is, if there is a satisfying solution, then the assignments to the parameters precisely encode parameter values s.t.\ for all schedulers the reachability probability is $< \frac{1}{2}$. 
Likewise, if there is no satisfying solution, there are no parameter values inducing an MDP fulfilling the reachability constraint.

Target states reach the target with probability $1$, and when under all strategies the probability to reach a target is $0$, then it surely is $0$ for the maximising scheduler; vice versa, if there is a scheduler for which the probability to reach the target is positive, then the probability will not be $0$ under the maximising scheduler.   
The solver tries to assign sufficiently small values to the states in order to satisfy $v_{\sinit} \lessdot \frac{1}{2}$, yet has to assign at least what each action locally yields, thereby assigning at least the value from the maximising action.

\paragraph*{Well-defined}
We construct an encoding for $\exists\forall\reach^{\unlhd}_\mathrm{wd}$, that is for a pMDP $\mdp$:
\begin{align}
\label{eq:defproblemwdub}
  \exists\;\inst \in \ParamSpace^{\mathrm{wd}},
  \forall\;\sched \in \Sched.\;
    \Pr_{\mdp[\inst]}^{\sched}(\lozenge T) < \frac{1}{2}. \end{align}

We cannot reuse the encoding from the graph-preserving case. The sets $P^\inst_{=0}$ of transitions that become zero vary, and thus there is not a single set  $S^{\inst}_{=0,\forall \sched}$.
Furthermore, the number of different sets is exponential in the number of parameters, thus it cannot be efficiently precomputed.
We therefore encode their computation into the encoding.
As we consider maximising schedulers, we have to assign probability $0$ exactly iff there is no path to the target states.
Notice that we only have to consider finite paths. There are various ways to
encode the computation of these states, e.g.\ via~\cite{DBLP:conf/aaai/ChatterjeeCD16} for POMDPs.
 We  construct this encoding using an idea from~\cite{DBLP:journals/tcs/WimmerJAKB14} for counterexamples in parameter-free MDPs.
 The idea is that a state has a path to the target if it has a successor state $s'$ (under the current parameter assignment) which  has a path to the target, and that $s'$ is closer to the target (to prevent cyclic arguments). 
 To encode that a state is closer to the target, we use variables to rank the states along a path: a path to the target gets strictly increasing rank along its states, preventing cycles.

\subparagraph*{Encoding.}
We use the following set of additional variables: $\{ p_s, r_s \mid s \in S \}$. 
Here, we assume that $p_s$ are Boolean variables.
Intuitively, variable $p_s$ being true means that, for any fixed parameter assignment, state $s$ has a positive probability to reach the target, i.e., a path to the target.
We encode being closer to the target by the auxiliary variable $r_s$: If the
value $r_s$ is larger than $r_{s'}$ than it must be closer to the target. 
We encode the pMDP $\mdp$ as the 
conjunction $\Upphi_{\forall,\unlhd}^\mathrm{wd}(\mdp)$ of:

\begin{align*}
		& \bigwedge_{s \in T}  v_s \doteq 1 
		\quad\land\quad \bigwedge_{s \in S\setminus T} \neg p_s \rightarrow v_s \doteq 0 
		\quad\land\quad \bigwedge_{s \in S \setminus T} \left( p_s  \rightarrow \bigwedge_{\act \in \Act(s)} v_s \geqslant \sum_{s' \in S} P(s,\act,s') \cdot v_{s'} \right) \\ 
		& \bigwedge_{s \in S \setminus T} \left( p_s \leftrightarrow \bigvee_{\act \in \Act(s)}\bigvee_{s'\in S} \left( P(s,\act,s') \gtrdot 0 \rightarrow \left( p_{s'} \land r_{s} \lessdot r_{s'} \right) \right)	\right),
	\end{align*}%
	and encode with $\Upphi(\ParamSpace^\mathrm{wd})$ the parameter space:
\begin{align*}
	\bigwedge_{x\in\Params} 0 \leqslant x \leqslant 1.
\end{align*}
Together, the formula $\Upphi^\mathrm{wd}_{\unlhd}(\mdp) \land \Upphi(\ParamSpace^\mathrm{wd}) \land v_{\sinit} \lessdot \frac{1}{2}$ encodes precisely \eqref{eq:defproblemwdub}. 	
In particular, for a fixed assignment to $\Params$, the states such that $p_s$ is assigned true are exactly the states $S^{\inst}_{=0,\forall \sched}$, and the encoding then is correct following the reasoning from the graph-preserving case.

\subsection{$\exists\forall$ with lower bounds}
We consider encodings for: 
\begin{align*}
  \exists\;\inst \in \ParamSpace^{*},
  \forall\;\sched \in \Sched.\;
    \Pr_{\mdp[\inst]}^{\sched}(\lozenge T) \unrhd \frac{1}{2} \end{align*}

We thus search for an encoding that checks whether there exist parameter values such that even the minimising scheduler is above the threshold.

\paragraph*{Graph preserving}
We consider:
\begin{align}
\label{eq:defproblemgplb}
  \exists\;\inst \in \ParamSpace^{\mathrm{gp}},
  \forall\;\sched \in \Sched.\;
    \Pr_{\mdp[\inst]}^{\sched}(\lozenge T) > \frac{1}{2} \end{align}
    The non-strict version is analogously encoded.

We only require a slight adaption to the upper-bounded case, as we are now considering minimising schedulers.
	We let $S_{?} = S \setminus (T \cup S^{\mathrm{gp}}_{=0,\exists \sched})$. We encode the pMDP $\mdp$ by the conjunction $\Upphi_{\forall,\unrhd}^\mathrm{gp}(\mdp)$ of:
\begin{align*}
		& \bigwedge_{s \in T}  v_s \doteq 1 
		\quad\land\quad \bigwedge_{s \in  S^{\mathrm{gp}}_{=0,\exists \sched}}  v_s \doteq 0 
		\quad\land\quad \bigwedge_{s \in S_{?}} \bigwedge_{\act \in \Act(s)} v_s \leqslant \sum_{s' \in S} P(s,\act,s') \cdot v_{s'}. 			\end{align*}
%and encode with $\Upphi(\ParamSpace^\mathrm{gp})$ the parameter space:
%\begin{align*}
%	\bigwedge_{x\in\Params} 0 \lessdot x \lessdot 1
%\end{align*}
%
Together, the formula $\Upphi_{\forall,\unrhd}^\mathrm{gp}(\mdp) \land \Upphi(\ParamSpace^\mathrm{gp}) \land v_{\sinit} \gtrdot \frac{1}{2}$ encodes\footnote{For a non-strict bound, substitute $\gtrdot$ with $\geqslant$} precisely \eqref{eq:defproblemgplb}.	

\paragraph*{Well-defined}

We consider:
\begin{align}
\label{eq:defproblemwdlb}
  \exists\;\inst \in \ParamSpace^{\mathrm{wd}},
  \forall\;\sched \in \Sched.\;
    \Pr_{\mdp[\inst]}^{\sched}(\lozenge T) > \frac{1}{2} \end{align}
    The non-strict version is analogously encoded.

\subparagraph*{Encoding.}
We use the following set of additional variables: $\{ p_s, r_s \mid s \in S \}$. 
Here, we assume that $p_s$ are Boolean variables.
Intuitively, variable $p_s$ is true means that, for any fixed parameter assignment, state $s$ has a positive probability to reach the target irrespectively of the selected action, i.e., a path to the target starting with any action.
We encode being closer to the target by the auxiliary variable $r_s$: If the
value $r_s$ is larger than $r_{s'}$ than it must be closer to the target. 
We encode the pMDP $\mdp$ by the conjunction $\Upphi_{\forall,\unrhd}^\mathrm{wd}(\mdp)$ of:

\begin{align*}
		& \bigwedge_{s \in T}  v_s \doteq 1 
		\quad\land\quad \bigwedge_{s \in S\setminus T} \neg p_s \rightarrow v_s \doteq 0 
		\quad\land\quad \bigwedge_{s \in S \setminus T} \left( p_s  \rightarrow \bigwedge_{\act \in \Act(s)} v_s \leqslant \sum_{s' \in S} P(s,\act,s') \cdot v_{s'} \right) \\ 
		& \bigwedge_{s \in S \setminus T} \left( p_s \leftrightarrow \bigwedge_{\act \in \Act(s)}\bigvee_{s'\in S} \left( P(s,\act,s') \gtrdot 0 \rightarrow \left( p_{s'} \land r_{s} \lessdot r_{s'} \right) \right)	\right).
	\end{align*}%
%	and encode with $\Upphi(\ParamSpace^\mathrm{wd})$ the parameter space:
%\begin{align*}
%	\bigwedge_{x\in\Params} 0 \leqslant x \leqslant 1
%\end{align*}
Together, the formula $\Upphi_{\forall,\unrhd}^\mathrm{wd}(\mdp) \land \Upphi(\ParamSpace^\mathrm{wd}) \land v_{\sinit} \gtrdot \frac{1}{2}$ encodes precisely \eqref{eq:defproblemwdlb}. 	
In particular, for a fixed assignment to $\Params$, the states such that $p_s$ is assigned true are exactly the states $S^{\inst}_{=0,\exists \sched}$, and the encoding then is correct following the reasoning from the graph-preserving case.

\subsection{$\exists\exists$ with upper bounds}

\paragraph*{Graph preserving}

We consider encodings for: 
\begin{align}
  \exists\;\inst \in \ParamSpace^{gp},
  \exists\;\sched \in \Sched.\;
    \Pr_{\mdp[\inst]}^{\sched}(\lozenge T) \unlhd \frac{1}{2} 
    \label{eq:defproblemexgpub}
\end{align}
We thus search for an encoding that checks whether there exist parameter values such that the minimising scheduler is below the threshold.

We take the encoding from~\cite{prophesy_journal}. We let $S_{?} = S \setminus (T \cup S^{\mathrm{gp}}_{=0,\exists \sched})$. 
We encode the pMDP $\mdp$ by the conjunction $\Upphi_{\exists}^\mathrm{gp}(\mdp)$ of:
\begin{align*}
		& \bigwedge_{s \in T}  v_s \doteq 1 
		\quad\land\quad \bigwedge_{s \in  S^{\mathrm{gp}}_{=0,\exists \sched}}  v_s \doteq 0 
		\quad\land\quad \bigwedge_{s \in S_{?}} \bigvee_{\act \in \Act(s)} v_s \doteq \sum_{s' \in S} P(s,\act,s') \cdot v_{s'}. 			\end{align*}
Together, the formula $\Upphi_{\exists,\unlhd}^\mathrm{gp}(\mdp) \land \Upphi(\ParamSpace^\mathrm{gp}) \land v_{\sinit} \lessdot \frac{1}{2}$ encodes\footnote{For a non-strict bound, substitute $\lessdot$ with $\leqslant$} precisely \eqref{eq:defproblemexgpub}.

\paragraph*{Well defined}

Omitted. ETR membership follows from Lemma~\ref{lem:existwdgeeqexistwdge} and the encoding for $\exists\forall$.
An direct encoding is a straightforward combination of the ingredients above.

\subsection{$\exists\exists$ with lower bounds}

Omitted. ETR membership follows from Lemma~\ref{lem:existwdgeeqexistwdge}, Corollary~\ref{cor:eeggpeqerggwd}, and the encoding for $\exists\forall$.
An direct encoding is a straightforward combination of the ingredients above.

%\clearpage
%\setcounter{tocdepth}{5}
%\tableofcontents

\end{document}